\documentclass[12pt]{article}

\usepackage{amsmath,amssymb,amsthm,mathtools}
\usepackage{bm}
\usepackage{geometry}
\usepackage{hyperref}
\usepackage[dvipsnames]{xcolor}
\definecolor{darkblue}{rgb}{0, 0, 0.5}
\hypersetup{
     colorlinks = true,
     citecolor = Maroon,
     urlcolor = darkblue,
     linkcolor = darkblue
}
\usepackage{setspace}
\usepackage{enumerate}
\usepackage{graphicx}
\usepackage{booktabs}
\usepackage{floatrow}
\usepackage[round]{natbib} 
\usepackage{comment}
\usepackage{algorithm}
\usepackage{algpseudocode}

\newcommand{\E}{\mathbb{E}}
\newcommand{\Cov}{\mathrm{Cov}}

\newcommand{\T}{\mathsf{T}} %

\DeclareMathOperator*{\argmin}{arg\,min}
\usepackage{xspace}
\newcommand{\est}{\ensuremath{\textnormal{IVFR}}\xspace}
\newcommand{\CLP}{CLP\xspace}
\newcommand{\MP}{MP\xspace}

\newcommand{\Prob}{\mathbb{P}} 
\newcommand{\R}{\mathbb{R}} 
\newcommand{\calP}{\mathcal{Y}} 
\newcommand{\calQ}{Q(\mathcal{Y})} 
\newcommand{\dd}{\mathrm{d}} 
\newcommand{\toP}{ \overset{p}{\to}}

\newtheorem{theorem}{Theorem}
	\newtheorem{lemma}{Lemma} 
	\newtheorem{proposition}{Proposition}

    \newtheorem{corollary}{Corollary}
	\newtheorem{assumption}{Assumption}
	\newtheorem{example}{Example}
	\newtheorem{remark}{Remark}
	\newtheorem*{example*}{Example}

\title{IV regression with distribution-valued outcomes\thanks{We are grateful to Tim Christensen and Xiaohong Chen for feedback. We also thank seminar participants at the University of Michigan Econometrics Workshop for helpful comments. W\"uthrich is also affiliated with CESifo. Van Dijcke gratefully acknowledges support from the Lawrence Klein Econometrics Fellowship and the Rackham Predoctoral Fellowship at the University of Michigan. This research was supported in part through computational resources and services provided by Advanced Research Computing at the University of Michigan, Ann Arbor.  We thank Grammarly, GPT 5.2--5.5, Claude Opus 4.6--4.7, and \url{coarse.ink} for language editing, proofreading, and coding assistance. The authors are solely responsible for all content.}}
\author{David Van Dijcke\thanks{Department of Economics, University of Michigan. Email: \url{dvdijcke@umich.edu}}\quad \quad Kaspar W\"uthrich\thanks{Department of Economics, University of Michigan. Email: \url{kasparwu@umich.edu}}} 
\date{\today}

\begin{document}

\maketitle

\begin{abstract}
We develop IV Fr\'echet regression (\est), an instrumental-variable (IV) method for settings where the outcome is an entire distribution. Framing the problem as an IV regression in 2-Wasserstein space, \est extends global Fr\'echet regression to the case with endogenous covariates. \est projects IV-weighted quantile curves onto the space of valid distributions and then recovers the corresponding regression coefficient functions. The projection provably reduces the estimation error in finite samples and guarantees valid fitted distributions. We show that the \est estimator converges weakly to a mean-zero Gaussian process and establish the validity of a multiplier bootstrap procedure for uniform inference. In simulations, the projection reduces the integrated mean squared error (IMSE) by up to 63\% relative to existing methods. Revisiting the effects of Chinese import competition on the wage distribution within commuting zones, the proposed method produces 9--10\% narrower confidence bands than existing methods. Using our novel uniform confidence bands, we find no evidence that import competition reduced wages at the very bottom of the distribution, but only between the 10th and 35th quantile. We also revisit the effect of county food stamp programs on the county's birth weight distribution and find no significant effects.



\end{abstract}

\newpage
\onehalfspacing

\section{Introduction} 

This paper introduces an instrumental variables (IV) framework for settings where the outcome of interest is a distribution. A leading example is when researchers are interested in the effect of a treatment that varies at the group level on the distribution of an outcome within groups. Analyzing effects on the entire distribution of an outcome within groups, rather than on simple group-level averages, provides a better understanding of the impact of a treatment and allows for studying important distributional issues, such as inequality and tail risk. For instance, building on \citet{autor2013china}, \citet{chetverikov2016iv} (\CLP henceforth) estimated the effect of import competition on the distribution of local wages and found that low-wage earners were most affected by increased import competition. In this application, the treatment (import competition) varies at the group (commuting zone (CZ)) level, and the outcome is the distribution within groups (the wage distribution within a CZ). Other examples include policies implemented at the firm, hospital, school, county, state, or country level that affect the entire distribution of employee, patient, student, or population outcomes (see also \citet{van2025regression}). In many applications, these group-level treatments are endogenous, which motivates the use of IVs for estimating causal effects. 

We propose a flexible IV regression framework allowing researchers to estimate the effects of real-valued treatments on distribution-valued outcomes, using real-valued instruments. We refer to the proposed procedure as IV Fr\'echet regression (\est), because it can be viewed as an extension of the Fr\'echet framework for regression on metric spaces of \citet{petersen2019frechet} to the IV setting. We will use the grouped data terminology throughout the paper, but emphasize \est is very general and can be applied to panel data as well as any other setting where the outcomes of interest are distribution-valued. 

We consider a linear structural model for the (known or estimated) group-level quantile function corresponding to the distribution-valued outcome of interest. We show that this structural quantile function is identified as the solution to a Fr\'echet IV regression problem in 2-Wasserstein space. This identification result suggests a computationally straightforward plug-in estimation strategy: (i) construct plug-in estimates of the IV weights, (ii) compute IV-weighted average quantile curves at each covariate value, (iii) project each curve onto the space of valid quantile functions---solving the sample Fr\'echet IV problem---and (iv) recover coefficient functions by OLS. We show that the resulting estimator converges weakly to a zero-mean Gaussian process and propose uniform inference procedures based on the multiplier bootstrap.

A key feature of \est is the projection step (iii), which has several desirable properties. First, it ensures that each estimated conditional distribution is a valid probability measure and that the estimator has the interpretation of an IV-weighted Wasserstein barycenter. Second, it provably decreases the error in estimating the structural quantile function in finite samples, which in turn leads to a finite-sample improvement in coefficient estimation. Finally, under a weak monotonicity condition on the population IV-weighted quantile functions, the projection step does not affect the asymptotic distribution, so that bootstrap inference remains valid.

We evaluate the finite-sample performance of \est in Monte Carlo simulations. The projection reduces IMSE by up to 63\% relative to existing methods under weak to moderate instruments, with the gain diminishing to below 1\% under stronger instruments. Our pointwise and uniform inference confidence bands exhibit good finite sample coverage.

We demonstrate the usefulness of \est in two empirical applications. First, we revisit the  distributional wage effects of Chinese import competition in \CLP. We find that \est produces 9--10\% narrower pointwise confidence bands than CLP. Using our new uniform confidence bands, however, we find no evidence that wages are reduced at the very bottom of the distribution, but only between the 10th and 35th quantile. Second, we revisit the estimation of the causal effect of food stamps on the county birth weight distribution in \citet[][\MP henceforth]{melly2025minimum} and find no significant effects across the distribution. 

\paragraph{Literature.} We contribute to several strands of the literature. Our first contribution is to the literature on distributional effects in panel
and grouped data settings
\citep[e.g.,][]{koenker2004quantile,canay2011simple,galvao2015efficient, galvao2016smoothed,chetverikov2016iv,galvao2020unbiased, chen2023group, gunsilius2023distributional,  pons2024quantile, torous2024optimal, chen2025quantile,melly2025minimum}. We show that when the object of interest is the effect of a group-level variable on a group-level distribution, the problem can be formulated naturally as a regression with distribution-valued outcomes. This perspective places existing grouped quantile IV estimators in a Fr\'echet regression framework: in the absence of individual-level covariates, the estimators of \CLP and \MP  \citep[building on][]{hausman1981panel} coincide with the unprojected version of our estimator. The Fr\'echet formulation shows that the estimand is an IV-weighted Wasserstein barycenter, yields fitted distributions that are valid by construction, and leads to a natural projection step that improves the finite-sample performance while leaving the first-order asymptotic distribution
unchanged. It also allows us to derive novel theoretical results: we establish the properties of \est under misspecification of the linear quantile model and propose a multiplier bootstrap for constructing uniform confidence bands. See Section  \ref{sec:relationship_existing_literature} for a formal discussion of the relationship of our model and \CLP and \MP and Section \ref{sec:simulation} for a simulation comparison.

Our second contribution is to the literature on Fr\'echet regression \citep{petersen2019frechet} and to the recent literature extending  Fr\'echet regression and related methods to causal inference settings \citep[e.g.,][]{ lin2023causal, katta2024interpretable, kurisu2024geodesic, hoshino2024functional, bhattacharjee2025doubly, van2025regression,zhou2025geodesic, kurisu2025regression, kurisu2026lee}. We contribute to this literature by providing a linear IV framework for estimating treatment effects of endogenous group-level treatments; as well as by establishing asymptotic normality and deriving uniform confidence bands for this novel estimator, complementing related results for global and local Wasserstein-Fr\'echet regression under exogeneity \citep{ petersen2021wasserstein,van2025regression, xu2025wasserstein, song2026inference}.

Third, we contribute to the literature on quantile models with endogeneity more broadly \citep[e.g.,][among many others]{abadie2002instrumental,chernozhukov2005iv,chernozhukov2006instrumental,chernozhukov2008instrumental,lee2007endogeneity,lee2007nonparametric,frolich2013unconditional,kaplan2017smoothed,vuong2017counterfactual,decastro2019smoothed,wuthrich2019closed,wuthrich2020comparison,kaido2021decentralization,beyhum2023instrumental, holovchak2025distributional}. The goal of this literature is to estimate the effects on the quantiles of real-valued scalar outcomes, whereas we consider the functional IV regression case in which the outcomes themselves are distribution-valued.

A last paper that does not fit neatly in these three groups is \citet{qu2024distributionally}, who use the Wasserstein space to study IV methods, but focus on distributional robustness for classical IV assumptions rather than distributional treatment effects.

\paragraph{Notation.}
We use the following notation throughout. For a vector $v \in \R^p$, $\|v\| = (v^{\!\T} v)^{1/2}$ denotes the Euclidean norm. For a positive semi-definite matrix $A$, $\|v\|_A = (v^{\!\T} A v)^{1/2}$ is the $A$-weighted norm. For functions $f \in L^2([0,1])$, $\|f\|_{L^2} = (\int_0^1 f(u)^2\, du)^{1/2}$ is the standard $L^2$ norm. The 2-Wasserstein distance between two distributions with quantile functions $Q_1, Q_2$ is $W_2(\mu,\nu) = \|Q_1 - Q_2\|_{L^2}$. For a measurable function \(f\), let \(P f:=\E[f(W)]\) and let
\(\mathbb P_n f:=n^{-1}\sum_{j=1}^n f(W_j)\) denote the empirical measure
applied to \(f\). We also write $\leadsto$ for weak convergence in the sense of \citet{vaart1996weak}, $\overset{p}{\to}$ for convergence in probability, $\hat{\mathbb{G}}_\beta^*\rightsquigarrow_{\Prob}\mathbb{G}_\beta$ for conditional weak convergence in probability \citep[Section~2.9]{vaart1996weak}: $\sup_{h\in BL_1}|E^x[h(\hat{\mathbb{G}}_\beta^*)]-E[h(\mathbb{G}_\beta)]|\toP 0$, where $BL_1$ is the set of functions $\ell^\infty([a,b])^{p+1}\to\R$ with Lipschitz constant and supremum norm both bounded by one. Finally, we write $\ell^\infty(T)$ for the space of bounded functions on $T$ equipped with the supremum norm. For a closed convex set $K$ in a Hilbert space, $\Pi_K$ denotes the metric projection onto $K$.

\section{Setup and model}

Consider a setting with $n$ groups indexed by $j=1, \dots, n$. We are interested in estimating the effect of group-level variables $X_j \in \R^p$ with support $\mathcal{X} \subseteq \R^p$ on a distribution-valued outcome $Y_j\in \mathcal{Y}$, where $\calP$ is the space of one-dimensional cumulative distribution functions (CDFs) with finite variance. Let $Q_{Y_j} \in \calQ$ denote the quantile function corresponding to the CDF $Y_j$, where $\calQ$ is the space of quantile functions corresponding to $\calP$. Note that for a fixed quantile level $u\in (0,1)$, $Q_{Y_j}(u)$ is a real-valued random variable. We also have access to a vector of group-level IVs, which includes an intercept, $Z_j \in \R^{l+1}$, with $l\ge p$. We assume throughout that $\{(X_j,Y_j,Z_j)\}_{j=1}^n$ are sampled i.i.d. We fix $0 < a < b < 1$ and state all asymptotic results for quantile levels $u \in [a,b]$, avoiding tail quantiles.

We consider a linear model for $Q_{Y_j}$, noting that our asymptotic results will allow for misspecification,
\begin{equation}\label{eq:linear_model}
  Q_{Y_j}(u) = \beta_0(u) + \beta_1(u)^{\!\T}(X_j-\mu_X) + \eta_j(u), \quad E[\eta_j(u)]=0, \quad u\in (0,1),
\end{equation}
where $\mu_X:=E[X_j]$, $\beta_0(u)$ is a scalar intercept function, $\beta_1(u)$ is a $p$-dimensional vector of slope functions, and $\eta_j(u)$ is an unobserved error term. Given that we include an intercept, the assumption that $E[\eta_j(u)]=0$ is a normalization. In what follows, we will often refer to $q(x,u):=\beta_0(u) + \beta_1(u)^{\!\T}(x-\mu_X)$ as the structural quantile function.

Writing the model in terms of demeaned variables $\tilde{X}_j=X_j-\mu_X$ will be convenient for our theoretical analysis.\footnote{When we use the demeaned coefficient and instrument vectors, we drop the intercept from those vectors to guarantee non-singularity of the covariance matrices.} Note that due to the linearity, we can always reparametrize the model and write it in terms of $X_j$ as
\begin{equation*}
  Q_{Y_j}(u) = \tilde\beta_0(u) + \beta_1(u)^{\!\T}X_j + \eta_j(u),
\end{equation*}
where $\tilde\beta_0(u)=\beta_0(u)-\beta_1(u)^{\!\T}\mu_X$. Finally, define $\mathbf{X}_j \coloneqq \left( 1, (X_j - \mu_X)^T\right)^T$.

\section{IV Fr\'echet regression}

In this section, we discuss identification and estimation in the \est model.

\subsection{Population problem and identification}
\label{sec:population_fivr_problem_identification}
We obtain identification of the structural quantile function $q(x,u)$ using the vector of IVs, $Z_j$. To do so, we impose the following standard assumptions. Let $\tilde{Z}_j:=Z_j-E[Z_j]$.

\begin{assumption}[Instrument exogeneity]\label{ass:instrument_exogeneity}
Under correct specification of the linear quantile model \eqref{eq:linear_model},
\[
E\!\left[\tilde Z_j \eta_j(u)\right]=0
\qquad\text{for all } u\in(0,1).
\]
\end{assumption}

This is the standard IV orthogonality condition for the structural error in
\eqref{eq:linear_model}. Under misspecification, we define the pseudo-true
coefficient functions by the population 2SLS projection given below; in general,
the corresponding pseudo-residual need not satisfy $E[\tilde Z_j \xi_j(u)]=0$
in the overidentified case.

Define the population matrices $\Sigma_{ZZ} := E[\tilde Z_j\tilde Z_j^{\T}]$, $\Sigma_{ZX} := E[\tilde Z_j \tilde{X}_j^{\T}]$, and $\Sigma_{XX} := E[\tilde{X}_j \tilde{X}_j^{\!\T}]$.
\begin{assumption}[Full rank] \label{ass:full_rank}
$\Sigma_{ZZ}$ is positive definite and $\Sigma_{ZX}$ has full column rank $p$.
\end{assumption}

To motivate \est, consider first the case where $X_j$ is exogenous, $E[\tilde{X}_j\eta_j(u)]=0$. In this case, we can use conventional global Fr\'echet regression \citep{petersen2019frechet} to obtain the structural quantile function $q(x,u)$. \citeauthor{petersen2019frechet} motivate global Fr\'echet regression as a generalization of standard linear regression. Specifically, suppose that $Y_j\in \mathbb{R}$ and the conditional expectation function is linear, $m(x):=E\left[Y_j\mid X_j=x\right]=\alpha_0+\alpha_1^{\!\T}(x-\mu_X)$. Then, the formal characterization of the conditional expectation is $m(x)=\argmin_{m\in \mathbb{R}}E\left[s(X_j,x)d_E^2(Y_j,m)\right]$, where $d_E(\cdot)$ is the standard Euclidean norm and $s(z,x)=1+(z-\mu_X)^{\!\T} \Sigma_{XX}^{-1}(x-\mu_X)$ are the linear regression weights. The idea of global Fr\'echet regression is to replace the Euclidean norm $d_E$ with a more general norm $d$ suitable for outcomes taking values in general metric spaces.

Here, we extend Fr\'echet regression to settings where $X_j$ is endogenous and the outcomes are distribution-valued. To motivate our approach, suppose first that $Y_j\in \mathbb{R}$ satisfies the linear IV model $Y_j=\alpha_0 + \alpha_1^{\!\T}(X_j-\mu_X) + \eta_j$ with $E[\eta_j]=0$. A key observation underlying \est is that the linear model $m(X_j):=\alpha_0 + \alpha_1^{\!\T}(X_j-\mu_X)$ can be obtained by rewriting the canonical two-stage least-squares (2SLS) estimator as
$$
m(x)=\argmin_{m \in \mathbb{R}} E\left[s(Z_j,x)d_E^2(Y_j,m) \right],
$$
where
\begin{equation}\label{eq:s-weight-fivr_intercept_main}
   s(Z_j, x) := 1 + (x-\mu_X)^{\!\T} (\Sigma_{ZX}^{\!\T}\Sigma_{ZZ}^{-1}\Sigma_{ZX})^{-1} \Sigma_{ZX}^{\!\T}\Sigma_{ZZ}^{-1} \tilde Z_j.
\end{equation}
See the proof of Lemma~\ref{lem:identification_FIVR} in Appendix \ref{sec:proofs_appendix} for a derivation.

Now, replace the Euclidean norm with the 2-Wasserstein distance $W_2$ suitable for distribution-valued outcomes, which for two one-dimensional distributions $Y_1, Y_2 \in \mathcal{Y}$ is defined as,
\[
W_2(Y_1, Y_2) \coloneqq \left( \int_{0}^1 \left( Q_{Y_1}(u)  - Q_{Y_2}(u) \right)^2 \dd u \right)^{\frac12}.
\]

Then, the above characterization of IV-Fr\'echet regression leads to the following instrumental-variables version of Fr\'echet regression,
\begin{equation}\label{eq:frechet-objective-fivr_intercept_pop}
  m^{\est}(x) = \argmin_{m\in\calP} E\left[s(Z_j,x)\,W_2^{2}(Y_j,m)\right],
\end{equation}
where, by construction, for each given $x$, $m^{\est}(x)$ is a distribution function, i.e., it lies in $\calP$.

It follows from the following Proposition that the quantile function associated with $m^{\est}(x)$, denoted $Q_{m^{\est}(x)}$, is the $L^2$-projection of the IV-weighted quantile function, $\psi_x(u) \coloneqq E[s(Z_j,x)Q_{Y_j}(u)]$, onto the space of quantile functions $\calQ$:
\begin{equation}\label{eq:projection-fivr_intercept_main}
   Q_{m^{\est}(x)}(u) := \Pi_{\mathcal{Q}}\left(E[s(Z_j,x)Q_{Y_j}(u)]\right).
\end{equation}
\begin{proposition}[Fr\'echet characterization of the IV-weighted quantile function]\label{prop:frechet_equiv}
Suppose Assumption~\ref{ass:full_rank} holds and the IV-weighted Fr\'echet
functional is finite, i.e.,
\[
E\!\left[|s(Z,x)|\,\|Q_Y\|_{L^2(0,1)}^2\right]<\infty.
\]
Then the minimizer of the IV-weighted Fr\'echet functional over distributions is
the $L^2$-projection of the IV-weighted mean quantile curve
$\psi_x(u)=E[s(Z,x)Q_Y(u)]$ onto the quantile cone,
\[
Q\left(\argmin_{w\in\calP}E\!\left[s(Z,x)W_2^2(Y,w)\right]\right)
=
\Pi_{\mathcal Q}(\psi_x),
\]
where $Q(\cdot)$ maps CDFs to their corresponding quantile function.
In particular, if $\psi_x(\cdot)$ is itself a valid quantile function, then
$\psi_x$ is the a.e.-unique minimizer.
\end{proposition}

By the linearity of $s(Z_j,x)$ in $x$, the function $\psi_x(u)$ is affine in $x$ 
and can be written as $\psi_x(u) = \beta_0^{\text{unc}}(u) + 
\beta_1^{\text{unc}}(u)^{\!\T}(x - \mu_X)$, where $\beta^{\text{unc}}(u)$ are the standard 2SLS coefficients applied quantile by quantile. Also, define the pseudo-true residual as,
\[
\xi_j(u):=Q_{Y_j}(u)-\mathbf X_j^{\!\T}\beta^{\mathrm{unc}}(u).
\]
By construction, $\beta^{\mathrm{unc}}(u)$ satisfies the population 2SLS normal
equations. In the overidentified case, however, this does not generally imply
$E[\tilde Z_j\xi_j(u)]=0$. A different GMM weighting matrix would in general
define a different pseudo-true coefficient function. Throughout the paper, we focus on
the 2SLS choice, but other linear GMM estimators could be used as well.

In population and under correct specification, the projection $\Pi_{\mathcal{Q}}$ is inactive and hence $Q_{m^{\est}(x)}(u)$ coincides with the solution to the linear 2SLS regression quantile by quantile, i.e., $E[s(Z_j,x)Q_{Y_j}(u)]$. Moreover, the functional object $m^{\est}(x)$ has the interpretation of being the (signed) IV-weighted Wasserstein barycenter of the group-level distributions $Y_j$. In other words, for a given $x$, $m^{\est}(x)$ is the IV-weighted average of each group's distribution in probability space---which implies it can rightfully be called the ``average'' instrumented distribution. For more discussion on (conditional) Wasserstein barycenters, we refer to \citet{agueh2011barycenters}, \citet{fan2024conditional}, and  \citet{panaretos2020invitation}.

The \est coefficients are then equal to,
\begin{align}
& \beta^\est_0(u) = E[Q_{m^\est(X)}(u)] \\
& \beta^\est_1(u) = \Sigma_{XX}^{-1} E[ \tilde{X}_j Q_{m^\est(X)}(u)],
\end{align}
the OLS coefficients corresponding to the projected quantile function $Q_{m^\est}(x)(u)$. 

The following lemma shows that, under correct specification, $Q_{m^\est(x)}(u)$ and $(\beta^\est_0(u),$ $ \beta^\est_1(u))$ recover the structural quantile function $q(x,u)$ and coefficients $(\beta_0(u),\beta_1(u))$, respectively. 

\begin{lemma}[Identification \est]\label{lem:identification_FIVR}
Suppose that Assumptions \ref{ass:instrument_exogeneity} and \ref{ass:full_rank} 
as well as the linear quantile model in Eq. \eqref{eq:linear_model} hold. Then 
for all $x \in \mathcal{X}$ and $u \in (0,1)$,
\[
Q_{m^{\est}(x)}(u) = \psi_x(u) = q(x,u),
\]
and $\beta^{\est}(u) = \beta^{\text{unc}}(u) = \beta(u)$.
\end{lemma}

Under correct specification, the projection is inactive and
$\psi_x(\cdot)=q(x,\cdot)$ is a valid quantile function for all $x\in\mathcal X$.
Hence $\Pi_{\mathcal Q}(\psi_x)=\psi_x$ and the Fr\'echet IV problem
\eqref{eq:frechet-objective-fivr_intercept_pop} recovers the structural quantile
function exactly.

Under misspecification, $\psi_x(\cdot)$ may violate monotonicity for some 
$x \in \mathcal{X}$, so $\Pi_{\mathcal{Q}}(\psi_x)$ may differ from $\psi_x$ 
and consequently $\beta^{\est}(u)$ may differ from $\beta^{\text{unc}}(u)$. The 
\est coefficients remain well-defined and interpretable: for each $x$, the 
Fr\'echet IV solution $\Pi_{\mathcal{Q}}(\psi_x)$ is the closest valid probability 
distribution to $\psi_x$ in $W_2$-distance, and $\beta^{\est}(u)$ parametrizes the best linear approximation to these projected distributions.

\subsection{What does \est identify?}
\label{sec:identification_interpretation}

Here, we discuss the interpretation of the coefficient $\beta_1(u)$. The formal results underlying this discussion are presented in Appendix \ref{app:individual_decomposition}. The coefficient $\beta_1(u)$ in the group-level quantile regression identifies the \emph{causal effect} of group-level treatments on the group outcome distribution. When treatment is assigned at the group level, this group-level effect is the natural causal parameter. It captures the total response of the group’s distribution to the treatment, incorporating all channels through which the treatment operates. Note that this causal interpretation does not require any rank invariance assumption, as it operates at the group and not at the individual level. 

The group-level model \eqref{eq:linear_model} is agnostic about the within-group structure.  The group quantile function $Q_{Y_j}(u)$ is the primitive; no assumptions are made about what generates it within the group. In particular, the within-group distribution can arise from an arbitrary mixture of individual-level outcomes, and the composition of the group may itself respond to treatment. The parameter $\beta_1(u)$ captures all these channels.

To illustrate, suppose that the treatment of interest is binary. Using potential outcomes notation, denote the counterfactual group-level CDFs without and with treatment as $Y(0)$ and $ Y(1)$, respectively. Then, under exogeneity and assuming the linear model holds, the causal parameter \est targets is the \textit{total causal effect},
\[
\beta_1(u) = E[Q_{Y(1)}(u) - Q_{Y(0)}(u)],
\]
the average of the group-specific quantile treatment effects $Q_{Y(1)}(u) - Q_{Y(0)}(u)$. This object has a direct counterfactual interpretation: it is the difference between the group quantile functions with and without the treatment, averaging over group-specific responses. Thus, the parameter $\beta_1(u)$ is a natural analog of the average treatment effect in the standard setting where the outcome is scalar-valued.

This is different from the \emph{direct effect} denoted by $\delta(u)$, which we show in Appendix \ref{app:individual_decomposition} is the effect identified by approaches that control for individual-level covariates within groups (such as \CLP and \MP with individual controls). The direct effect 
$\delta(u)$ holds the group composition (e.g., the share of low-skilled and high-skilled workers) constant and measures only the within-type response (e.g., the wage response of low-skilled vs.\ high-skilled workers). By contrast, the total effect $\beta(u)$ captures both the within-type response and the compositional response, capturing changes in who is in the group. For the direct effect to identify a \textit{causal} effect on individuals, a rank invariance assumption is needed. 

The total causal effect is a policy-relevant parameter in many applications. If a policymaker evaluates the impact of import competition on the local wage distribution, they typically care about what actually happens to wages in the affected regions, not a hypothetical scenario where wages change but workers cannot move. Indeed, the latter is often not a feasible policy. The total effect captures the equilibrium displacement of the distribution, and correctly characterizes the counterfactual effect on the treated unit, i.e., the entire group. 

That said, the direct effect $\delta(u)$ can be informative when the goal is to isolate the within-type response, for instance in decomposition exercises that aim to separate wage structure changes from compositional shifts \citep[e.g.,][]{chernozhukov2013inference}. The distinction between $\beta_1(u)$ and $\delta(u)$ is related to the distinction between unconditional and conditional quantile effects in the distributional treatment effects literature \citep{firpo2009unconditional, frolich2013unconditional}; see also Remark 1 in \MP for a related discussion.

\begin{example}[Total vs.\ direct effect: import competition]
\label{ex:total_vs_direct}
Consider a commuting zone exposed to an import competition shock, as in \CLP. The shock is assigned at the commuting-zone level, but workers within the zone  differ in their exposure. Non-college workers are more likely to be employed in  manufacturing and therefore experience larger wage losses, while college workers  are relatively insulated.

An estimand of the direct effect $\delta(u)$, which controls for individual education, targets a direct,  composition-fixed effect: how wages change within education groups, holding the  mix of workers fixed. This is useful for isolating an average over within-type wage responses, but it is not the same as the effect of the shock on the commuting zone's actual wage distribution. The latter also reflects how the shock changes who remains in the zone, which sectors shrink, and where different worker types fall in the post-shock wage distribution.

This distinction arises because quantiles of a mixture are not mixtures of quantiles. Even if the wage loss at every given quantile were identical across different education groups, the direct effect $\delta(u)$ would vary with $u$ because the composition of workers at each group quantile $u$ varies. If the shock also induces  selective out-migration of non-college workers, the composition-fixed effect and the effect on the observed CZ wage distribution can diverge further.

The \est estimand is the latter object: the effect of import competition on the commuting zone's distribution as a group-level outcome. This is the relevant 
counterfactual for a policymaker asking what happened to the local wage distribution in the affected region. A composition-fixed estimand answers a different question: what  would have happened to within-type wages if the worker composition had been held 
fixed. Both objects are useful, but they correspond to different counterfactuals.
\end{example}
\subsubsection{Relationship to existing estimators}
\label{sec:relationship_existing_literature}

We now compare $\est$ to the estimators of \CLP and \MP. There are two cases.

\paragraph{Without individual-level covariates.} When \CLP and \MP do not control for individual characteristics, their estimand coincides with ours: both identify the total effect $\beta(u)$. In this case, working directly with random distributions at the group level, as the proposed \est estimator does, has several advantages. First, it allows us to avoid imposing any restrictions at the individual level, requiring only that the group quantile functions are linear in group-level covariates, while the individual-level structure can be fully nonparametric. Moreover, our framework also explicitly allows for misspecification of this group-level linear structure. By contrast, the statistical results in \CLP and \MP depend on the linear structure at the individual as well as group level. Second, our Fr\'echet regression formulation clarifies that estimation can proceed in a single step: regressing the group-level quantile functions on the treatment. Third, by projecting the quantile functions and recovering coefficients by OLS, our estimator can exploit the functional nature of the data to improve precision, as formally shown in Theorem~\ref{thm:pw_ols_improvement} below. Fourth, our estimator for $\beta(u)$ is guaranteed to produce valid quantile functions at every observed covariate value, while \CLP and \MP are not.\footnote{One could, alternatively, achieve the last two points by using monotone rearrangement instead of projection \citep{chernozhukov2010quantile}. However, unlike projection, rearrangement does not naturally arise in the Fr\'echet regression framework, and one would lose the (IV-weighted) barycenter interpretation.}

Indeed, the simulations of Section~\ref{sec:simulation} show projection gains of up to 63\% in IMSE, driven by finite-sample non-monotonicity in the fitted quantile functions. In our replication of \CLP's empirical application in Section~\ref{sec:empirical}, the fitted conditional wage quantile functions implied by unprojected \est violate monotonicity in around 2\% of CZ--decade cells, so the projected point estimates are close to the unprojected ones. However, the projected bootstrap still delivers meaningfully tighter confidence bands. 

\paragraph{With individual-level covariates.} When \CLP and \MP control for individual characteristics $W_{ij}$, they identify the direct (within-type) effect $\delta(u)$. As discussed above, these are different estimands, and which estimand is of interest depends on the application and the question being asked. 

Finally, we emphasize that, in the presence of individual-level covariates, the instrument exogeneity assumption in \CLP and \MP is \textit{not} weaker than the analogous assumption in \est without individual-level covariates. When the instrument used to identify the group-level treatment varies
at the group level, our condition $\E[\tilde Z_j\eta_j(u)] = 0$ is equivalent to the orthogonality condition between the instrument and the
group-level unobservable imposed by \CLP and \MP. Individual-level covariates change the estimand from a total group-level effect to a conditional/direct effect but do not weaken the required exclusion restriction for the group-level
IV variation. This differs from the standard logic that conditioning on covariates can make the instrument exogeneity assumption more credible. See Appendix~\ref{sec:iv_decomposition} for a formal derivation. Note that, due to its functional nature, \est does not support the inclusion of individual-level covariates. Formulating a version of distribution-valued regression with individual-level covariates is an interesting avenue for future research, see also the distribution-on-distribution regression approaches of \citet{oliva2013distribution} and \citet{ghodrati2022distribution}.

\subsection{The \est estimator}\label{sec:div_estimator}

The \est estimator solves the sample analogue of the population Fr\'echet IV problem \eqref{eq:frechet-objective-fivr_intercept_pop} at each observed covariate value, and recovers the linear parametrization by OLS:
\begin{equation}\label{eq:div_estimator}
\hat{\beta}^{\est}(u) = (\hat{\mathbf{X}}^{\!\T}\hat{\mathbf{X}})^{-1}\hat{\mathbf{X}}^{\!\T} \Pi_{\mathcal{Q}}(\hat{\psi}_{X_j})(u),
\end{equation}
where $\hat{\mathbf{X}}$ is the $n \times (p+1)$ matrix with rows $(1, (X_j - \hat{\mu}_X)^{\!\T})$, and $\hat{\psi}_x(u) := \frac{1}{n}\sum_{i=1}^n \hat{s}_i(Z_i, x)\, Q_{Y_i}(u)$ is the IV-weighted average quantile curve at $x$, with plug-in IV weights
$$
\hat{s}_{j}(Z_j,x) := 1 + (x-\hat{\mu}_X)^{\!\T} (\hat{\Sigma}_{ZX}^{\!\T}\hat{\Sigma}_{ZZ}^{-1}\hat{\Sigma}_{ZX})^{-1} \hat{\Sigma}_{ZX}^{\!\T}\hat{\Sigma}_{ZZ}^{-1} (Z_j-\hat{\mu}_Z).
$$

Since $\hat{\psi}_x(u)$ is linear in $x$, it defines the unprojected coefficient functions \[\tilde{\beta}(u) = (\tilde{\beta}_0(u), \tilde{\beta}_1(u)^{\!\T})^{\!\T}\] with $\hat{\psi}_x(u) = \tilde{\beta}_0(u) + \tilde{\beta}_1(u)^{\!\T}(x - \hat{\mu}_X)$. As above, because the IV weights can be negative, $\hat{\psi}_{X_j}(\cdot)$ may not be monotone. The projection $\Pi_{\mathcal{Q}}$ maps each IV-weighted average onto the space of valid quantile functions, yielding $\hat{Q}(X_j, u) := \Pi_{\mathcal{Q}}(\hat{\psi}_{X_j})(u)$---the closest valid probability distribution (in $W_2$) to the IV-weighted average at covariate value $X_j$. This is the sample IV-Fr\'echet barycenter, the sample analogue of the population projection \eqref{eq:projection-fivr_intercept_main}. It can be computed as an isotonic regression by the pool-adjacent-violators algorithm (PAVA) \citep{ayer1955empirical, miles1959complete, kruskal1964nonmetric}. If only samples from the distribution $Y_j$ are available, we replace $Q_{Y_j}(u)$ by the empirical quantile function $\widehat{Q}_{Y_j}(u)$. We show in Section \ref{sec:empirical_quantiles} below that under a weak condition on the relationship between within- and across-group sample sizes, this additional first-stage estimation does not affect our asymptotic results.

The projection guarantees that each $\hat{Q}(X_j, \cdot)$ is a valid quantile function and is closer to the true structural quantile function $q(X_j, \cdot)$ than the unprojected $\hat{\psi}_{X_j}$ in all $L^p$ norms (Lemma~\ref{lemma:improvement} below). The OLS step in \eqref{eq:div_estimator} then recovers the linear parametrization from the projected quantile functions. 

\section{Theoretical properties of \est}

\subsection{Finite-sample properties}

In this section, we study the finite sample properties of the \est projection step and show that it decreases the estimation error. The first lemma shows that the projection step decreases the estimation error in the quantile function. This is a classical property of isotonic projections \citep{Robertson1988} and is reproduced from \citet[Lemma 2]{van2025regression}.

\begin{lemma}[Improvement in estimation: quantile function]\label{lemma:improvement}
Suppose that $\hat{Q}$ is an estimator of some true quantile function $Q_0 \in \mathcal{Q}$. Then $\hat{Q}^* = \Pi_{\mathcal{Q}}(\hat{Q})$ satisfies $\| \hat{Q}^*-Q_0\|_{L^p} \leq\|\hat{Q}-Q_0\|_{L^p}$ for all $p \in[1, \infty]$.
\end{lemma}

Applied to \est, this result implies that at each observed $X_j$, the projected $\hat{Q}(X_j, \cdot) = \Pi_{\mathcal{Q}}(\hat{\psi}_{X_j})$ is closer to any valid quantile function than the unprojected $\hat{\psi}_{X_j}$, in every $L^p$ norm. 

Next, we show that this improvement translates into finite sample improvements for the proposed estimator of the coefficient functions, which are often the main objects of interest. Throughout, let $b=(b_0,b_1)$ with $b_0\colon[0,1]\to\R$ and $b_1\colon[0,1]\to\R^p$ denote \emph{reference coefficients} such that $q_b(X_j,\cdot):=b_0(\cdot)+b_1(\cdot)^{\!\T}(X_j-\hat\mu_X)\in\mathcal Q$ for all $j=1,\ldots,n$. The results below hold for any such $b$, and no assumption on the linear model~\eqref{eq:linear_model} is needed. In practice, $b$ is set equal to the estimation target. Under correct specification of~\eqref{eq:linear_model}, setting $b_0(u)=\beta_0(u)+(\hat{\mu}_X-\mu_X)^{\!\T}\beta_1(u)$ and $b_1=\beta_1$ gives $q_b(X_j,u)=q(X_j,u)\in\mathcal Q$, so the bounds measure the estimation error for the true coefficients $(\beta_0,\beta_1)$. Under misspecification, a natural choice is the pseudo-true parameter $b = \beta^{\mathrm{unc}}$: under Assumption~\ref{asspt:average_strictness} below, it produces strictly increasing quantile functions for all $x\in\mathcal X$ with slope at least $\kappa>0$, so $q_b(X_j,\cdot)\in\mathcal Q$ for all $j$ with probability approaching one as $n\to\infty$. In either case, the projection brings the \est coefficients closer to the target than their unprojected counterparts.
\begin{theorem}[Improvement in estimation: joint coefficients]\label{thm:pw_ols_improvement}
Let $b=(b_0,b_1)$ be any target coefficients such that $q_b(X_j,\cdot)\in\mathcal Q$ for all $j=1,\ldots,n$. Then for any realization of the sample,
\begin{equation}\label{eq:pw_improvement}
\|\hat{\beta}_0^{\est} - b_0\|^2_{L^2} + \|\hat{\beta}_1^{\est} - b_1\|^2_{\hat{\Sigma}_{XX}, L^2} \;\leq\; \|\tilde{\beta}_0 - b_0\|^2_{L^2} + \|\tilde{\beta}_1 - b_1\|^2_{\hat{\Sigma}_{XX}, L^2},
\end{equation}
where $\hat{\Sigma}_{XX} = \frac{1}{n}\sum_{j=1}^n (X_j - \hat{\mu}_X)(X_j - \hat{\mu}_X)^{\!\T}$ and, for a vector-valued function $f\colon [0,1] \to \R^p$,
\[
\|f\|_{\hat{\Sigma}_{XX}, L^2}^2 := \int_0^1 f(u)^{\!\T}\, \hat{\Sigma}_{XX}\, f(u)\, \dd u.
\]
\end{theorem}

We emphasize that Theorem \ref{thm:pw_ols_improvement} does not assume the linear model \eqref{eq:linear_model}. The bound in Theorem \ref{thm:pw_ols_improvement} is a weighted joint improvement for estimating $(b_0, b_1)$. It does not guarantee an improvement for estimating any $b_{1,k}$ separately. The following results provide conditions under which the slope coefficients also improve individually.

To isolate individual slope coefficients, we use the Frisch--Waugh--Lovell (FWL) decomposition and the variational characterization of the isotonic projection. Let $C$ be the $n\times p$ matrix with rows $(X_j-\hat\mu_X)^{\!\T}$, write $c_k$ for its $k$-th column and $C_{-k}$ for $C$ with column $k$ removed. Define the FWL residual $r_k:=M_{-k}c_k$ with $M_{-k}:=I_n-C_{-k}(C_{-k}^{\!\T} C_{-k})^{-1}C_{-k}^{\!\T}$, let $r_{jk}$ denote its $j$-th entry, and set $J_k:=\{j\le n:r_{jk}\neq 0\}$ and $\hat v_k:=\frac{1}{n}\sum_{j=1}^n r_{jk}^2>0$. Let $\pi_k$ denote the coefficient vector from the sample linear projection of $c_k$ on $C_{-k}$, so that $c_k=C_{-k}\pi_k+r_k$. For a given target $b$, define the \emph{nuisance estimation error}
\begin{equation}\label{eq:ejk}
e_{jk}(u)
\;:=\;
\underbrace{(\tilde\beta_0-b_0)(u)}_{\text{intercept error}}
\;+\;
\underbrace{c_{j,-k}^{\!\T}
\bigl[
(\tilde\beta_{1,-k}-b_{1,-k})(u)+\pi_k(\tilde\beta_{1,k}(u)-b_{1,k}(u))
\bigr]}_{\text{nuisance slope errors}},
\end{equation}
which collects all estimation errors except the error in the $k$-th slope along the residualized variation $r_{jk}$. This decomposition is justified by the algebraic identity,
\begin{equation}\label{eq:decomp_identity}
\hat\psi_{X_j}(u)
\;=\;
q_b(X_j,u)\;+\;e_{jk}(u)\;+\;r_{jk}\bigl(\tilde\beta_{1,k}(u)-b_{1,k}(u)\bigr),
\end{equation}
which separates the unprojected fitted curve $\hat\psi_{X_j}$ into three components: the target quantile function $q_b(X_j,\cdot)$, the nuisance error $e_{jk}$, and the $k$-th slope error along the identifying variation $r_{jk}$.

\begin{proposition}[Improvement in estimation: isolated coefficients]\label{prop:coord_improvement}
Let $b=(b_0,b_1)$ be any reference coefficients such that $q_b(X_j,\cdot)\in\mathcal Q$ for all $j\in J_k$. Fix $k\in\{1,\ldots,p\}$. Then for any realization of the sample,
\begin{equation}\label{eq:coord_bound_general}
\|\hat\beta_{1,k}^{\est}-b_{1,k}\|_{L^2}^2
\;\le\;
\|\tilde\beta_{1,k}-b_{1,k}\|_{L^2}^2
\;-\;\frac{1}{n\hat v_k}\sum_{j\in J_k}\|D_{X_j}\|_{L^2}^2
\;-\;\frac{2}{n\hat v_k}\sum_{j\in J_k}\langle D_{X_j},\, e_{jk}\rangle_{L^2},
\end{equation}
where $D_{X_j}(u):=\hat{Q}(X_j,u)-\hat\psi_{X_j}(u)$ is the projection correction.
\end{proposition}

The bound~\eqref{eq:coord_bound_general} decomposes the change in $k$-th coefficient error into a \emph{projection gain} $\frac{1}{n\hat v_k}\sum_{j\in J_k}\|D_{X_j}\|^2_{L^2}$, which is always non-negative, and a \emph{nuisance cross-term} $\frac{2}{n\hat v_k}\sum_{j\in J_k}\langle D_{X_j}, e_{jk}\rangle_{L^2}$, which captures the interaction between projection corrections and nuisance estimation errors and can be negative or positive. An improvement obtains whenever the nuisance cross-term is positive, or when it is negative but does not overwhelm the projection gain. The following corollary gives a sufficient condition that eliminates the nuisance cross-term entirely, thus guaranteeing an improvement.

\begin{corollary}[Sufficient condition for coefficient-wise improvement]
\label{cor:coord_sufficient}
Under the setup of Proposition~\ref{prop:coord_improvement}, if the nuisance errors do not make the target quantile functions non-monotonic, i.e.,
\begin{equation}\label{eq:Ak_feas}
q_b(X_j,\cdot)+e_{jk}(\cdot)\in\mathcal Q
\qquad\text{for every }j\in J_k,
\end{equation}
then the nuisance cross-term in~\eqref{eq:coord_bound_general} is eliminated and
\begin{equation}\label{eq:coord_bound_clean}
\|\hat\beta_{1,k}^{\est}-b_{1,k}\|_{L^2}^2
\;\le\;
\|\tilde\beta_{1,k}-b_{1,k}\|_{L^2}^2
\;-\;\frac{1}{n\hat v_k}\sum_{j\in J_k}\|D_{X_j}\|_{L^2}^2
\;\le\;
\|\tilde\beta_{1,k}-b_{1,k}\|_{L^2}^2.
\end{equation}
\end{corollary}

The decomposition~\eqref{eq:decomp_identity} clarifies what condition~\eqref{eq:Ak_feas} requires: the error in the $k$-th slope along the identifying variation $r_{jk}$ must be \emph{necessary} for any monotonicity violations in $\hat\psi_{X_j}$. The nuisance errors $e_{jk}$ may erode the monotonicity margin of $q_b(X_j,\cdot)$, but as long as they do not destroy it entirely, every projection correction is fixing a violation attributable to the $\beta_{1,k}$ error. Put differently,~\eqref{eq:Ak_feas} requires $q_b(X_j,\cdot)$ to increase steeply enough to absorb whatever local decreases $e_{jk}$ introduces. Since every term in $e_{jk}$ is $O_p(n^{-1/2})$ uniformly in $u$ and $j$---the latter by Assumption~\ref{asspt:bounded_support}, which bounds $\|c_{j,-k}\|$---the condition holds asymptotically whenever the target quantile functions have slopes bounded away from zero. This is the same requirement needed for uniform convergence of the estimators, formalized in Assumption~\ref{asspt:average_strictness} below. When $p=1$, there is no partialling out ($r_{j1}=X_j-\hat\mu_X$), the nuisance slope term in $e_{j1}$ vanishes, and~\eqref{eq:Ak_feas} reduces to $q_b(X_j,\cdot)\in\mathcal Q$: the assumed monotonicity of the target.

Theorem~\ref{thm:pw_ols_improvement} and Proposition~\ref{prop:coord_improvement} sit at two ends of a spectrum. Theorem~\ref{thm:pw_ols_improvement} requires no additional assumptions beyond valid reference coefficients, but guarantees improvement only for all coefficients jointly. Proposition~\ref{prop:coord_improvement} isolates improvements coefficient by coefficient, at the cost of requiring condition~\eqref{eq:Ak_feas} on which errors drive the monotonicity violations. A natural intermediate approach is to pick a subset $S$ of coefficients for which~\eqref{eq:Ak_feas} is expected to hold. The same argument then yields an improvement result for that subset, with the nuisance error collecting only estimation errors from the intercept and the slopes outside~$S$.

\subsection{Asymptotic distribution}\label{sec:asymptotic_distribution}

In this section, we establish functional central limit theorems for the \est
estimator under weak conditions. We do not require the linear model
\eqref{eq:linear_model} to hold in population. When the full group-level quantile
functions are observed for each group, the unprojected results allow these quantile functions to be discrete. For the projected results, we impose below an average smoothness condition that rules out synchronized jumps at fixed quantile indices while still allowing individual group distributions to have atoms and flat parts. This contrasts with existing inferential results for global Wasserstein-Fr\'echet regression \citep{petersen2021wasserstein} and grouped quantile methods (e.g., \CLP; \MP), which typically impose density or pathwise smoothness conditions. Throughout this section, we work on $[a,b] \subset (0,1)$, and the projection $\Pi_{\mathcal{Q}}$ should hence be understood with respect to the relevant cone of monotone functions in $L^2([a,b])$.

\begin{assumption}[Sampling] \label{asspt:sampling}
$\{(Z_j, X_j, Y_j)\}_{j=1}^n$ is sampled i.i.d.\ from the joint distribution $F$ on $\R^l \times \R^p \times \calP$, where $\calP$ is the space of CDFs with finite second moments.
\end{assumption}

\begin{assumption}[Finite 4th moments] \label{asspt:finite_moments}
$\sup_{u \in [a,b]} E[|Q_{Y_j}(u)|^{4}] < \infty $, $E[\|Z_j\|^{4}] < \infty$, and $E[\|X_j\|^{4}] < \infty$.
\end{assumption}

Assumption~\ref{asspt:finite_moments} corresponds to the standard 4th moment condition in linear IV regression \citep[e.g.,][]{hansen2022econometrics}. 

To establish the weak convergence of the fully projected \est coefficients, we will additionally impose a boundedness restriction on the covariates.
\begin{assumption}[Bounded support] \label{asspt:bounded_support}
The support $\mathcal{X}$ of $X$ is bounded: $\sup_{x \in \mathcal{X}} \|x - \mu_X\| \leq B$ for some $B < \infty$.
\end{assumption}

Assumption \ref{asspt:sampling} requires i.i.d. sampling \textit{across} groups, but importantly leaves the within-group sampling scheme fully unrestricted. Additionally, Assumptions \ref{asspt:finite_moments} and \ref{asspt:bounded_support} are implied by the assumptions of \CLP and \MP, who assume bounded support of both $X$ and $Z$.

The next assumption requires that $\psi_x(\cdot)$ is uniformly strictly increasing on $[a,b]$ with slope $\geq \kappa$, for all $x$ in the support.

\begin{assumption}[Average strictness] \label{asspt:average_strictness}
There exists $\kappa>0$ such that for all $c<d$ with $[c,d]\subset [a,b]$ and all $x \in \mathcal{X}$,
\[
\E\!\Big[s(Z,x)\,\big\{Q_Y(d)-Q_Y(c)\big\}\Big]\;\ge\;\kappa\,(d-c).
\]
\end{assumption} 
Assumption \ref{asspt:average_strictness} is weaker than requiring every group to have a density bounded away from zero (as, e.g., in \CLP). In particular, it allows for flat parts and atoms in individual group quantile functions. The assumption only requires that, on every quantile interval, a non-negligible IV-weighted share of groups contributes a positive increment. Technically, Assumption~\ref{asspt:average_strictness} gives \(\psi_x\) a uniform monotonicity margin on \([a,b]\), so the projection is asymptotically inactive and the functional delta method applies. Note that Assumption~\ref{asspt:average_strictness} does not require smoothness. It is implied by a fixed positive IV-weighted fraction of groups locally ``moving'' (heterogeneous supports/atoms). See Remark ~\ref{rem:fang_santos} below for the case where Assumption~\ref{asspt:average_strictness} fails.

For the projected asymptotic results, we also require a mild average smoothness
condition on the cross-group average quantile process.

\begin{assumption}[Average Lipschitz quantile increments]\label{asspt:average_lipschitz}
There exists $K<\infty$ such that for all $c<d$ with $[c,d]\subset[a,b]$, and all $x \in \mathcal{X}$
\[
\E\!\left[s(Z,x) \left( Q_Y(d)-Q_Y(c) \right)\right]\le K(d-c).
\]
\end{assumption}

Assumption~\ref{asspt:average_lipschitz} is an average continuity condition in
the quantile index. It allows individual group distributions to have atoms, but
rules out a positive mass of groups having a jump at the same quantile index.
It is weaker than the pathwise Lipschitz conditions imposed in \CLP on the underlying group-specific coefficient and error
processes. Note that it is distinct from Assumption~\ref{asspt:average_strictness}. Assumption~\ref{asspt:average_strictness} gives a lower bound on the IV-weighted population slope, while
Assumption~\ref{asspt:average_lipschitz} gives an upper bound.

With these assumptions in hand, we now first establish functional CLTs for the unprojected and projected
estimators of the structural quantile function at a fixed $x$, respectively.

\begin{theorem}[CLT for the unprojected quantile function] \label{thm:convergence_unprojected}
Under Assumptions \ref{ass:full_rank}--\ref{asspt:finite_moments},
for each fixed $x\in\mathcal X$,
\[
\sqrt{n}\bigl(\hat\psi_x(\cdot)-\psi_x(\cdot)\bigr)
\leadsto
\mathbb G_x(\cdot)
\qquad\text{in }\ell^\infty([a,b]),
\]
where $\mathbb G_x$ is a tight mean-zero Gaussian process with covariance kernel
\[
\Gamma_x(u,u') = E\!\left[\phi_x(W;u)\phi_x(W;u')\right],
\]
and $\phi_x(W;u)$ is the influence function of $\hat\psi_x(u)$, given in
Appendix~\ref{sec:covariance_derivation}. Under correct specification, this
reduces to $\phi_x(W;u)=s(Z,x)\eta(u)$.
\end{theorem}

\begin{theorem}[CLT for the projected quantile function] \label{thm:convergence_projected}
Under the assumptions of Theorem~\ref{thm:convergence_unprojected} and
Assumptions~\ref{asspt:average_strictness} and
\ref{asspt:average_lipschitz}, for each fixed $x\in\mathcal X$,
\[
\sqrt{n}\left(\Pi_{\mathcal Q}(\hat\psi_x)(\cdot)-\psi_x(\cdot)\right)
\leadsto
\mathbb G_x(\cdot)
\qquad\text{in }\ell^\infty([a,b]),
\]
with the same Gaussian limit process as in
Theorem~\ref{thm:convergence_unprojected}.
\end{theorem}

The next two theorems establish functional CLTs for the unconstrained and \est estimators of the coefficient vector, respectively.

\begin{theorem}[CLT for the unconstrained coefficients] \label{thm:convergence_beta}
Under Assumption~\ref{ass:full_rank}--\ref{asspt:finite_moments},
\[
\sqrt{n}\big(\tilde{\beta}(\cdot) - \beta^{\mathrm{unc}}(\cdot)\big) \leadsto \mathbb{G}_{\beta}(\cdot) \quad \text{in } \ell^\infty([a,b])^{p+1},
\]
where $\mathbb{G}_{\beta}$ is a tight mean-zero Gaussian process with covariance kernel
\[
\Omega(u, u') = T \, \E\!\left[\Phi(u)\,\Phi(u')^{\!\T}\right] T^{\!\T},
\]
with $T := \mathrm{diag}(1, S_0)$, $S_0 := (\Sigma_{ZX}^{\!\T}\Sigma_{ZZ}^{-1}\Sigma_{ZX})^{-1}\Sigma_{ZX}^{\!\T}\Sigma_{ZZ}^{-1}$, and score
\[
\Phi(u)
:=
\begin{pmatrix}
Q_Y(u) - E[Q_Y(u)] \\[2pt]
\tilde{Z}\,\xi(u)
\end{pmatrix}
\in \R^{1+l},
\qquad
\xi(u) := Q_Y(u) - \mathbf{X}^{\!\T}\beta^{\mathrm{unc}}(u).
\]
In particular, the slope--slope block of $\Omega$ is $S_0\,\E[\tilde Z\tilde Z^{\!\T}\xi(u)\xi(u')]\,S_0^{\!\T}$. Under correct specification, $\beta^{\mathrm{unc}}(u) = \beta(u)$ and $\xi(u) = \eta(u)$.
\end{theorem}

\begin{theorem}[CLT for the \est estimator]\label{thm:pw_ols_clt}
Under Assumptions~\ref{ass:full_rank}--\ref{asspt:average_lipschitz},
\[
\sqrt{n}\big(\hat{\beta}^{\est}(\cdot)-\beta^{\est}(\cdot)\big)
\leadsto
\mathbb{G}_\beta(\cdot)
\quad \text{in } \ell^\infty([a,b])^{p+1},
\]
the same Gaussian process as the unprojected estimator.
\end{theorem}
Under these conditions, the \est estimator has the same asymptotic variance as
its unprojected variant, while Theorem~\ref{thm:pw_ols_improvement} guarantees smaller finite-sample error.

\begin{remark}[Inference when Assumption \ref{asspt:average_strictness} fails]
\label{rem:fang_santos}
If one does not maintain Assumption~\ref{asspt:average_strictness}, the
projection need not be asymptotically inactive. In that case, inference can in
principle be based on the Hadamard directional differentiability of the isotonic
projection map. If
\[
\sqrt n(\hat\psi_x-\psi_x)\leadsto \mathbb G_x,
\]
then the directional delta method gives
\[
\sqrt n\{\Pi_{\mathcal Q}(\hat\psi_x)-\Pi_{\mathcal Q}(\psi_x)\}
\leadsto
D\Pi_{\mathcal Q}(\psi_x)[\mathbb G_x].
\]
Under Assumption~\ref{asspt:average_strictness}, this derivative is the identity.
When the monotonicity constraint binds, for instance because \(\psi_x\) has flat
segments or because \(\psi_x\notin\mathcal Q\), the derivative instead projects
the perturbation onto the relevant tangent/critical cone, so the limit is
generally non-Gaussian. Inference can then be conducted using methods for
directionally differentiable functionals \citep[e.g.,][]{fang2019inference}. \qed
\end{remark}

\subsection{Inference}\label{sec:inference}
Under the conditions of Theorem~\ref{thm:pw_ols_clt}, projected and unprojected
\est share the limiting process $\mathbb G_\beta$. Therefore, one can use the unprojected influence functions as a starting point for developing inference procedures. We now describe pointwise and uniform procedures, distinguishing unprojected and projected variants. The uniform bootstrap confidence band results are new to the literature. 

\paragraph{Pointwise inference.} Write $\bar{Q}_n(u):=n^{-1}\sum_{j=1}^n Q_{Y_j}(u)$ for the cross-group average quantile function. For unprojected estimator $\tilde\beta(u)$, the asymptotic variance at a fixed $u_0$ is $\Omega(u_0, u_0)/n$ from Theorem~\ref{thm:convergence_beta}. A consistent estimator is
\[
\hat{\Omega}(u_0, u_0) = \hat{T} \left(\frac{1}{n}\sum_{j=1}^n \hat{\Phi}_j(u_0)\hat{\Phi}_j(u_0)^{\!\T}\right) \hat{T}^{\!\T},
\]
where $\hat{T} = \mathrm{diag}(1, \hat{S}_{\mathrm{2SLS}})$, $\hat{S}_{\mathrm{2SLS}}$ is the sample analogue of $S_0$, and the score is
\[
\hat{\Phi}_j(u) = \begin{pmatrix} Q_{Y_j}(u) - \bar{Q}_n(u) \\[2pt] (Z_j - \bar{Z}_n)\,\hat{\xi}_j(u) \end{pmatrix},
\qquad
\hat{\xi}_j(u) = Q_{Y_j}(u) - \hat{\mathbf{X}}_j^{\!\T}\tilde{\beta}(u).
\]
The sandwich pointwise $(1-\alpha)$ CI for $\beta_k(u_0)$ is $\tilde\beta_k(u_0)\pm z_{1-\alpha/2}\,\hat\sigma_k(u_0)/\sqrt{n}$, where $\hat\sigma_k^2(u_0):=\hat\Omega_{kk}(u_0,u_0)$.

For the \est estimator, the projected pointwise CI replaces $\tilde\beta$ with $\hat\beta^{\est}$ and recomputes the residuals accordingly: define $\hat{\xi}_j^{\est}(u):=Q_{Y_j}(u)-\hat{\mathbf{X}}_j^{\!\T}\hat\beta^{\est}(u)$ and
\[
\hat{\Phi}_j^{\est}(u) = \begin{pmatrix} Q_{Y_j}(u) - \bar{Q}_n(u) \\[2pt] (Z_j - \bar{Z}_n)\,\hat{\xi}_j^{\est}(u) \end{pmatrix},
\]
with $\hat\Omega^{\est}$ and $\hat\sigma_k^{\est}$ defined analogously. Under the conditions of Theorem~\ref{thm:pw_ols_clt}, both variance estimators are consistent for $\Omega$. However, the projected residuals $\hat\xi^{\est}_j$ are closer to the population residuals in finite samples (by Theorem~\ref{thm:pw_ols_improvement}), which can yield tighter pointwise CIs, as we demonstrate in the simulations and the empirical application below.

\paragraph{Uniform inference.} For uniform confidence bands over $u \in [a,b]$, we use a multiplier bootstrap. Let $\{\omega_j\}_{j=1}^n$ be i.i.d. multiplier weights independent of the data satisfying $E[\omega_j]=0$, $E[\omega_j^2]=1$, and $E|\omega_j|^{2+\delta}<\infty$ for some $\delta>0$. The unprojected bootstrap process is
\[
\hat{\mathbb{G}}_\beta^*(u) = \hat{T}\,\frac{1}{\sqrt{n}}\sum_{j=1}^n \omega_j \hat{\Phi}_j(u).
\]
Throughout, $P^x$ and $E^x$ denote probability and expectation conditional on the data (i.e., with respect to the multiplier weights only). Recall that $\hat{\mathbb{G}}_\beta^*\rightsquigarrow_{\Prob}\mathbb{G}_\beta$ denotes conditional weak convergence in probability \citep[Section~2.9]{vaart1996weak}. The following theorem establishes the conditional weak convergence of $\hat{\mathbb{G}}_\beta^*$.

\begin{theorem}[Multiplier bootstrap validity]\label{thm:bootstrap}
Let Assumption~\ref{ass:full_rank}--\ref{asspt:finite_moments} hold, and let
$\{\omega_j\}_{j=1}^n$ be i.i.d.\ multipliers, independent of the data, satisfying
$E[\omega_j]=0$, $E[\omega_j^2]=1$, and $E|\omega_j|^{2+\delta}<\infty$ for some
$\delta>0$. Then, conditionally on the data,
\[
\hat{\mathbb{G}}_\beta^*(\cdot)\rightsquigarrow_{\Prob}\mathbb{G}_\beta(\cdot)
\qquad\text{in }\ell^\infty([a,b])^{p+1},
\]
where $\mathbb{G}_\beta$ is the Gaussian process in Theorem~\ref{thm:convergence_beta}.
\end{theorem}

The $(1-\alpha)$ unprojected uniform band for $\beta_k(u)$ is $\tilde{\beta}_k(u) \pm \hat{c}_{1-\alpha}\cdot\hat{\sigma}_k(u)/\sqrt{n}$, where $\hat{c}_{1-\alpha}$ is the $(1-\alpha)$-quantile of $\sup_{u\in[a,b]} |\hat{\mathbb{G}}_{\beta,k}^*(u)| / \hat{\sigma}_k(u)$ conditional on the data. In practice, we sample $\omega_j \sim N(0,1)$, which satisfies the general conditions on the weights and is the standard approach.

The projected bootstrap constructs the bootstrap process differently: for each draw of multiplier weights, it computes bootstrap unconstrained coefficients $\tilde\beta^*(u)$, evaluates $\hat\psi^*_{X_j}(u):=\tilde\beta^*_0(u)+\tilde\beta^*_1(u)^{\!\T}(X_j-\hat\mu_X)$ at each observed $X_j$, applies $\Pi_{\mathcal{Q}}$ to each curve, and recovers projected bootstrap coefficients $\hat\beta^{\est,*}(u)$ according to Eq.\ \eqref{eq:div_estimator}. The projected bootstrap process is $\hat{\mathbb{G}}_\beta^{\est,*}(u):=\sqrt{n}(\hat\beta^{\est,*}(u)-\hat\beta^{\est}(u))$, and the projected uniform band for $\beta_k^{\est}(u)$ is constructed as above but with $\hat{\mathbb{G}}_\beta^{\est,*}$ replacing $\hat{\mathbb{G}}_\beta^*$, $\hat\beta^{\est}$ replacing $\tilde\beta$, and $\hat\sigma_k^{\est}$ replacing $\hat\sigma_k$. As in the pointwise case,  the projected bootstrap generally yields tighter pointwise CIs, as confirmed by the simulations and applications below. The following corollary establishes its validity.

\begin{corollary}[Projected bootstrap validity]\label{cor:bootstrap_projected}
Under the assumptions of Theorem~\ref{thm:bootstrap} and
Assumptions~\ref{asspt:bounded_support}--\ref{asspt:average_lipschitz}, the projected bootstrap satisfies
\[
\hat{\mathbb{G}}_\beta^{\est,*}(\cdot)
\rightsquigarrow_{\Prob}
\mathbb{G}_\beta(\cdot)
\qquad\text{in }\ell^\infty([a,b])^{p+1},
\]
and yields asymptotically valid uniform confidence bands for
$\beta_k^{\est}(u)$.
\end{corollary}

\subsection{Empirical quantile functions}\label{sec:empirical_quantiles}

In the previous sections, we have treated the group-level quantile functions $Q_{Y_j}$ as observed objects. This is suitable for settings with population-level within-group data (e.g., when all workers in a commuting zone or firm are observed). Here, we consider the case where researchers only observe a sample of the within-group realizations.

Suppose that for each group $j$, we observe draws $V_{jk}$, $k=1,\ldots,m_j$, from the distribution $Y_j$, which importantly need not be sampled i.i.d. The empirical distribution function and the associated empirical quantile function are
\[
\widehat Y_j(y):=\frac{1}{m_j}\sum_{k=1}^{m_j}\mathbf{1}\{V_{jk}\le y\},
\qquad
\widehat Q_{Y_j}(u):=\inf\{y:\widehat Y_j(y)\ge u\},
\quad u\in[a,b].
\]
The feasible IV-weighted quantile curve is then
\[
\bar\psi_x(u):=\frac{1}{n}\sum_{j=1}^n \hat s_j(Z_j,x)\,\widehat Q_{Y_j}(u),
\]
and we let $\bar{\tilde\beta}(u)=(\bar{\tilde\beta}_0(u),\bar{\tilde\beta}_1(u)^{\!\T})^{\!\T}$ denote the corresponding unconstrained coefficient functions defined by
\[
\bar\psi_x(u)=\bar{\tilde\beta}_0(u)+\bar{\tilde\beta}_1(u)^{\!\T}(x-\hat\mu_X).
\]
The projected feasible estimator is
\[
\bar\beta^{\est}(u)
:=
(\hat{\mathbf X}^{\!\T}\hat{\mathbf X})^{-1}\hat{\mathbf X}^{\!\T}\Pi_{\mathcal Q}(\bar\psi_{X_j})(u),
\]
that is, the estimator obtained from \eqref{eq:div_estimator} after replacing each $Q_{Y_j}$ by $\widehat Q_{Y_j}$. All other objects in Sections~\ref{sec:div_estimator}--\ref{sec:inference} are modified analogously.

Under the following high-level condition, this additional within-group sampling step is asymptotically negligible.

\begin{assumption}[High-level condition on empirical quantile functions]\label{asspt:empirical_quantiles}
There exist deterministic rates $r_{m_j}\downarrow 0$ such that, with
\[
R_n
:=
\max_{1\le j\le n}\sup_{u\in[a,b]}
\bigl|\widehat Q_{Y_j}(u)-Q_{Y_j}(u)\bigr|,
\]
we have
\[
R_n
=
O_p\!\left(\max_{1\le j\le n} r_{m_j}\right),
\qquad
\sqrt{n}\,\max_{1\le j\le n} r_{m_j}\to 0.
\]
\end{assumption}

 Assumption~\ref{asspt:empirical_quantiles} is stated in two parts. The first part requires that the stochastic fluctuations of the empirical quantile estimator are controlled uniformly over groups by a deterministic rate $\max_j r_{m_j}$, which follows from standard results under appropriate uniformity conditions, see Remark \ref{rem:sharper_growth} below.  Note that this allows for the within-group sampling to exhibit dependence by using appropriate quantile function estimators \citep[e.g.,][]{rio2017asymptotic}. Note that $Q_{Y_j}$ is itself a random target, but convergence here is conditional on a given group $j$. The randomness of $Q_{Y_j}$ across groups is accounted for in the across-group asymptotics above, and Assumption~\ref{asspt:empirical_quantiles} only controls the additional within-group estimation error $\widehat Q_{Y_j}-Q_{Y_j}$.  The second part  is a purely arithmetic growth condition ensuring that rate is fast enough relative to $n$. Together they imply $\sqrt{n}R_n = o_p(1)$, i.e., the first-stage quantile estimation error is uniformly smaller than the $\sqrt{n}$ rate governing the across-group IV problem. 
 
 The maximum over $j$ appears because Proposition~\ref{prop:empirical_quantiles} below requires control uniformly across all groups. For instance, if $m_{\min}:=\min_{1\le j\le n} m_j$ and the empirical quantiles satisfy
\[
R_n = O_p(m_{\min}^{-1/2}),
\] then Assumption~\ref{asspt:empirical_quantiles} reduces to the sufficient condition $n/m_{\min}\to 0$. More generally, any quantile estimator whose uniform error over groups is $o_p(n^{-1/2})$ can be used.

\begin{remark}[Sufficient conditions under density regularity]\label{rem:sharper_growth}
The high-level Assumption~\ref{asspt:empirical_quantiles} can be verified under various primitive conditions, yielding different growth requirements on $m_{\min}$ relative to $n$.

\begin{enumerate}[(i)]
\item \emph{Distribution-free (DKW).}
The Dvoretzky--Kiefer--Wolfowitz inequality \citep{dvoretzky1956asymptotic} gives, for each fixed group,
\(\sup_u|\widehat Q_{Y_j}(u)-Q_{Y_j}(u)|=O_p(m_j^{-1/2})\)
under a uniform lower density bound. Since Assumption~\ref{asspt:empirical_quantiles}
requires the maximum over \(j=1,\ldots,n\), a union bound gives
\[
R_n
=
O_p\!\left(\sqrt{\frac{\log n}{m_{\min}}}\right).
\]
Thus, using the DKW inequality, Assumption~\ref{asspt:empirical_quantiles} holds when
\[
\frac{n\log n}{m_{\min}}\to0.
\]
\item \emph{Bahadur linearization.} If, additionally, the within-group densities are 
uniformly bounded away from zero and Lipschitz on the relevant support (as assumed in Assumptions~5 and~7 in \CLP), one can decompose the empirical quantile 
error $\widehat Q_{Y_j}(u)-Q_{Y_j}(u)$ into a leading linear term plus a higher-order remainder. The leading Bahadur term is conditionally mean zero and is averaged across groups in the second stage. Its contribution is therefore of order
\(O_p((n m_{\min})^{-1/2})\), hence \(o_p(n^{-1/2})\) whenever
\(m_{\min}\to\infty\). The binding term is the uniform Bahadur remainder,
which is of order $\max_j\sup_u|r_j(u)|=O_p(m_{\min}^{-3/4}(\log(n\vee m_{\min}))^{3/4})$, which is faster  than the $m_{\min}^{-1/2}$ rate from DKW. This yields the sufficient condition
\[
\frac{n^{2/3}\log(n\vee m_{\min})}{m_{\min}}\to 0,
\]
which matches \CLP and is considerably weaker than $n/m_{\min}\to 0$ in that 
it allows the within-group samples to be much smaller than the number of groups.
\end{enumerate}
We state Assumption~\ref{asspt:empirical_quantiles} at a high level to accommodate both  regimes and to separate the across-group IV theory (Sections~\ref{sec:div_estimator}--\ref{sec:inference}) from the within-group estimation problem, which may employ quantile estimators other than the empirical quantile function.
\end{remark}

The next proposition shows that, under Assumption~\ref{asspt:empirical_quantiles}, replacing the true group quantile functions by empirical ones has no first-order effect. 

\begin{proposition}[Empirical quantile convergence]\label{prop:empirical_quantiles}
Suppose Assumptions~\ref{ass:full_rank}, \ref{asspt:finite_moments}, and~\ref{asspt:empirical_quantiles} hold.

\begin{enumerate}[(i)]
\item For each fixed $x\in\R^p$,
\[
\|\bar\psi_x-\hat\psi_x\|_{\ell^\infty([a,b])}
=
o_p(n^{-1/2}),
\qquad
\|\Pi_{\mathcal Q}(\bar\psi_x)-\Pi_{\mathcal Q}(\hat\psi_x)\|_{\ell^\infty([a,b])}
=
o_p(n^{-1/2}).
\]

\item The unconstrained coefficient functions satisfy
\[
\|\bar{\tilde\beta}-\tilde\beta\|_{\ell^\infty([a,b])^{p+1}}
=
o_p(n^{-1/2}).
\]

\item The projected coefficient functions satisfy
\[
\|\bar\beta^{\est}-\hat\beta^{\est}\|_{\ell^\infty([a,b])^{p+1}}
=
o_p(n^{-1/2}).
\]
\end{enumerate}
\end{proposition}

As an immediate consequence, all of the limit theories in Section~\ref{sec:asymptotic_distribution} continue to hold verbatim for the feasible estimators.

\begin{corollary}[Asymptotic limits with empirical quantiles]\label{cor:empirical_quantiles}
Under Assumption~\ref{asspt:empirical_quantiles}, the feasible estimators based on $\widehat Q_{Y_j}$ have the same weak limits as their infeasible counterparts based on $Q_{Y_j}$. More precisely:
\begin{enumerate}[(i)]
\item if the assumptions of Theorem~\ref{thm:convergence_unprojected} hold, then for each fixed $x$,
\[
\sqrt{n}\bigl(\bar\psi_x(\cdot)-\psi_x(\cdot)\bigr)
\leadsto
\mathbb{G}_x(\cdot)
\qquad\text{in }\ell^\infty([a,b]);
\]

\item if the assumptions of Theorem~\ref{thm:convergence_projected} hold, then for each fixed $x$,
\[
\sqrt{n}\bigl(\Pi_{\mathcal Q}(\bar\psi_x)(\cdot)-\psi_x(\cdot)\bigr)
\leadsto
\mathbb{G}_x(\cdot)
\qquad\text{in }\ell^\infty([a,b]);
\]

\item if the assumptions of Theorem~\ref{thm:convergence_beta} hold, then
\[
\sqrt{n}\bigl(\bar{\tilde\beta}(\cdot)-\beta^{\mathrm{unc}}(\cdot)\bigr)
\leadsto
\mathbb{G}_\beta(\cdot)
\qquad\text{in }\ell^\infty([a,b])^{p+1};
\]

\item if the assumptions of Theorem~\ref{thm:pw_ols_clt} hold, then
\[
\sqrt{n}\bigl(\bar\beta^{\est}(\cdot)-\beta^{\est}(\cdot)\bigr)
\leadsto
\mathbb{G}_\beta(\cdot)
\qquad\text{in }\ell^\infty([a,b])^{p+1}.
\]
\end{enumerate}
\end{corollary}

\section{Monte Carlo simulations}
\label{sec:simulation}

To empirically validate the theoretical results above, we now compare unprojected and projected \est across several DGP configurations. Unprojected \est coincides with the \CLP and \MP estimators without individual covariates, though algorithmically our estimator, by working directly with the quantile functions as outcomes, does not require first-stage quantile regression. For each, we report coefficient integrated mean squared error (IMSE), the average squared Wasserstein distance $W_2^2$ between estimated and true conditional quantile functions, (uniform) confidence band coverage, and computational benchmarks. Across the reported configurations, the IMSE improves by up to 63\%. Our uniform and projected bootstraps also achieve correct nominal coverage with up to 1.4\% smaller width in the simulations (10\% in the empirical application). Finally, computational benchmarks show our one-step estimator is faster than the two-step estimators of \CLP and \MP without individual covariates (see Appendix \ref{sec:computation}). 

\subsection{Data-generating process}
The simulations use the grouped-data IV design in \CLP as the benchmark. Given our focus on unconditional group-level treatment effect, we drop \CLP's individual-level covariate term (see the discussion in Section \ref{sec:identification_interpretation}). The data-generating process (DGP) is
\begin{equation}\label{eq:dgp_sim}
Y_{ij} = b_0(U_{ij}) + \frac{U_{ij}}{2} + X_j\,\gamma(U_{ij}) + (\zeta_j-\tfrac{1}{2})U_{ij}
        + \sum_{k=1}^{p-1} W_{jk}h_k(U_{ij}),
\end{equation}
where $U_{ij}\sim U(0,1)$, $\zeta_j\sim U(0,1)$ and
\[
X_j^0=(\pi_Z+\delta(\zeta_j-\tfrac{1}{2}))Z_j+\zeta_j+\nu_j,
\qquad
Z_j,\nu_j\sim \exp(0.25\,\mathcal{N}(0,1)).
\]
The terms $U_{ij}/2$ and $(\zeta_j-1/2)U_{ij}$ are written separately to match CLP's location term and group error. Together they equal $\zeta_j U_{ij}$.

Panel~D in Table \ref{tab:simulations} is the exact no-individual-covariate CLP benchmark. If $b_0\equiv0$, $\gamma(u)=\sqrt u$, $\pi_Z=1$, $\delta=0$, $p=1$, and $X_j=X_j^0$, then
\[
y_{ij}=X_j\sqrt{u_{ij}}+\zeta_j u_{ij},\qquad X_j=Z_j+\zeta_j+\nu_j,
\]
which gives the DGP in CLP Appendix A with the individual covariate term omitted. The derivative of this population quantile function is $X_j/(2\sqrt u)+\zeta_j$, so the benchmark is strongly monotone because $X_j>0$ and $\zeta_j\ge0$.

Panels~A--C keep the same first-stage endogeneity but add features that are common in applications: moderate instruments, controls, heterogeneous first stages, and treatment effects that need not be monotone in $u$. In these panels we standardize the endogenous first stage, $\bar X_j^0=(X_j^0-\bar X^0)/s_{X^0}$, and set $X_j=B_X\tanh(\bar X_j^0/B_X)$ with $B_X=1.5$ to make sure $X$ has bounded support without artificially clipping the range. Controls, when included, are bounded in the same way: $W_{jk}=B_W\tanh(\tilde W_{jk}/B_W)$ with $B_W=1.5$. Panel~B sets their true effects to zero, so it isolates the finite-sample cost of estimating extra nuisance coefficient functions. Panel~C uses mixed control effects: odd controls have $h_k(u)=u$ and even controls have $h_k(u)=0.1\sin(k\pi u)$. The treatment coefficient is
\[
\gamma(u)=\sqrt{u}+\beta_{\text{slope}}\sin(2\pi u),
\]
so $\beta_{\text{slope}}=0$ gives the CLP coefficient and $\beta_{\text{slope}}>0$ adds curvature.

The extra baseline term $b_0(u)=\sigma q_0(u)$ lets us vary the shape of the within-group outcome distribution while preserving valid population quantile functions. For Panels~A--C, we choose $\sigma$ analytically. Let $m_0=\inf_{u\in[a,b]}q_0'(u)$, $L_\gamma=\sup_{u\in[a,b]}|\gamma'(u)|$, and $L_k=\sup_{u\in[a,b]}|h_k'(u)|$. We set
\[
\sigma m_0>B_XL_\gamma+B_W\sum_k L_k.
\]
Then every population quantile function generated by the DGP is strictly increasing on $[a,b]$. Any monotonicity violations reported below are therefore finite-sample violations in the fitted 2SLS curves, not failures of the population DGP.

\subsection{Results}

Table~\ref{tab:simulations} reports coefficient IMSE, $W_2^2$, and the fraction of groups with non-monotone fitted curves, all averaged over 500 replications with 19 quantile grid points from $u = 0.05$ to $u = 0.95$. Unless otherwise noted, the designs use $n = 50$ groups and $N = 50$ observations per group.

\begin{table}[htbp]
\centering\scriptsize
\caption{Monte Carlo results: IMSE and $W_2^2$ improvement from projection.}
\label{tab:simulations}
\resizebox{\textwidth}{!}{%
\begin{tabular}{lccccccr}
\toprule
 & \multicolumn{3}{c}{Coefficient IMSE} & \multicolumn{3}{c}{$W_2^2$} & \\
\cmidrule(lr){2-4} \cmidrule(lr){5-7}
 & Unproj. & Proj. & \% Gain & Unproj. & Proj. & \% Gain & Inv.~(\%) \\
\midrule
\multicolumn{8}{l}{\textit{Panel A: Instrument strength} ($n=50$, $N=50$, CLP-base first stage)} \\
\addlinespace[2pt]
$F$ & & & & & & & \\
5 & 31.06 & 11.54 & 62.8 & 17.82 & 9.38 & 47.3 & 14.7 \\
10 & 0.389 & 0.359 & 7.9 & 0.291 & 0.275 & 5.4 & 5.8 \\
21 & 0.155 & 0.155 & 0.4 & 0.156 & 0.156 & 0.2 & 1.4 \\
\addlinespace[4pt]
\multicolumn{8}{l}{\textit{Panel B: Included controls \& heterogeneous first stage} ($n=50$, $N=50$, median $F\approx 11$--$20$)} \\
\addlinespace[2pt]
$p=1$, $\delta=0.0$ & 0.155 & 0.155 & 0.4 & 0.156 & 0.156 & 0.2 & 1.4 \\
$p=1$, $\delta=0.5$ & 0.912 & 0.656 & 28.1 & 0.637 & 0.509 & 20.2 & 4.7 \\
$p=2$, $\delta=0.0$ & 0.139 & 0.137 & 1.2 & 0.164 & 0.163 & 0.5 & 1.7 \\
$p=2$, $\delta=0.5$ & 0.615 & 0.435 & 29.4 & 0.467 & 0.376 & 19.6 & 5.2 \\
$p=5$, $\delta=0.0$ & 0.184 & 0.182 & 1.1 & 0.245 & 0.244 & 0.4 & 2.6 \\
$p=5$, $\delta=0.5$ & 0.654 & 0.494 & 24.5 & 0.515 & 0.437 & 15.1 & 6.5 \\
\addlinespace[4pt]
\multicolumn{8}{l}{\textit{Panel C: Realistic combination} ($n=50$, $N=25$, median $F\approx 11$)} \\
\addlinespace[2pt]
$p=5$, $\delta=1.0$, LN, $\beta_{\text{slope}}=0.2$ & 9.23 & 7.42 & 19.6 & 7.85 & 6.94 & 11.5 & 13.0 \\
\addlinespace[4pt]
\multicolumn{8}{l}{\textit{Panel D: CLP no-individual-covariate benchmark} (median $F$ shown by row)} \\
\addlinespace[2pt]
$F \approx 11$, $(n,N)=(25,25)$ & 0.281 & 0.233 & 17.0 & 0.071 & 0.066 & 7.2 & 11.2 \\
$F \approx 11$, $(n,N)=(25,50)$ & 0.274 & 0.249 & 8.8 & 0.067 & 0.064 & 3.7 & 6.3 \\
$F \approx 20$, $(n,N)=(50,50)$ & 0.024 & 0.024 & 0.1 & 0.033 & 0.033 & 0.0 & 1.3 \\
\bottomrule
\end{tabular}%
}
\floatfoot{\footnotesize \textit{Notes:} 500 replications; $n = 50$, $N = 50$ unless noted; 19 quantile grid points ($u = 0.05, \ldots, 0.95$). IMSE $= \int \|\hat{\beta}_1(u) - \beta_1(u)\|^2\, du$ averaged over replications; $W_2^2$: average squared Wasserstein distance between estimated and true conditional distributions; ``Inv.'': fraction of groups with non-monotone fitted $\hat{\psi}_{X_j}$. Panel~A varies the first-stage $F$-statistic ($p = 1$, no controls). Panel~B adds included controls ($p$) with zero true direct effects and first-stage heterogeneity ($\delta$), isolating the finite-sample cost of estimating nuisance coefficient functions; its first-stage $F$ range reflects the displayed rows. Panel~C combines small cells, many controls, a heterogeneous first stage, lognormal outcomes, treatment heterogeneity, and a first stage above the usual weak-instrument threshold. All panels build on the CLP Appendix-A grouped-data structure. Panel~D replicates the exact DGP from \CLP: $Y_{ij} = X_j\sqrt{U_{ij}} + \zeta_j U_{ij}$ and $X_j = Z_j + \zeta_j + \nu_j$, with uncentered log-normal $Z_j$ and $\nu_j$.}
\end{table}

\paragraph{Instrument strength (Panel~A).} With a lognormal base distribution, a single endogenous regressor, and no controls, the largest gains occur when the first stage is weak or moderate. At $F \approx 5$, \est reduces IMSE by 63\% and $W_2^2$ by 47\%, with 15\% of groups having non-monotone fitted distributions. At $F \approx 10$ the corresponding gains fall to 8\% and 5\%. With a stronger first stage ($F \approx 21$), violations are rare and gains are close to zero.

\paragraph{Controls and heterogeneous first stage (Panel~B).} Panel~B adds controls to the estimating equation while setting their true direct effects to zero. This isolates the finite-sample cost of estimating more nuisance coefficient functions. With $\delta=0$, increasing $p$ from 1 to 5 raises the share of non-monotone fitted curves from 1.4\% to 2.6\%, while projection gains remain small. With first-stage heterogeneity ($\delta=0.5$), violations rise from 4.7\% to 6.5\%. The projection then matters for estimation: with $p=2$, IMSE falls by 29\% and $W_2^2$ by 20\%; with $p=5$, the reductions are 25\% and 15\%. The gains need not be monotone in $p$, but the invalid-rate column shows the dimensionality channel directly. Each additional regressor adds a dimension to the fitted curve $\hat{\psi}_{X_j}(u)=\tilde{\beta}_0(u)+\sum_k\tilde{\beta}_{1,k}(u)(X_{jk}-\hat\mu_k)$, making it harder for all dimensions to satisfy monotonicity at once.

\paragraph{Realistic combination (Panel~C).} Panel~C combines small cells ($N=25$), $p=5$ controls, a heterogeneous first stage ($\delta=1$), lognormal outcomes, treatment-effect curvature ($\beta_{\text{slope}}=0.2$), and median first-stage $F\approx 11$. This conservative stress test gives a 20\% IMSE reduction and a 12\% $W_2^2$ reduction. About 13\% of fitted quantile functions are non-monotone under unprojected \est.

\paragraph{CLP benchmark (Panel~D).} Panel~D is the exact no-individual-covariate CLP parameter point described above. In this all-positive-treatment design, population monotonicity is very strong: $Q_{Y_j}'(u)=X_j/(2\sqrt{u})+\zeta_j>0$. Projection gains are therefore concentrated at the smallest group counts, where sampling noise and weak first stages still create occasional fitted-curve violations. At $(n,N)=(25,25)$, IMSE falls by 17\%; by $(n,N)=(50,50)$ the gain is essentially zero.

\subsection{Inference}

Theorem~\ref{thm:pw_ols_clt} establishes that projected and unprojected \est share the same asymptotic Gaussian process under the stated assumptions, so standard inference (sandwich SEs, multiplier bootstrap) applies without modification.

Table~\ref{tab:coverage} reports pointwise and uniform coverage alongside band widths from 500 replications with $B = 500$ bootstrap draws. For unprojected \est we use sandwich SEs and the unprojected bootstrap; for projected \est we use the projected bootstrap, which runs PAVA inside each bootstrap draw. The first three rows use the simple $p=1$ design and vary $n$ and $N$ while keeping the median first-stage $F$ at about 10 or 20. The last three use the realistic design from Panel~C and vary $n$ and $N$, with a median $F$-statistic of 11 and 23.

\begin{table}[htbp!]
\centering\small
\caption{Coverage and uniform band widths for $\beta_1(u)$ (nominal 95\%).}
\label{tab:coverage}
\begin{tabular}{lcc cc cc r}
\toprule
& \multicolumn{2}{c}{PW coverage (\%)} & \multicolumn{2}{c}{UB coverage (\%)} & \multicolumn{2}{c}{UB width} & \\
\cmidrule(lr){2-3} \cmidrule(lr){4-5} \cmidrule(lr){6-7}
& Unproj. & Proj. & Unproj. & Proj. & Unproj. & Proj. & $\Delta$ width \\
\midrule
\multicolumn{8}{l}{\textit{Simple design ($p = 1$)}} \\
\addlinespace[2pt]
$n=50$, $N=25$, $F\approx 10$ & 95.4 & 95.4 & 93.8 & 93.8 & 2.31 & 2.28 & $-1.4\%$ \\
$n=50$, $N=50$, $F\approx 10$ & 95.5 & 95.5 & 94.0 & 94.0 & 1.74 & 1.73 & $-0.8\%$ \\
$n=100$, $N=50$, $F\approx 20$ & 95.5 & 95.5 & 95.8 & 95.4 & 1.25 & 1.25 & $-0.2\%$ \\
\addlinespace[4pt]
\multicolumn{8}{l}{\textit{Realistic design ($p = 5$, $\delta=1$, $\beta_{\text{slope}}=0.2$)}} \\
\addlinespace[2pt]
$n=50$, $N=25$, $F\approx 11$ & 95.6 & 95.7 & 96.8 & 96.6 & 6.97 & 6.94 & $-0.4\%$ \\
$n=50$, $N=50$, $F\approx 11$ & 94.6 & 94.6 & 93.8 & 93.8 & 5.23 & 5.19 & $-0.7\%$ \\
$n=100$, $N=50$, $F\approx 23$ & 95.6 & 95.6 & 95.8 & 95.8 & 3.67 & 3.66 & $-0.1\%$ \\
\bottomrule
\end{tabular}
\floatfoot{\footnotesize \textit{Notes:} 500 replications; $B = 500$ multiplier bootstrap draws; lognormal base distribution throughout. PW: pointwise; UB: uniform band. Band widths are medians across replications. ``$\Delta$ width'': percentage reduction in UB width for projected \est relative to unprojected \est. Unprojected \est uses sandwich SEs and the unprojected bootstrap; projected \est uses the projected bootstrap (projection applied inside each draw). \textit{Simple design} rows use $p=1$, $\delta=0$, and $\beta_{\text{slope}}=0$.}
\end{table}

Coverage is near-nominal for both methods across all configurations, confirming the asymptotic equivalence in Theorem~\ref{thm:pw_ols_clt}. The projection does not distort inference. Because \est has lower IMSE at comparable standard errors, it has weakly higher power for testing $\beta_1(u) = 0$.

The projected bootstrap produces similar or narrower uniform bands for \est. In these coverage simulations, the reduction is below 1.5\%; in the empirical CLP application it is about 10\% (Section~\ref{sec:empirical}). The projected bootstrap captures the finite-sample regularization from the projection, yielding a weakly tighter critical value for the sup-statistic.

\section{Empirical illustrations}\label{sec:empirical}

We present two applications. Section~\ref{sec:empirical_clp} revisits the distributional effects of Chinese import competition on wages studied by \citet{autor2013china} and \CLP, illustrating the finite-sample gains from the \est projection and inference. Section~\ref{sec:empirical_fsp} studies the Food Stamp Program's effect on birth weights analyzed by \citet{almond2011inside} and \MP, using the decomposition from Appendix~\ref{app:individual_decomposition} to compare the \est estimates to the conditional quantile treatment effect estimates.

\subsection{Import competition and the wage distribution}\label{sec:empirical_clp}

We revisit the distributional effects of Chinese import competition on U.S.\ wages, following \CLP and \citet{autor2013china} (ADH henceforth). Their setting is well-suited for illustrating \est: there is one endogenous group-level treatment (import exposure), a standard instrument (other-country imports), and a distribution-valued outcome (the wage distribution within each commuting zone (CZ)).

\paragraph{Setting and data.}
\CLP estimate the effect of rising Chinese imports on the distribution of log weekly wages across U.S.\ commuting zones. The unit of observation is a CZ--decade pair ($j = 1, \ldots, 722$ CZs, two periods: 1990--2000 and 2000--2007, giving $n = 1{,}444$ observations). For each CZ--decade, the outcome is the vector of 19 within-CZ empirical quantiles $\hat Q_{Y_j}(u_q)$ for $u_q \in \{0.05, 0.10, \ldots, 0.95\}$. The endogenous variable $X_j$ is the change in Chinese import exposure per worker in commuting zone $j$. Following \citet{autor2013china}, $Z_j$ is a shift-share instrument that assigns Chinese exports to other high-income countries to each CZ using lagged local industry employment shares, isolating China-side supply variation. The IV regression of $\hat Q_{Y_j}(u)$ on $X_j$ at the CZ level includes six continuous controls (manufacturing employment share, college share, foreign-born share, female employment share, routine task intensity, outsourcing exposure), eight census region dummies, and a period dummy, with the same controls included in the instrument vector $Z_j$. Regressions are weighted by CZ start-of-period population and standard errors are clustered by state (using a cluster multiplier bootstrap for the uniform bands). The first-stage $F$-statistic is 533, so the instrument is strongly relevant.

We apply \est to this data and compare the results with the original \CLP quantile-by-quantile 2SLS estimates, which do not include individual-level covariates. In their replication package, CLP estimate a 2SLS regression with the decade-equivalent \textit{change} in CZ wage quantiles as outcome. To match this approach, we project the level quantile functions in each period and then recover the slope coefficients by OLS on the long difference between these projected quantile functions. Since both estimators share the same asymptotic distribution under the conditions
of Theorem~\ref{thm:pw_ols_clt}, any differences between them reflect the finite-sample effect of the \est projection step.

\paragraph{Results.} Figure~\ref{fig:clp_comparison} overlays the \CLP estimates (black circles) with the \est estimates (red triangles) for the full sample. Both series tell the same qualitative story: import competition depresses wages across the distribution, with the largest effects at lower quantiles (around $-1.4$ log points at the 10th quantile) and smaller effects in the upper tail ($-0.4$ to $-0.5$). The two coefficient functions are visually indistinguishable. The reason, as mentioned, is that \CLP's specification is in long differences so the natural object whose validity the \est projection enforces is the implied \emph{level} CZ wage quantile function for each (CZ, decade) cell. Anchoring those level quantile functions at the IPUMS-derived 1990 baseline quantile, fewer than $2\%$ violate monotonicity (0.0\% in the full sample, 0.8\% for females, 1.4\% for males), and the projection alters the slope by at most $0.0012$ log points at any quantile. This does not mean the projection is irrelevant: applying it inside the bootstrap still smooths the finite-sample distribution of the coefficient process and tightens the confidence bands reported below.

\begin{figure}[ht!]
\centering
\includegraphics[width=0.85\textwidth]{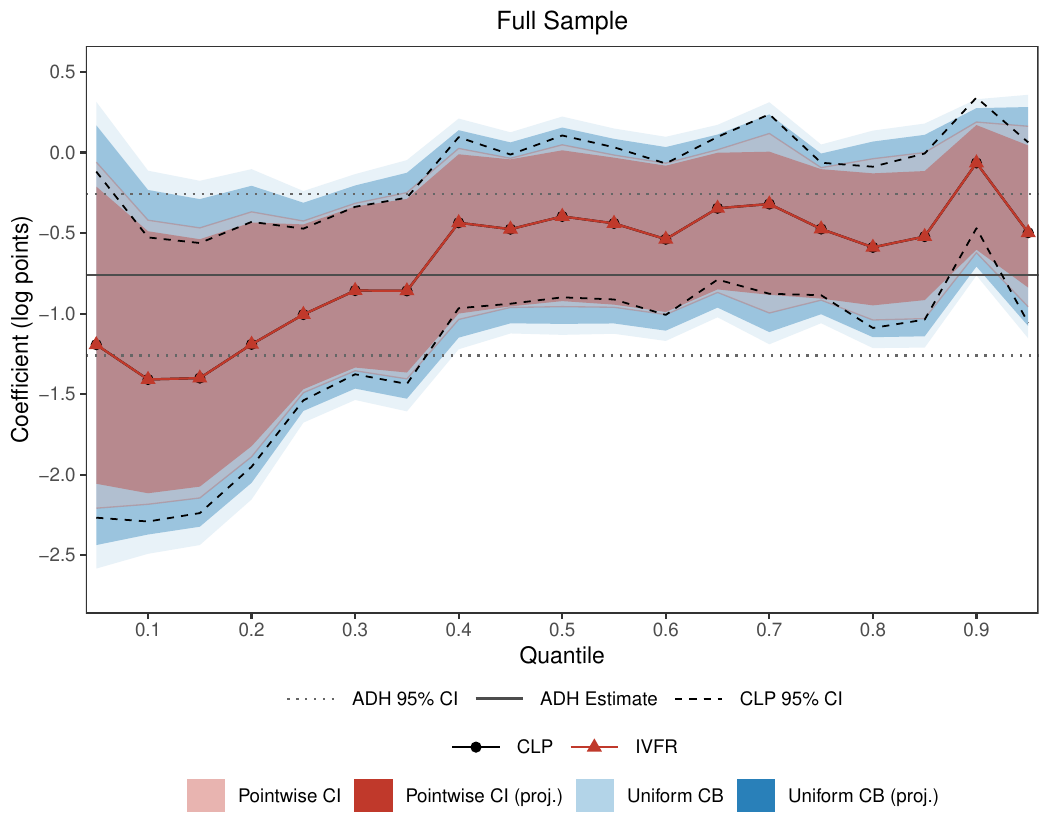}
\caption{Chinese import competition and the U.S. wage distribution.}
\floatfoot{\footnotesize \textit{Notes:} \CLP  estimates and \est estimates of the effect of Chinese import competition on log weekly wages, full sample ($n = 1{,}444$). Black circles: \CLP quantile-by-quantile 2SLS. Red triangles: \est. Dashed lines: \CLP pointwise 95\% CI (clustered by state). Red shaded bands around \est: projected bootstrap pointwise CI (darkest), sandwich pointwise CI (medium), and uniform confidence band (lightest), all clustered by state. Horizontal lines: ADH mean estimate and 95\% CI from \citet{autor2013china}.}
\label{fig:clp_comparison}
\end{figure}

Figure~\ref{fig:clp_comparison} displays four layers of confidence bands around the \est estimates. The two red bands show pointwise 95\% CIs: the outer (lighter) band uses cluster-robust sandwich SEs (identical to those of \CLP); the inner (darker) band uses projected-bootstrap SEs that apply PAVA inside each multiplier-bootstrap draw. The two blue bands show cluster-robust uniform 95\% bands over $u \in [0.05, 0.95]$: the outer (lighter) band uses the unprojected multiplier-bootstrap critical value; the inner (darker) band uses the projected-bootstrap critical value. The uniform bands are new---\CLP did not provide a procedure to construct uniform confidence bands. The projected pointwise CIs are on average 9.4\% narrower than the sandwich CIs (with reductions of up to 22\% at quantiles where the unconstrained coefficient function is most non-monotone, e.g.\ $u = 0.85$); the projected uniform bands are 10.1\% narrower at every $u$. Thus, the projection does meaningful empirical work for precision even though it does not change the CLP point estimate. Finally, we find that the unprojected multiplier-bootstrap SEs are within 1\% of the sandwich SEs at every quantile, confirming that both estimate the same asymptotic variance.

The uniform bands refine \CLP's conclusions about where effects are significant. \CLP reported pointwise CIs, which reject the null of no effect at 12 of 19 quantiles. The projected pointwise CIs reject at 15 quantiles, illustrating the power gain coming from the narrower confidence bands. Pointwise inference, however, does not account for simultaneous testing across the quantile grid. The projected uniform band, which does, rejects the null at quantiles concentrated in the 10th--35th percentile range, plus $u = 0.75$. The very bottom of the distribution ($u = 0.05$), where the \CLP pointwise CI barely excludes zero, does not survive the uniform correction. Above the median, effects are negative but imprecisely estimated. The projection adds modest power here as well: the unprojected uniform bands reject at 6 quantiles, and the projection flips $u = 0.75$ to significant. Import competition thus has its strongest and most robust effects on lower-middle wages, rather than at the very bottom of the distribution.

\paragraph{Subsampling exercise.} The full \CLP sample has 722 CZs, causing unprojected and projected \est to produce near-identical point estimates, in line with the simulation evidence. As shown, the projection nonetheless still provides benefits in the form of tighter confidence bands. To further assess the practical gains from the projection step at smaller sample sizes, we subsample the data at various CZ counts and compare IMSEs.

\begin{figure}[hb!]
\centering
\includegraphics[width=\textwidth]{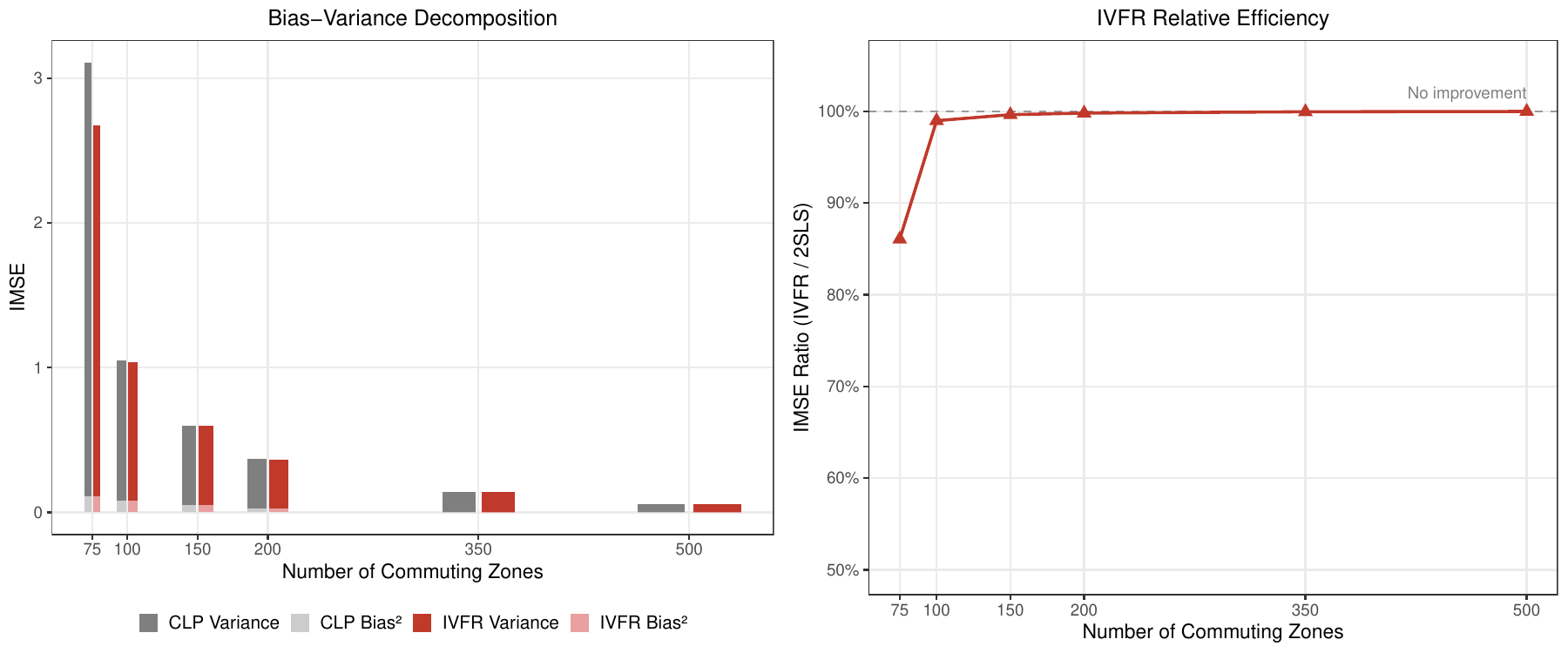}
\caption{Unprojected vs. projected \est IMSE in \CLP subsampling exercise}
\floatfoot{\footnotesize \textit{Notes:} Left panel: IMSE bias--variance decomposition for unprojected and projected \est at different numbers of commuting zones $M$ ($M = 50$ excluded; both methods have IMSE $> 170$). Bars show integrated variance (dark) and integrated squared bias (light), stacked. Right panel: IMSE ratio (\est); values below 100\% indicate \est outperforms. Based on 500 random subsamples per $M$, with full-sample projected \est as target.}
\label{fig:clp_subsampling}
\end{figure}

Concretely, for each subsample size $M \in \{75, 100, 150, 200, 350, 500\}$, we draw 500 random subsets of $M$ commuting zones (keeping both decades for each CZ), re-estimate both unprojected and level-projected \est on each subsample, and compute the $\text{IMSE} = \int (\hat{\beta}_1(u) - \beta_1^*(u))^2\, du$, where $\beta_1^*(u)$ is the full-sample \est estimate. We decompose the IMSE into integrated squared bias and integrated variance.

Figure~\ref{fig:clp_subsampling} reports the results. The left panel shows the bias-variance decomposition; both methods are variance-dominated at every sample size, with squared bias accounting for less than 10\% of total IMSE. The projection's effect is concentrated at the smallest sample sizes, where weaker first stages create larger fitted-CDF excursions for it to fix. At $M = 75$ (150 observations for 17 regressors), \est reduces IMSE by 14\%, almost entirely through variance reduction. At $M = 100$ the gain shrinks to 1\%, and at $M \ge 150$ the projection has essentially no IMSE effect. The simulation results in Table~\ref{tab:simulations} show the same pattern across DGPs: projection improves the IMSE materially when instruments are weak or fitted curves are oscillatory.

\subsection{Food stamps and the birth weight distribution}\label{sec:empirical_fsp}

We now turn to the empirical application in \MP, who study the effect of the Food Stamp Program (FSP) on birth weights. This setting provides a natural laboratory for the decomposition developed in Appendix~\ref{app:individual_decomposition}: the treatment (food stamps) is individually targeted but varies at the group level, creating a gap between the conditional quantile treatment effect $\delta(u)$ estimated by \MP and the group-level distributional effect $\beta(u)$ estimated by \est.

\paragraph{Setting and data.}
Following \citet{almond2011inside} and \MP, we study the staggered county-level introduction of the FSP between 1964 and 1975. Groups are county--trimester cells; within-group units are births in a county--trimester. The outcome is birth weight in grams. We use natality microdata from the NCHS (1968--1977) merged with county-level food stamp adoption dates from \citet{almond2011inside}, restricting to births by Black mothers ($n \approx 2.8$ million individual births in approximately 17{,}000 county--trimester groups with at least 25 observations). Controls include per capita income, government transfers, and 1960 county characteristics interacted with a linear time trend. All regressions include county, state$\times$year, and trimester fixed effects, with standard errors clustered by county.

The treatment indicator $\mathit{fsp}_{ct}$ equals one if a food stamp program was in place at least three months before birth in county $c$ in trimester $t$. Identification in \MP treats variation in FSP rollout as quasi-random, conditional on controls, so both estimators use $Z = X$. 

\paragraph{Estimands.}
\MP estimate a conditional quantile model with individual-level covariates (child sex, mother's age and its square, legitimacy status),
\begin{equation}\label{eq:mp_model}
Q(u, \mathit{bw}_{ict} \mid \mathit{fsp}_{ct}, x_{1ict}, x_{2ct}, v_{ct}) = \mathit{fsp}_{ct}\, \delta(u) + x_{1ict}'\, \beta(u) + x_{2ct}'\,\gamma_2(u) + \alpha(u, v_{ct}),
\end{equation}
where $\mathit{fsp}_{ct}$ was introduced above, $ x_{1ict}$ are individual-level controls, $x_{2ct}$ are the county-level controls mentioned above, and $\alpha(u, v_{ct})$ is the county-level unobservable.
The coefficient $\delta(u)$ is the direct within-type effect: the shift in the $u$-th conditional quantile for a given type of mother. The \est estimand $\beta_1(u)$ targets the effect on the realized group quantile---the actual $u$-th percentile of birth weights in a county--trimester cell.

\paragraph{Results.}
We replicate the \MP estimates using their \texttt{mdqr} package \citep{mdqr} and estimate $\beta_1(u)$ via \est on the same data. Figure~\ref{fig:fsp_decomposition} shows the decomposition from equation~\eqref{eq:TE_beta_new3}:
\[
\beta_1(u) = \delta(u) + \underbrace{\E[\bar W_j(1) - \bar W_j(0)]'\gamma(u)}_{\text{composition}} + \underbrace{\E[\Delta_j(u;1) - \Delta_j(u;0)]}_{\text{re-ranking}}.
\]
The composition-fixed quantile $Q_{Y_j}^{\oplus}(u)$ is computed as the within-group mean of the first-stage fitted values from \texttt{mdqr}, and the re-ranking gap $\Delta_j(u) = Q_{Y_j}(u) - Q_{Y_j}^{\oplus}(u)$ is the difference between the realized group quantile and this average. Each component is then regressed on $\mathit{fsp}$ with the same fixed effects.

\begin{figure}[ht!]
\centering
\includegraphics[width=0.85\textwidth]{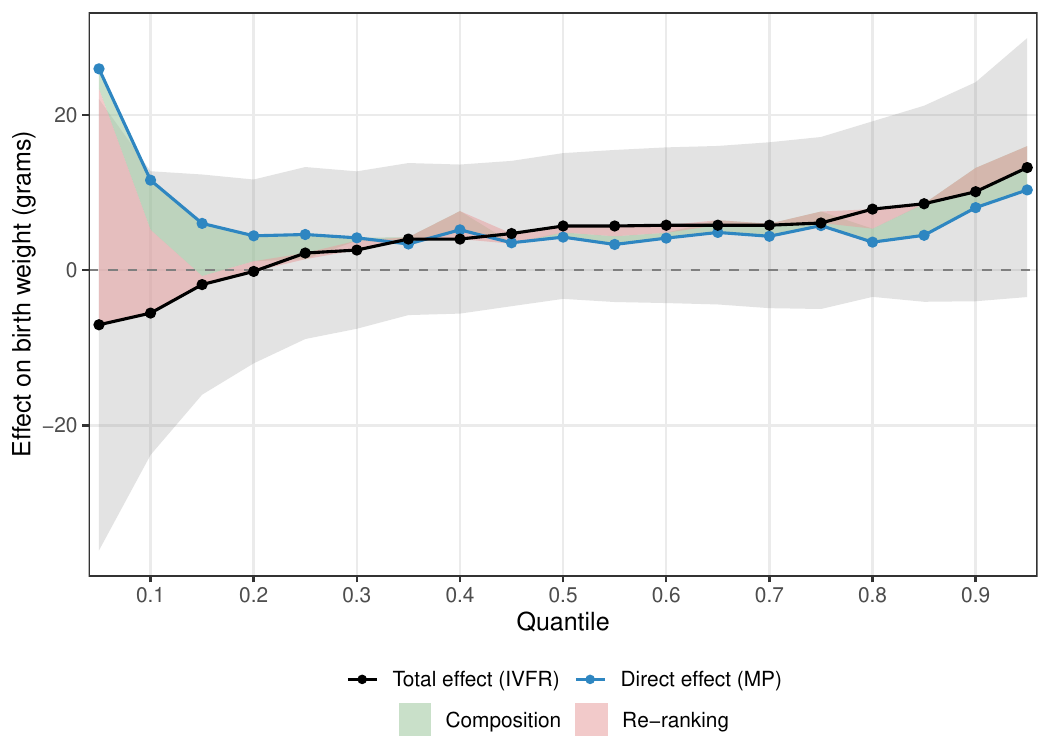}
\caption{Effects of FSP on the birth weight distribution (black mothers).}
\floatfoot{\footnotesize \textit{Notes:} Black: total \est effect $\beta_1(u)$, with a 95\% uniform confidence band (grey ribbon) from a cluster multiplier bootstrap over $u\in[0.05,0.95]$. Blue: direct within-type effect $\delta(u)$ from \citet{melly2025minimum}.}
\label{fig:fsp_decomposition}
\end{figure}

At the 5th percentile, $\delta(0.05) \approx 26$ grams (s.e.\ $9.5$): holding mother type fixed, FSP raises the conditional 5th percentile of birth weight by about 26 grams, a statistically significant effect.\footnote{Our point estimates of $\delta(u)$ are slightly smaller than those reported in \MP---e.g., 26 vs.\ nearly 30 grams at the 5th percentile---most likely reflecting minor differences in the NCHS natality vintage and the county crosswalk used to merge births with the \citet{almond2011inside} FSP rollout data, which leave us with 18{,}865 black county--trimester cells versus 19{,}482 in \MP.} Yet $\beta_1(0.05) \approx -7$ grams (s.e.\ $10.8$)---the actual 5th percentile of the county birth weight distribution does not significantly move. The 33-gram gap is almost entirely accounted for by the re-ranking
component, while the composition component is negligible.

The \est coefficient function $\beta_1(u)$ slopes upward: it is essentially zero at most quantiles and rises to roughly $+13$ grams at the 95th percentile. The pointwise 95\% CI excludes zero at $u=0.95$ and is borderline at the next three grid points, consistent with FSP shifting the right tail of the birth weight distribution modestly upward. The uniform 95\% confidence band, however, contains zero at every quantile. The pointwise significance at the top should therefore be read as suggestive evidence rather than a statistically robust finding at the conventional uniform level.

The large re-ranking term can be explained as follows. An individual at the \emph{group}'s 5th percentile is generally not at her own \emph{conditional} 5th percentile. A young unmarried mother (whose child has lower baseline birth weight) sitting at the group's 5th percentile may be at, say, her conditional 15th percentile, where the treatment effect $\delta(0.15) \approx 6$ grams is far smaller than $\delta(0.05) \approx 26$ grams. The re-ranking term aggregates these within-type percentile shifts across all types at the group quantile cutoff: because $\delta(u)$ is steeply decreasing at the left tail, the effective treatment effect at the group's 5th percentile is a density-weighted average of $\delta$ evaluated at \emph{higher} within-type ranks, where the effect is much smaller.

The composition channel---whether FSP changes \emph{who gives birth}---is negligible throughout the distribution. This is consistent with food stamps affecting nutrition rather than fertility decisions.

While the model in \eqref{eq:mp_model} imposes a common $\delta(u)$ across types, we can run the \MP estimator separately within each of four demographic cells (mother age $<24$/$\geq 24$ $\times$ legitimate/illegitimate) to examine heterogeneity. Figure~\ref{fig:fsp_type_qte} shows the type-specific conditional QTEs $\delta_k(u)$. At the 5th percentile, the effects range from $60$ grams for older married mothers to $-35$ grams for younger married mothers, with younger unmarried mothers---the dominant type at the group's left tail---showing an intermediate effect of approximately $40$ grams.

The overall $\delta(0.05) \approx 26$ grams is thus an average across heterogeneous type-specific effects, implicitly weighted by each type's conditional density at the group quantile cutoff.  Types that are more concentrated at the left tail of the birth weight distribution---young unmarried mothers, who account for 46\% of the density weight at the 5th percentile but only 36\% of births---receive disproportionate weight. This density-weighting mechanism is inherent to any estimand based on conditional quantiles evaluated at a common quantile index.

\begin{figure}[htbp]
\centering
\includegraphics[width=0.85\textwidth]{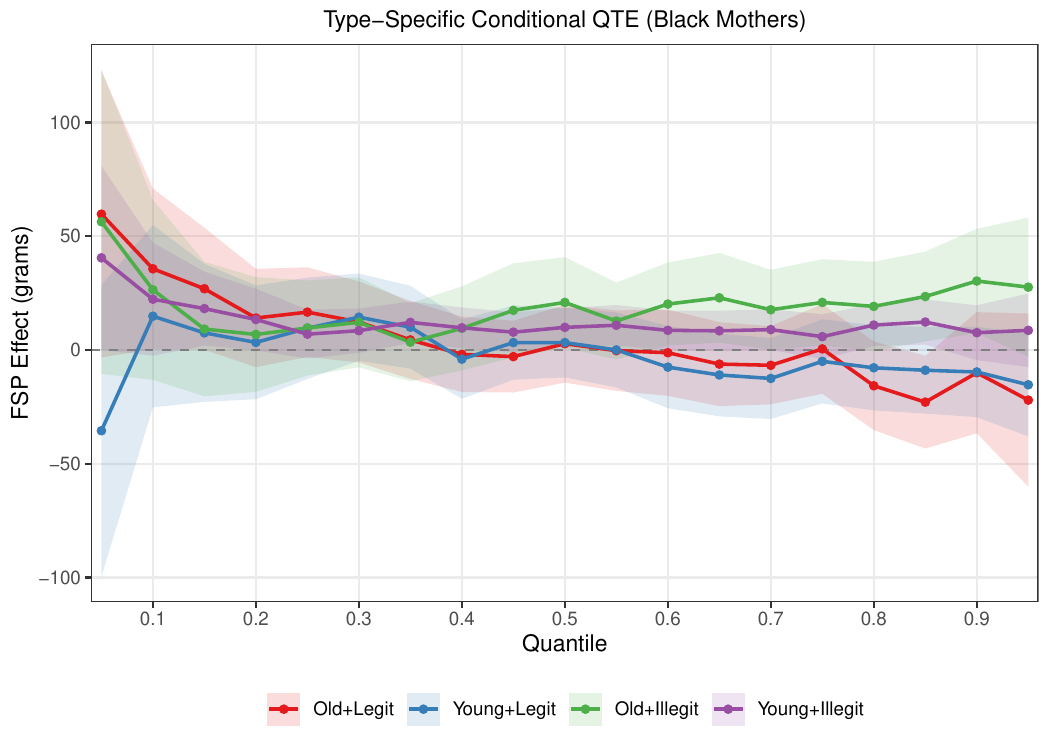}
\caption{Conditional QTEs hide large type heterogeneity at the left tail.}
\floatfoot{\footnotesize \textit{Notes:} Type-specific conditional QTEs $\delta_k(u)$ for black mothers, estimated by running \texttt{mdqr} separately within each of four demographic cells (mother age $<24$/$\geq 24$ $\times$ legitimate/illegitimate). Shaded bands are pointwise 95\% CIs clustered by county. }
\label{fig:fsp_type_qte}
\end{figure}

\paragraph{Discussion.}
The conditional quantile treatment effect is large and positive at the left
tail, while the aggregate distributional effect is close to zero and statistically
insignificant, because the types who drive the conditional estimate are not the types who are marginal at the group quantile cutoff. Young unmarried mothers---who comprise 36\% of births but 46\% of the probability density weight across types at the group's 5th percentile---drive the conditional estimate with a type-specific effect of approximately 40 grams, while the aggregate effect on the county's actual 5th percentile is indistinguishable from zero.

The two estimands answer different questions. The conditional effect $\delta(0.05) \approx 26$ grams documents that FSP delivers large nutritional benefits to the most vulnerable births, holding observed maternal type fixed. The \est estimate $\beta_1(0.05) \approx -7$ grams (insignificant) shows that these within-type gains do not translate to detectable shifts in the left tail of the county's birth-weight distribution. A county health department tracking the actual low-birth-weight rate---rather than conditional quantiles within demographic cells---would not detect a significant effect of FSP on the left tail of the birth weight distribution. This is because the babies at the county's 5th percentile are predominantly of young unmarried mothers sitting at much higher within-type ranks (around their 12th--15th conditional percentile), where the treatment effect $\delta$ is only 6--10 grams.

More broadly, this application illustrates when the two approaches diverge. For treatments that operate at the individual level but are identified through group-level variation (food stamps, school vouchers, Medicaid expansions), the conditional effect $\delta(u)$ captures the individual-level mechanism, while $\beta_1(u)$ captures the aggregate distributional impact. The re-ranking channel---which is first-order whenever $\delta(u)$ varies with $u$ and types are heterogeneously distributed across the group's outcome distribution---can dramatically attenuate the aggregate effect even when the conditional effect is large. For treatments that operate at the group level, like the import competition shock in Section~\ref{sec:empirical_clp}, the group quantile $\beta_1(u)$ is the natural estimand, and the decomposition is not needed.

\section{Conclusion}

This paper develops IV Fr\'echet regression (\est), a framework for estimating the effect of endogenous group-level treatments on distribution-valued outcomes. The approach recasts grouped quantile IV regression as an instrumental-variables problem in Wasserstein space. This perspective yields a simple estimator: construct IV-weighted average quantile curves, project them onto the space of valid quantile functions, and recover coefficient functions by OLS.

The paper makes three main contributions. First, it provides an identification result showing that, under standard quantile IV conditions, the structural distributional effect is the solution to an IV-weighted Fr\'echet problem. This gives the fitted object a clear interpretation as an instrumented average distribution. Second, it introduces a monotone projection step that guarantees valid fitted distributions and weakly improves finite-sample estimation error, while leaving the first-order asymptotic distribution unchanged under mild conditions. Third, it establishes functional asymptotic normality and multiplier-bootstrap procedures for pointwise and novel uniform inference over quantile indices.

Simulations and two empirical applications illustrate the practical value of the method. In finite samples, the projection can substantially reduce the integrated mean squared error (IMSE) relative to existing grouped quantile IV estimators. In an application to Chinese import competition, the method delivers tighter confidence bands and lower IMSE. Additionally, using our novel uniform confidence bands, we show that the evidence for wage losses is concentrated away from the very bottom of the distribution. In a second application to the effect of county-level food stamp programs on the birth weight distribution, we find no evidence for distributional effects using our uniform bands. More broadly, our results suggest that directly modeling outcomes as random distributions can sharpen both estimation and inference in settings where policy effects are inherently distributional.

\newpage

\bibliography{references} 
\bibliographystyle{apalike} 

\appendix
\clearpage

\section{The projection operator}

Here, we provide more formal details about the projection operator $\Pi_{\mathcal{Q}}$. Let $L^2([0,1])$ be the Hilbert space of square-integrable functions on the unit interval, equipped with the standard inner product $\langle f, g \rangle = \int_0^1 f(u) g(u) \dd u$ and norm $\|f\|_{L^2} = \sqrt{\langle f, f \rangle}$. Let $\calQ \subset L^2([0,1])$ be the subset of functions that are non-decreasing. $\calQ$ is a closed convex cone in $L^2([0,1])$.

The projection operator $\Pi_{\mathcal{Q}}: L^2([0,1]) \to \calQ$ maps any function $f \in L^2([0,1])$ to its unique closest element in $\calQ$ under the $L^2$ norm:
\[ \Pi_{\mathcal{Q}}(f) := \argmin_{h \in \calQ} \| f - h \|_{L^2}^2 = \argmin_{h \in \calQ} \int_0^1 (f(u) - h(u))^2 \dd u. \]
This projection is well-defined and unique \citep{rychlik2012projecting}, see also \citet[Proof of Proposition 2]{petersen2021wasserstein}. Computationally, for a function evaluated on a grid, $\Pi_{\mathcal{Q}}(f)$ can be computed using the Pool Adjacent Violators Algorithm (PAVA) \citep{ayer1955empirical, miles1959complete, kruskal1964nonmetric}. If $f \in \calQ$, then $\Pi_{\mathcal{Q}}(f) = f$.

\section{Decomposition under an individual-level model}
\label{app:individual_decomposition}

This appendix develops a formal decomposition of the group-level treatment effect $\beta_1(u)$ into direct, composition, and re-ranking components, under an individual-level quantile model, as in \CLP and \MP. The total causal effect $\beta_1(u) = E[Q_{Y(1)}(u) - Q_{Y(0)}(u)]$ was defined in Section~\ref{sec:identification_interpretation}, and we assume the treatment $X_j$ is binary for ease of exposition.

\subsection{Individual-level model}

To understand the channels through which group-level treatment affects the outcome distribution, consider a model at the level of within-group individuals, as in \CLP and \MP.

Let individuals $i$ be nested within groups $j$, and let $W_{ij} \in \mathbb{R}^d$ denote observed individual characteristics. Suppose individual outcomes satisfy:
\begin{equation}\label{eq:micro}
 Q_{Y_{j}}(u \mid W_{ij}, X_j) = W_{ij}^{\T} \gamma(u) + X_j^{\T} \delta(u) + \alpha_j(u), \qquad \E[\alpha_j(u) \mid W_{ij}, X_j] = 0,
\end{equation}
where $\gamma(u)$ is the coefficient on individual characteristics, $\delta(u)$ is the direct effect of group treatment on the quantile conditional on those characteristics, and $\alpha_j(u)$ is group-level unobserved heterogeneity. By the Skorohod representation \citep{skorokhod1956limit}, this implies a model for the scalar-valued outcome $Y_{ij}$ of individual $i$ in group $j$,
\begin{equation} \label{eq:skorohod}
Y_{i j}=W_{i j}^{\mathrm{T}} \gamma\left(U_{i j}\right)+X_j^{\mathrm{T}} \delta\left(U_{i j}\right)+\alpha_j\left(U_{i j}\right),
\end{equation}
where $U_{ij} \sim \mathrm{Uniform}[0,1]$.
Denote the distribution of individual characteristics within group $j$ as $H_j(\cdot; X_j)$, with mean $\bar{W}_j(X_j) := \int w \, \dd H_j(w; X_j)$.

A key feature of this setup is that the mixing distribution $H_j$ may itself depend on $X_j$. Treatment can affect who is in each group through hiring, attrition, migration, or sorting.

In what follows, we decompose $\beta_1(u)$ into a direct effect $\delta(u)$ and two composition terms. To interpret the decomposition economically, we consider the setting where the group is a commuting zone, the individual is a worker, and the outcome distribution is the local wage distribution.

\subsection{Two quantile concepts}

Two distinct group-level quantile objects arise from aggregating the individual conditional quantile model. Let observable worker ``types'' be indexed by $w$, with type distribution $H_j(w;X)$ in group $j$ under treatment state $X\in\{0,1\}$. Assume the linear conditional quantile specification
\begin{equation}\label{eq:indiv_q}
Q_{Y_{j}}(u\mid w,X)=w^{\T}\gamma(u)+X\,\delta(u)+\alpha_j(u),
\end{equation}
For ease of exposition, assume for the moment that $\alpha_j(u)$ is exogenous to treatment assignment (random assignment at the group level); Section~\ref{sec:iv_decomposition} discusses the IV analogue.

\paragraph{The composition-fixed quantile.}
Following \citet{petersen2021wasserstein}, the Wasserstein--Fr\'echet integral (WFI) quantile averages conditional \emph{quantile functions} across types:
\begin{equation}\label{eq:wfi_new3}
Q_{Y_j}^{\oplus}(u;X)
:=\int Q_{Y_{j}}(u\mid w,X)\,\dd H_j(w;X)
=\bar W_j(X)^{\T}\gamma(u)+X\,\delta(u)+\alpha_j(u),
\end{equation}
where $\bar W_j(X):=\int w\,\dd H_j(w;X)$. This object answers: \emph{what is the average type-specific $u$-quantile, holding the type composition fixed at $H_j(\cdot;X)$?}

\paragraph{The realized group quantile.}
The observed group quantile is the quantile of the mixture distribution:
\begin{equation}\label{eq:group_new3}
Q_{Y_j}(u;X)
:=\inf\Big\{y:\int F_{Y_{j}}(y\mid w,X)\,\dd H_j(w;X)\ge u\Big\},
\end{equation}
where $F_{j}$ is the CDF of the individual-level outcome $Y_{ij}$ in \eqref{eq:skorohod}, i.e., the inverse of the quantile function $Q_{Y_{j}}$. The group-level quantile function, $Q_{Y_j}(u; X)$, is what the researcher observes: \emph{the unconditional $u$-quantile of outcomes in group $j$.}

Since quantiles and mixing do not commute, define the gap
\begin{equation}\label{eq:Delta_new3}
\Delta_j(u;X):=Q_{Y_j}(u;X)-Q_{Y_j}^{\oplus}(u;X).
\end{equation}
A convenient representation uses the re-ranking map,
\begin{equation}\label{eq:rerank_map_new3}
t_{jX}(w;u)\;:=\;F_{Y_{j}}\!\big(Q_{Y_j}(u;X)\mid w,X\big)\in[0,1],
\end{equation}
the within-type percentile attained by type $w$ at the \emph{group} $u$-quantile cutoff. Let $Q_{Y_{j}}(\cdot\mid w,X)$ denote the conditional quantile function (inverse CDF). Then
\begin{equation}\label{eq:Delta_as_D_new3}
\Delta_j(u;X)=\int D_{jX}(u\mid w)\,\dd H_j(w;X),
~~
D_{jX}(u\mid w):=
Q_{Y_{j}}\!\big(t_{jX}(w;u)\mid w,X\big)-Q_{Y_{j}}(u\mid w,X).
\end{equation}
Thus the gap is the average \emph{within-type percentile shift} required for each type to hit the common group cutoff.

\subsection{Treatment effect decomposition}

Let $\beta(u)$ be the slope from the population projection of the realized group quantile $Q_{Y_j}(u;X_j)$ onto $(1,\tilde X_j)$, with $\tilde X_j:=X_j-\mu_X$ and $\Sigma_{XX}:=\E[\tilde X_j^2]$:
\begin{equation}\label{eq:beta_proj_new3}
\beta(u):=\Sigma_{XX}^{-1}\E\!\left[\tilde X_j\,Q_{Y_j}(u;X_j)\right].
\end{equation}
Using the identity
\begin{equation}\label{eq:Q_split_new3}
Q_{Y_j}(u;X)=Q_{Y_j}^{\oplus}(u;X)+\Delta_j(u;X),
\end{equation}
together with the closed form for $Q_{Y_j}^{\oplus}$ in \eqref{eq:wfi_new3}, we obtain the decomposition
\begin{equation}\label{eq:beta_decomp_new3}
\beta_1(u)
=
\underbrace{\delta(u)}_{\text{Direct (within-type)}}
+
\underbrace{\Sigma_{XX}^{-1}\E\!\left[\tilde X_j\,\bar W_j(X_j)^{\T}\gamma(u)\right]}_{\text{Mean composition / sorting}}
+
\underbrace{\Sigma_{XX}^{-1}\E\!\left[\tilde X_j\,\Delta_j(u;X_j)\right]}_{\text{Re-ranking / aggregation}}
.
\end{equation}
Equation \eqref{eq:beta_decomp_new3} is the form that extends immediately to IV: in the IV case, the same terms appear with $\tilde X_j$ replaced by the instrument-induced projection of $X_j$ (and the corresponding $\Sigma_{XX}^{-1}$ replaced by the usual 2SLS matrix), exactly as in 2SLS.

Given the binary treatment, \eqref{eq:beta_decomp_new3} is equivalent to a decomposition in treatment effects. Specifically, for any square-integrable function $g(X_j)$,
\begin{equation}\label{eq:binary_proj_identity}
\Sigma_{XX}^{-1}\E\!\left[\tilde X_j\,g(X_j)\right]
=
\E\!\left[g(1)-g(0)\right],
\end{equation}
since $\Sigma_{XX}=\mu_X(1-\mu_X)$ and $\E[\tilde X_j g(X_j)]=\mu_X(1-\mu_X)\{g(1)-g(0)\}$ with $X_j$ binary. Applying \eqref{eq:binary_proj_identity} yields
\begin{equation}\label{eq:TE_beta_new3}
\beta_1(u)
=
\underbrace{\delta(u)}_{\text{Direct (within-type)}}
+
\underbrace{\E\!\left[\bar W_j(1)-\bar W_j(0)\right]^{\T}\gamma(u)}_{\text{Mean composition / sorting}}
+
\underbrace{\E\!\left[\Delta_j(u;1)-\Delta_j(u;0)\right]}_{\text{Re-ranking / aggregation (TE on the gap)}}
.
\end{equation}

This decomposition separates three economically distinct channels by which treatment changes a group's realized wage quantile $Q_{Y_j}(u;X)$.
\begin{enumerate}
\item \textbf{Direct (within-type) effect $\delta(u)$.} This is the partial-equilibrium shift in the $u$-th \emph{conditional} quantile for a fixed type. For example, holding skill group fixed (high-skill vs.\ low-skill), $\delta(u)$ is the average change in the $u$-th wage quantile within each skill group when the group is treated.

\item \textbf{Mean composition / sorting $\E[\bar W_j(1)-\bar W_j(0)]^{\T}\gamma(u)$.} Treatment can change who is in the group, through, for example, entry/exit or sorting, so the treated group may become more high-skill on average. This term prices that workforce shift using $\gamma(u)$, i.e., the ``return to type'' at percentile $u$.

\item \textbf{Re-ranking / aggregation $\E[\Delta_j(u;1)-\Delta_j(u;0)]$.} Even if the group's average type were unchanged, the group's $u$-quantile depends on \emph{which within-type ranks} are marginal at the group cutoff. If high-skill wages have a thicker upper tail, the group's 90th percentile may correspond to a lower within-high-skill percentile than within-low-skill. Treatment can change this mapping (the $t_{jX}(w;u)$ functions) by changing dispersion/tails and/or changing the mix of types, which moves realized quantiles beyond what $\delta(u)$ captures.
\end{enumerate}

The re-ranking effect can be further decomposed. Abbreviating $t_{jX}(w;u)$ by $t_X$, it is simply the integral of \eqref{eq:indiv_q} over $H_j(\cdot;X)$ across all types $w$,
\begin{equation}\label{eq:Delta_decomp_binary_new4}
\begin{aligned}
\Delta_j(u;X)
&=
\E_{H_j(\cdot;X)}\!\Big[
w^{\T}\big(\gamma(t_{jX}(w;u))-\gamma(u)\big)
\Big]
+
X\,\E_{H_j(\cdot;X)}\!\Big[
\delta(t_{jX}(w;u))-\delta(u)
\Big]
\\
&\qquad\qquad
+
\E_{H_j(\cdot;X)}\!\Big[
\alpha_j(t_{jX}(w;u))-\alpha_j(u)
\Big].
\end{aligned}
\end{equation}

The gap $\Delta_j(u;X)$ is the average \emph{within-type percentile shift} required for each observable type to reach the common group cutoff $Q_{Y_j}(u;X)$.
If high-skill and low-skill workers have differently shaped conditional wage distributions (e.g., if the distribution of high-skilled workers has a thicker upper tail), then the workers who are marginal at the group's $u$-th quantile are not at the same within-type quantile across types: $t_{jX}(w;u)\neq u$ in general.
Equation \eqref{eq:Delta_decomp_binary_new4} makes clear that re-ranking affects the realized group quantile only through \emph{shape}, i.e., through how $\gamma(\cdot)$, $\delta(\cdot)$, and $\alpha_j(\cdot)$ vary with the quantile index.

Treatment can affect this rearrangement because $X$ changes the mixture distribution itself: it can shift the type distribution $H_j(\cdot;X)$ (entry/exit/sorting) and it can change the mapping $w\mapsto t_{jX}(w;u)$ (which within-type quantiles are pivotal at the group cutoff).
Thus, even holding fixed the within-type effect $\delta(u)$, the same non-commutativity mechanism that generates $\Delta_j(u;X)$ in any state will generally produce different gaps under $X=1$ and $X=0$.

The decomposition further shows that the re-ranking channel is first-order only to the extent that there are shape effects: it is amplified when (i) returns to type vary over the distribution ($u\mapsto\gamma(u)$ is not flat), (ii) treatment effects vary over quantiles ($u\mapsto\delta(u)$ is not flat), and/or (iii) the latent group quantile function $\alpha_j(u)$ is curved.

\subsection{IV interpretation} \label{sec:iv_decomposition}

When we relax the exogeneity assumption in favor of the instrument exogeneity assumption $E\left[Z_j \alpha_j(u)\right]=0$, the IV estimand admits a similar decomposition,
\begin{equation}\label{eq:beta-iv-decomp}
\beta_1^{\mathrm{IV}}(u)
=
\delta(u)
+
\Pi_{Z\to X}\E\!\left[\tilde Z_j\,\bar W_j(X_j)^{\T}\gamma(u)\right]
+
\Pi_{Z\to X}\E\!\left[\tilde Z_j\,\Delta_j(u;X_j)\right]
\end{equation}
where $\Pi_{Z\to X}:=(\Sigma_{XZ}\Sigma_{ZZ}^{-1}\Sigma_{ZX})^{-1}\Sigma_{XZ}\Sigma_{ZZ}^{-1}$ is the usual population 2SLS operator.

For the group-level IV regression model in \eqref{eq:linear_model}, we can write,
\begin{equation}\label{eq:eta_def}
\eta_j(u):=\alpha_j(u)+r^{\text{mean}}_j(u)+r^{\text{shape}}_j(u).
\end{equation}
where
\begin{align}
r^{\text{mean}}_j(u)
&:=
\bar W_j(X_j)^{\T}\gamma(u)
-
\Big(\Pi_{Z\to X}\E[\tilde Z_j\,\bar W_j(X_j)^{\T}\gamma(u)]\Big)X_j,
\label{eq:r_mean_def}
\\
r^{\text{shape}}_j(u)
&:=
\Delta_j(u;X_j)
-
\Big(\Pi_{Z\to X}\E[\tilde Z_j\,\Delta_j(u;X_j)]\Big)X_j,
\label{eq:r_shape_def}
\end{align}
the parts of the composition and re-ranking terms that are not predicted by the instrument.
Since these are projections, in the just-identified case, these residuals satisfy
$\E[\tilde Z_j r^{\text{mean}}_j(u)]=\E[\tilde Z_j r^{\text{shape}}_j(u)]=0$ by construction.

As a result, $\E[\tilde Z_j\eta_j(u)]=\E[\tilde Z_j\alpha_j(u)]$, so the IV moment condition
$\E[\tilde Z_j\eta_j(u)]=0$ is equivalent to,
\begin{equation}\label{eq:iv_exogeneity_alpha}
\E[\tilde Z_j\alpha_j(u)]=0.
\end{equation}
In the overidentified case, the same equivalence holds for the population 2SLS
normal equations,
\[
\Sigma_{XZ}\Sigma_{ZZ}^{-1}\E[\tilde Z_j\eta_j(u)]
=
\Sigma_{XZ}\Sigma_{ZZ}^{-1}\E[\tilde Z_j\alpha_j(u)].
\]
In other words, the instrument may induce changes in workforce composition $\bar W_j(X_j)$ and in re-ranking $\Delta_j(u;X_j)$---the equilibrium channels captured by \eqref{eq:beta-iv-decomp}---and validity hinges only on the standard restriction that $Z_j$ is orthogonal to the latent group component $\alpha_j(u)$. Moreover, any group-level instrument valid for the individual-level model in Eq. \eqref{eq:indiv_q} is also valid for the group-level model in Eq. \eqref{eq:linear_model}, and vice versa.

\section{Computational benchmarks}\label{sec:computation}
Both \CLP and \MP estimate the coefficient function $\beta_1(u)$ by running a separate regression for each quantile level $u$. \MP solve $n \times Q$ within-group quantile regressions, where $Q$ is the number of points in the quantile grid, in the first stage and then perform a second-stage GMM estimation. \CLP similarly estimate within-group quantile regressions and then run $Q$ independent 2SLS regressions on the group-level intercept. When there are no individual-level covariates, our approach avoids quantile regression entirely: it computes sample quantile functions directly by sorting within each group, then solves the 2SLS normal equations for all $Q$ quantile levels simultaneously in a single matrix operation. The key source of the speedup is this vectorization across the quantile grid---our second stage is one matrix solve rather than $Q$ separate regressions.

Table~\ref{tab:timing} reports median computation times across 20 replications for 19 quantile levels, with no individual-level covariates. At $n = 100$ groups and $N = 200$ individuals per group, \MP takes 5.4 seconds, \CLP takes 41 milliseconds, and \est (including the sample quantile computation) takes 7 milliseconds---a speedup of roughly $800\times$ over \MP and $6\times$ over \CLP. The \MP computation time grows with both $n$ and $N$ because it solves $n \times Q$ quantile regression problems on $N$ observations each; \CLP grows only with $n$ (its $Q$ regressions are on $n$ group-level observations); and our method is nearly insensitive to $N$ beyond the initial sort. At $n = 500$, $N = 1{,}000$ (500K total observations), \MP takes 22.7 seconds while \est takes 65 milliseconds, a $349\times$ speedup.

This speed difference matters for bootstrap inference. The multiplier bootstrap in Section~\ref{sec:asymptotic_distribution} requires $B$ evaluations of the estimator ($B = 2{,}000$ in our empirical application). With \est this takes seconds; repeating the \MP or \CLP estimation $B$ times would take minutes to hours.

\begin{table}[htbp]
\centering
\caption{Computational cost of \est relative to existing estimators.}
\label{tab:timing}
\small
\begin{tabular}{rr r rrr rr}
\toprule
$n$ & $N$ & Total obs & \MP & \CLP & 2SLS & \est & Speedup \\
\midrule
50 & 50 & 2,500 & 1,739 & 24 & 2 & 3 & $696\times$ \\
50 & 1,000 & 50,000 & 2,241 & 23 & 4 & 5 & $448\times$ \\
100 & 200 & 20,000 & 5,399 & 41 & 5 & 7 & $771\times$ \\
100 & 1,000 & 100,000 & 7,752 & 36 & 10 & 11 & $705\times$ \\
200 & 200 & 40,000 & 4,157 & 27 & 7 & 10 & $416\times$ \\
200 & 1,000 & 200,000 & 5,747 & 25 & 15 & 18 & $319\times$ \\
500 & 200 & 100,000 & 8,895 & 26 & 19 & 25 & $356\times$ \\
500 & 1,000 & 500,000 & 22,668 & 34 & 60 & 65 & $349\times$ \\
\bottomrule
\end{tabular}
\floatfoot{\footnotesize \textit{Notes:} Median computation time (milliseconds) for 19 quantile levels, no individual-level covariates; median of 20 replications. \MP: within-group quantile regression + second-stage 2SLS. \CLP: quantile-by-quantile 2SLS. 2SLS and \est: sample quantiles + vectorized matrix 2SLS (with PAVA projection for \est). Speedup column: \MP time / \est time. All timings are single-threaded on an Apple M1 Pro with 16\,GB RAM. The \est column includes sample quantile computation and the PAVA projection; the 2SLS column includes sample quantiles but not the projection.}
\end{table}

\clearpage
\section{Proofs}\label{sec:proofs_appendix}

\begin{proof}[Proof of Proposition \ref{prop:frechet_equiv}]
Using the quantile representation of the 2-Wasserstein distance, $W_2^2(\mu,\nu) = \int_0^1 (Q_\mu(u) - Q_\nu(u))^2\, du$, the IV-weighted Fr\'echet functional at candidate $w$ is
\begin{align} \label{eq:iv-weighted-frechet}
\E\!\left[s(Z,x)\, W_2^2(Y, w)\right]
&= \E\!\left[s(Z,x) \int_{(0,1)} \big(Q_Y(u) - Q_w(u)\big)^2\, du\right].
\end{align}
By the assumed finiteness of the IV-weighted Fr\'echet functional,
\[
E\!\left[|s(Z,x)|\,\|Q_Y\|_{L^2(0,1)}^2\right]<\infty,
\]
and since $Q_w\in L^2(0,1)$, Fubini's theorem allows us to exchange expectation
and integration:
\begin{align*}
\eqref{eq:iv-weighted-frechet} &= \int_{(0,1)} \E\!\left[s(Z,x)\big(Q_Y(u)^2 - 2 Q_Y(u) Q_w(u) + Q_w(u)^2\big)\right] du \\
&= \int_{(0,1)} \left\{\E[s(Z,x) Q_Y(u)^2] - 2 Q_w(u)\, \psi_x(u) + Q_w(u)^2\right\} du,
\end{align*}
where we used $\E[s(Z,x)] = 1$ and $\psi_x(u) = \E[s(Z,x) Q_Y(u)]$. Assumption~\ref{ass:full_rank} ensures that $s(Z,x)$ is well-defined, and
$E[s(Z,x)]=1$ because $E[\tilde Z]=0$. Completing the square,
\begin{align*}
&= \underbrace{\int_{(0,1)} \left\{\E[s(Z,x) Q_Y(u)^2] - \psi_x(u)^2\right\} du}_{=:\, C(x),\;\text{does not depend on }w} + \int_{(0,1)} \big(Q_w(u) - \psi_x(u)\big)^2\, du.
\end{align*}
The first term $C(x)$ does not depend on $w$, so the minimizer over $w \in \calP$ minimizes $\|Q_w - \psi_x\|_{L^2}^2$ over $Q_w \in \mathcal{Q}$, which is $\Pi_{\mathcal{Q}}(\psi_x)$ by definition of the $L^2$ projection. If $\psi_x$ is itself a valid quantile function, then $\Pi_{\mathcal{Q}}(\psi_x) = \psi_x$ and it is the a.e.-unique minimizer in $L^2$.
\end{proof}

\begin{proof}[Proof of Lemma \ref{lem:identification_FIVR}]

By Assumption~\ref{ass:instrument_exogeneity},
\[
\E[\tilde Z_j\eta_j(u)]=0.
\]
Together with Assumption~\ref{ass:full_rank}, this gives
\[
\psi_x(u)
=
\E[s(Z_j,x)Q_{Y_j}(u)]
=
\beta_0(u)\E[s(Z_j,x)]
+
\beta_1(u)^{\!\T}\E[s(Z_j,x)(X_j-\mu_X)]
+
\E[s(Z_j,x)\eta_j(u)].
\]
We evaluate each term in turn.
\begin{enumerate}
\item Since \( \E[\tilde{Z}_j] = 0 \), we have \( \E[s(Z_j, x)] = 1 \).

 \item For the second term,
  \[
  \E[s(Z_j, x) (X_j - \mu_X)] = (x - \mu_X)^{\!\T} (\Sigma_{ZX}^{\!\T} \Sigma_{ZZ}^{-1} \Sigma_{ZX})^{-1} \Sigma_{ZX}^{\!\T} \Sigma_{ZZ}^{-1} \E[\tilde{Z}_j (X_j - \mu_X)] = x - \mu_X,
  \]
  as \( \E[\tilde{Z}_j (X_j - \mu_X)] = \Sigma_{ZX} \) and \( (\Sigma_{ZX}^{\!\T} \Sigma_{ZZ}^{-1} \Sigma_{ZX})^{-1} \Sigma_{ZX}^{\!\T} \Sigma_{ZZ}^{-1} \Sigma_{ZX} = I_p \).
\item For the third term, \( \E[s(Z_j, x) \eta_j(u)] = \E[\eta_j(u)] + (x - \mu_X)^{\!\T} (\dots) \E[\tilde{Z}_j \eta_j(u)] = 0 \), since \( \E[\eta_j(u)] = 0 \) and \( \E[\tilde{Z}_j \eta_j(u)] = 0 \).
\end{enumerate}
As a result, 
\[
\psi_x(u) = \beta_0(u) + \beta_1(u)^{\!\T} (x - \mu_X) = q(x, u).
\]
Since \( q(x, u) \) is non-decreasing in \( u \) by construction, $\psi_x(\cdot) = q(x, \cdot)$ is a valid quantile function for all $x \in \mathcal{X}$.
\end{proof}

\begin{proof}[Proof of Theorem~\ref{thm:pw_ols_improvement}]
Let $b = (b_0, b_1)$ be any reference coefficients such that $q_b(X_j, \cdot) := b_0(\cdot) + b_1(\cdot)^{\!\T}(X_j - \hat{\mu}_X) \in \mathcal{Q}$ for all $j = 1, \ldots, n$.

\medskip
Since $q_b(X_j, \cdot) \in \mathcal{Q}$ for each $j$, Lemma~\ref{lemma:improvement} with $p = 2$ gives
\[
\int_0^1 \big|\hat{Q}(X_j, u) - q_b(X_j, u)\big|^2\, \dd u \;\leq\; \int_0^1 \big|\hat{\psi}_{X_j}(u) - q_b(X_j, u)\big|^2\, \dd u.
\]

Averaging over $j = 1, \ldots, n$ and exchanging the order of summation and integration:
\begin{equation}\label{eq:avg_improvement_pf}
\int_0^1 \frac{1}{n}\sum_{j=1}^n \big|\hat{Q}(X_j, u) - q_b(X_j, u)\big|^2\, \dd u \;\leq\; \int_0^1 \frac{1}{n}\sum_{j=1}^n \big|\hat{\psi}_{X_j}(u) - q_b(X_j, u)\big|^2\, \dd u.
\end{equation}

Fix $u \in [0,1]$. Let $\hat{\mathbf{X}}$ denote the $n \times (p+1)$ design matrix with $j$-th row $(1,\; (X_j - \hat{\mu}_X)^{\!\T})$. The \est coefficients are defined by $\hat{\beta}^{\est}(u) = (\hat{\mathbf{X}}^{\!\T}\hat{\mathbf{X}})^{-1}\hat{\mathbf{X}}^{\!\T} \hat{Q}(u)$, where $\hat{Q}(u) = (\hat{Q}(X_1,u),\ldots,\hat{Q}(X_n,u))^{\!\T}$. Write the OLS fitted values as $\hat{q}^{\est}(X_j, u) = \hat{\beta}_0^{\est}(u) + \hat{\beta}_1^{\est}(u)^{\!\T}(X_j - \hat{\mu}_X)$ and the residuals as $e_j(u) = \hat{Q}(X_j, u) - \hat{q}^{\est}(X_j, u)$.

Both $\hat{q}^{\est}(X_j, u)$ and $q_b(X_j, u)$ are linear in $(1, (X_j - \hat{\mu}_X)^{\!\T})$, so their difference $\hat{q}^{\est}(X_j, u) - q_b(X_j, u)$ lies in the column space of $\hat{\mathbf{X}}$, while $e_j(u)$ lies in its orthogonal complement by the OLS normal equations. The Pythagorean theorem therefore gives
\begin{equation}\label{eq:pythag_pf}
\frac{1}{n}\sum_{j=1}^n \big|\hat{Q}(X_j, u) - q_b(X_j, u)\big|^2 = \frac{1}{n}\sum_{j=1}^n \big|\hat{q}^{\est}(X_j, u) - q_b(X_j, u)\big|^2 + \frac{1}{n}\sum_{j=1}^n e_j(u)^2.
\end{equation}

Since $\frac{1}{n}\sum_{j=1}^n (X_j - \hat{\mu}_X) = 0$ by construction,  $\frac{1}{n}\hat{\mathbf{X}}^{\!\T}\hat{\mathbf{X}}$ is block-diagonal,
\[
\frac{1}{n}\hat{\mathbf{X}}^{\!\T}\hat{\mathbf{X}} = \begin{pmatrix} 1 & 0^{\!\T} \\ 0 & \hat{\Sigma}_{XX} \end{pmatrix}.
\]
Writing $\delta_0(u) = \hat{\beta}_0^{\est}(u) - b_0(u)$ and $\delta_1(u) = \hat{\beta}_1^{\est}(u) - b_1(u)$, the cross term vanishes:
\begin{align*}
\frac{1}{n}\sum_{j=1}^n \big|\hat{q}^{\est}(X_j, u) - q_b(X_j, u)\big|^2
&= \frac{1}{n}\sum_{j=1}^n \big|\delta_0(u) + \delta_1(u)^{\!\T}(X_j - \hat{\mu}_X)\big|^2 \\
&= \delta_0(u)^2 + 2\,\delta_0(u)\,\delta_1(u)^{\!\T}\underbrace{\frac{1}{n}\sum_{j=1}^n(X_j - \hat{\mu}_X)}_{=\,0} + \delta_1(u)^{\!\T}\hat{\Sigma}_{XX}\,\delta_1(u) \\
&= \delta_0(u)^2 + \delta_1(u)^{\!\T}\hat{\Sigma}_{XX}\,\delta_1(u).
\end{align*}

Since $\frac{1}{n}\sum_j e_j(u)^2 \geq 0$, combining~\eqref{eq:pythag_pf} and the above gives, for each $u$,
\begin{equation}\label{eq:lhs_lower_pf}
\big|\hat{\beta}_0^{\est}(u) - b_0(u)\big|^2 + \big(\hat{\beta}_1^{\est}(u) - b_1(u)\big)^{\!\T}\hat{\Sigma}_{XX}\big(\hat{\beta}_1^{\est}(u) - b_1(u)\big) \;\leq\; \frac{1}{n}\sum_{j=1}^n \big|\hat{Q}(X_j, u) - q_b(X_j, u)\big|^2.
\end{equation}

Since $\hat{\psi}_{X_j}(u) = \tilde{\beta}_0(u) + \tilde{\beta}_1(u)^{\!\T}(X_j - \hat{\mu}_X)$ is exactly linear in $\hat{\mathbf{X}}$, regressing $(\hat{\psi}_{X_1}(u), \ldots, \hat{\psi}_{X_n}(u))^{\!\T}$ on $\hat{\mathbf{X}}$ recovers $(\tilde{\beta}_0(u), \tilde{\beta}_1(u))$ with zero residuals. The same block-diagonal expansion therefore gives
\[
\frac{1}{n}\sum_{j=1}^n \big|\hat{\psi}_{X_j}(u) - q_b(X_j, u)\big|^2 = \big|\tilde{\beta}_0(u) - b_0(u)\big|^2 + \big(\tilde{\beta}_1(u) - b_1(u)\big)^{\!\T}\hat{\Sigma}_{XX}\big(\tilde{\beta}_1(u) - b_1(u)\big).
\]

To finish, integrate~\eqref{eq:lhs_lower_pf} over $u \in [0,1]$ and apply~\eqref{eq:avg_improvement_pf} to obtain,
\[
\|\hat{\beta}_0^{\est} - b_0\|_{L^2}^2 + \|\hat{\beta}_1^{\est} - b_1\|_{\hat{\Sigma}_{XX}, L^2}^2 \;\leq\; \int_0^1 \frac{1}{n}\sum_{j=1}^n \big|\hat{\psi}_{X_j}(u) - q_b(X_j, u)\big|^2\, \dd u = \|\tilde{\beta}_0 - b_0\|_{L^2}^2 + \|\tilde{\beta}_1 - b_1\|_{\hat{\Sigma}_{XX}, L^2}^2,
\]
where the equality integrates the previous equation over $u$, which gives Eq.~\eqref{eq:pw_improvement}.
\end{proof}

\begin{proof}[Proof of Proposition \ref{prop:coord_improvement}]
By FWL, for each $u$ the $k$-th slope coefficient from regressing \[(\hat{Q}(X_1,u),\ldots,\hat{Q}(X_n,u))^{\!\T}\] on $(1_n,C)$ is
\[
\hat\beta_{1,k}^{\est}(u)
=\frac{1}{n\hat v_k}\sum_{j=1}^n r_{jk}\,\hat{Q}(X_j,u).
\]
Since $\hat\psi_{X_j}(u)=\tilde\beta_0(u)+(X_j-\hat\mu_X)^{\!\T}\tilde\beta_1(u)$ is exactly linear in $(1_n,C)$, the same identity gives $\tilde\beta_{1,k}(u)=\frac{1}{n\hat v_k}\sum_j r_{jk}\,\hat\psi_{X_j}(u)$, so
\begin{equation}\label{eq:Delta_k}
\Delta_k(u):=\hat\beta_{1,k}^{\est}(u)-\tilde\beta_{1,k}(u)=\frac{1}{n\hat v_k}\sum_{j=1}^n r_{jk}\,D_{X_j}(u).
\end{equation}

Fix $j$ with $r_{jk}\neq 0$. Since $\hat{Q}(X_j,\cdot)=\Pi_{\mathcal Q}(\hat\psi_{X_j})$ is the closest element of $\mathcal Q$ to $\hat\psi_{X_j}$ in $L^2$, it satisfies the first-order condition: for every $h\in\mathcal Q$,
\[
\langle D_{X_j},\; h-\hat{Q}(X_j,\cdot)\rangle_{L^2}\ge 0,
\]
where remember that $D_{X_j}=\hat{Q}(X_j,\cdot)-\hat\psi_{X_j}$ is the projection correction. Since $q_b(X_j,\cdot)\in\mathcal Q$ by assumption, taking $h=q_b(X_j,\cdot)$:
\[
\langle D_{X_j},\; q_b(X_j,\cdot)-\hat{Q}(X_j,\cdot)\rangle_{L^2}\ge 0.
\]
Substituting $q_b(X_j,\cdot)-\hat{Q}(X_j,\cdot)=-e_{jk}-r_{jk}(\tilde\beta_{1,k}-b_{1,k})-D_{X_j}$ from~\eqref{eq:decomp_identity}:
\begin{equation}\label{eq:vi_unconditional}
r_{jk}\,\langle D_{X_j},\,\tilde\beta_{1,k}-b_{1,k}\rangle_{L^2}
\;\le\;
-\|D_{X_j}\|_{L^2}^2 - \langle D_{X_j},\,e_{jk}\rangle_{L^2}.
\end{equation}
Summing over $j$ with $r_{jk}\neq 0$, dividing by $n\hat v_k$, and noting that $\tilde\beta_{1,k}-b_{1,k}$ does not depend on $j$:
\[
\Big\langle \frac{1}{n\hat v_k}\sum_{j:\,r_{jk}\neq 0} r_{jk}\,D_{X_j},\;\tilde\beta_{1,k}-b_{1,k}\Big\rangle_{L^2}
\;\le\;
-\frac{1}{n\hat v_k}\sum_{j:\,r_{jk}\neq 0}\|D_{X_j}\|_{L^2}^2
-\frac{1}{n\hat v_k}\sum_{j:\,r_{jk}\neq 0}\langle D_{X_j},\,e_{jk}\rangle_{L^2}.
\]
By~\eqref{eq:Delta_k}, the left-hand side equals $\langle\Delta_k,\,\tilde\beta_{1,k}-b_{1,k}\rangle_{L^2}$, giving~\eqref{eq:crossterm_general}.
\begin{equation}\label{eq:crossterm_general}
\langle\tilde\beta_{1,k}-b_{1,k},\;\Delta_k\rangle_{L^2}
\;\le\;
-\frac{1}{n\hat v_k}\sum_{j\in J_k}\|D_{X_j}\|_{L^2}^2
-\frac{1}{n\hat v_k}\sum_{j\in J_k}\langle D_{X_j},\,e_{jk}\rangle_{L^2}.
\end{equation}
By Cauchy--Schwarz applied to~\eqref{eq:Delta_k} and the definition of $\hat v_k :=  \frac1n \sum_{i=1}^n r_{jk}^2$:
\begin{equation}\label{eq:norm_bound}
\|\Delta_k\|_{L^2}^2\;\le\;\frac{1}{n\hat v_k}\sum_{j\in J_k}\|D_{X_j}\|_{L^2}^2.
\end{equation}
Since $\hat\beta_{1,k}^{\est} - b_{1,k} = (\tilde\beta_{1,k} - b_{1,k}) + \Delta_k$, expanding the squared norm gives
\[
\|\hat\beta_{1,k}^{\est}-b_{1,k}\|_{L^2}^2
= \|\tilde\beta_{1,k}-b_{1,k}\|_{L^2}^2
+ 2\langle\tilde\beta_{1,k}-b_{1,k},\,\Delta_k\rangle_{L^2}
+ \|\Delta_k\|_{L^2}^2.
\]
Substituting the bounds~\eqref{eq:crossterm_general} and~\eqref{eq:norm_bound} on the cross-term and the squared correction yields~\eqref{eq:coord_bound_general}.
\end{proof}

\begin{proof}[Proof of Corollary \ref{cor:coord_sufficient}]
Under~\eqref{eq:Ak_feas}, the point $q_b(X_j,\cdot)+e_{jk}(\cdot)=\hat\psi_{X_j}(\cdot)+r_{jk}(b_{1,k}(\cdot)-\tilde\beta_{1,k}(\cdot))\in\mathcal Q$ (by~\eqref{eq:decomp_identity}) provides a tighter feasible point in the variational inequality, yielding
\[
r_{jk}\langle D_{X_j},\,\tilde\beta_{1,k}-b_{1,k}\rangle_{L^2}
\;\le\;
-\|D_{X_j}\|_{L^2}^2
\]
without the $\langle D_{X_j}, e_{jk}\rangle$ term. The rest follows as in the proof of Proposition~\ref{prop:coord_improvement}.
\end{proof}

\begin{lemma}[Local smoothness of weights]\label{lem:smooth_g_from_fullrank}
Let $\theta=(\mu_X,\mu_Z,\Sigma_{ZX},\Sigma_{ZZ})$ and
\[
g(\theta,Z)\;=\;1+(x-\mu_X)^{\!\top}
\Big(\Sigma_{ZX}^{\!\top}\Sigma_{ZZ}^{-1}\Sigma_{ZX}\Big)^{-1}
\Sigma_{ZX}^{\!\top}\Sigma_{ZZ}^{-1}(Z-\mu_Z).
\]
If Assumption~\ref{ass:full_rank} holds at $\theta_0$, then there exists an open neighborhood $\mathcal N$ of $\theta_0$ such that:
\begin{enumerate}[(i)]
\item $\Sigma_{ZZ}(\theta)$ and $A(\theta):=\Sigma_{ZX}(\theta)^{\!\top}\Sigma_{ZZ}(\theta)^{-1}\Sigma_{ZX}(\theta)$ are invertible for all $\theta\in\mathcal N$;
\item for each fixed $Z$, the map $\theta\mapsto g(\theta,Z)$ is $C^2$ on $\mathcal N$;
\item for any compact $\mathcal K\subset\mathcal N$ there exist constants $C_0,C_1<\infty$ such that
\[
\sup_{\theta\in\mathcal K}\big\|\nabla^2_{\theta\theta} g(\theta,Z)\big\|
\;\le\; C_0+C_1\|Z\|\qquad\text{for all $Z$.}
\]
\end{enumerate}
\end{lemma}

\begin{proof}
Since $\Sigma_{ZZ}(\theta_0)$ is p.d.\ and $A(\theta_0)$ is p.d.\ by Assumption~\ref{ass:full_rank}, and eigenvalues depend continuously on matrix entries, there is an open neighborhood $\mathcal N$ of $\theta_0$ on which both matrices remain invertible, proving (i). 

The maps $\theta\mapsto \Sigma_{ZZ}(\theta)$ and $\theta\mapsto \Sigma_{ZX}(\theta)$ are affine in the components of $\theta$. Matrix addition and multiplication preserve smoothness of the entries, and inversion is smooth on the set of nonsingular square matrices. Therefore $\theta\mapsto g(\theta,Z)$ is a composition of smooth maps on $\mathcal N$, hence $C^2$, proving (ii).

Finally, write $g(\theta,Z)=a(\theta)+b(\theta)^{\!\top}Z$. By (ii), $a$ and $b$ are $C^2$ on $\mathcal N$, so on any compact $\mathcal K\subset\mathcal N$ the quantities
\[
M_0:=\sup_{\theta\in\mathcal K}\big\|\nabla^2 a(\theta)\big\|,
\qquad
M_1:=\sup_{k}\sup_{\theta\in\mathcal K}\big\|\nabla^2 b_k(\theta)\big\|
\]
are finite. Since $\nabla^2 g(\theta,Z)=\nabla^2 a(\theta)+\sum_l Z_l\,\nabla^2 b_l(\theta)$, we have
\[
\big\|\nabla^2 g(\theta,Z)\big\|
\le M_0+M_1\sum_l |Z_l|
\le M_0+\sqrt{d}\,M_1\,\|Z\|,
\]
where $d=\dim(Z)$. Setting $C_0=M_0$ and $C_1=\sqrt{d}\,M_1$ yields (iii).
\end{proof}

\begin{lemma}[VC-subgraph for the weighted quantile class]\label{lem:VCproduct}
Let $\mathcal{Q}=\{y\mapsto Q_y(u): u\in[a,b]\}$ on $\mathcal{Y}$ be VC-subgraph and let $s_x:\mathcal{Z}\to\mathbb{R}$ be a fixed measurable function $z\mapsto s(z,x)$. For a given $x$, define the class on $\mathcal{Y}\times\mathcal{Z}$,
\[
\mathcal{S}\;=\;\big\{(y,z)\mapsto s(z,x)\,Q_y(u):\ u\in[a,b]\big\}.
\]
Then $\mathcal{S}$ is VC-subgraph. As a result, $\mathcal{S}$ satisfies the uniform entropy bound (2.5.1) of \citet{vaart1996weak} with envelope
$H(y,z)=|s(z,x)|\,\sup_{u\in[a,b]}|Q_y(u)|$.
\end{lemma}

\begin{proof}
Define $\widetilde{\mathcal{Q}}=\{(y,z)\mapsto Q_y(u): u\in[a,b]\}$ on $\mathcal{Y}\times\mathcal{Z}$ and the fixed function $\tilde s(y,z)=s(z,x)$. As shown in \citet[Lemma A-3]{van2025regression}, $\mathcal{Q}$ is VC-subgraph on $\mathcal{Y}$. Hence, the lifted class $\widetilde{\mathcal{Q}}$ is VC-subgraph on $\mathcal{Y}\times\mathcal{Z}$. To see this, note that the composition of $\widetilde{\mathcal{Q}}$ and the fixed projection $(y,z)\mapsto y$ preserves the VC-subgraph property by \citet[Lemma~2.6.18(vii)]{vaart1996weak}. Further, by \citet[Lemma~2.6.18(vi)]{vaart1996weak}, the product of $\widetilde{\mathcal{Q}}$ and the fixed function $(z,y) \to \tilde{s}(y, z)$ is VC-subgraph. VC-subgraph classes satisfy the uniform entropy bound in \citet[Eq. 2.6.7]{vaart1996weak}  with respect to the $L_2(P)$-metric for every $P$. The stated envelope follows directly.
\end{proof}

\begin{proof}[Proof of Theorem~\ref{thm:convergence_unprojected}]
Let
\[
\hat\psi_x(u)=\frac{1}{n}\sum_{j=1}^n \hat s_j(x)\,Q_{Y_j}(u),
\qquad
\psi_x(u)=\E\!\big[s(Z,x)Q_Y(u)\big],
\]
where $\hat s_j(x)=g(\hat\theta,Z_j)$ and $s(Z,x)=g(\theta_0,Z)$. 

Write
\[
\sqrt{n}\big(\hat\psi_x(u)-\psi_x(u)\big)
=\underbrace{\frac{1}{\sqrt{n}}\sum_{j=1}^n \Big(s(Z_j,x)Q_{Y_j}(u)-\psi_x(u)\Big)}_{(A)}
+\underbrace{\frac{1}{\sqrt{n}}\sum_{j=1}^n \big(\hat s_j(x)-s(Z_j,x)\big)Q_{Y_j}(u)}_{(B)}.
\]

\textit{Main term  (A).}
Consider $\mathcal{F}=\{(Y,Z)\mapsto s(Z,x)Q_Y(u)-\psi_x(u):u\in[a,b]\}$. 
Recall from Lemma \ref{lem:VCproduct} that
\[
\mathcal{F}_Y
\;=\;
\big\{\,y\mapsto Q_y(u):\ u\in[a,b]\,\big\}
\]
is a VC–subgraph class on $\mathcal{Y}$ with envelope
$F_Y(y)=\sup_{u\in[a,b]}|Q_y(u)|$.  
For each fixed $x$, write $s_x(z)=s(z,x)$ and define the class on $\mathcal{Y}\times\mathcal{Z}$,
\[
\mathcal{S}
\;=\;
\big\{(y,z)\mapsto s_x(z)\,Q_y(u):\ u\in[a,b]\big\}.
\]
By Lemma~\ref{lem:VCproduct},
$\mathcal{S}$ is VC–subgraph with envelope
\[
H(y,z) = |s_x(z)|\,F_Y(y) \ \lesssim\ (1+\|z\|)\,\sup_{u\in[a,b]}|Q_y(u)|.
\]
By Assumption~\ref{asspt:finite_moments}, $\E[\|Z\|^{4}]<\infty$. Since $Q_Y(\cdot)$ is nondecreasing on $[a,b]$, $\sup_{u \in [a,b]} |Q_Y(u)| = \max\{|Q_Y(a)|, |Q_Y(b)|\}$, so $\E[\sup_{u \in [a,b]} |Q_Y(u)|^4] \leq \E[|Q_Y(a)|^4] + \E[|Q_Y(b)|^4] < \infty$ by Assumption~\ref{asspt:finite_moments}. Therefore
\[
P H^2
\ \le\ 
\big(\E[ (1+\|Z\|)^4 ]\big)^{1/2}
\big(\E[ (\sup_{u}|Q_Y(u)|)^4 ]\big)^{1/2}
<\infty.
\]
Therefore, by Theorem~2.5.2 of \citet{vaart1996weak}, $\mathcal{S}$ is $P$–Donsker and
\[
(A)\ \leadsto\ \mathbb{G}_1\quad\text{in }\ell^\infty([a,b]),
\]
a tight mean-zero Gaussian process with covariance 
$\Cov(s(Z,x)Q_Y(u),\,s(Z,x)Q_Y(u'))$.

\textit{Empirical weights term (B).}
By Lemma~\ref{lem:smooth_g_from_fullrank}, a second-order expansion gives
\[
\hat s_j(x)-s(Z_j,x)
=\nabla_\theta s(\theta_0,Z_j)^\top(\hat\theta-\theta_0)+R_j,
\quad
|R_j|\ \le\ C(1+\|Z_j\|)\,\|\hat\theta-\theta_0\|^2.
\]
Hence, uniformly in $u$,
\[
\Big|\frac{1}{\sqrt{n}}\sum_{j=1}^n R_j\,Q_{Y_j}(u)\Big|
\ \le\ \sqrt{n}\,\|\hat\theta-\theta_0\|^2\;\frac{1}{n}\sum_{j=1}^n (1+\|Z_j\|)\sup_{u\in[a,b]}|Q_{Y_j}(u)|
\ =\ o_p(1),
\]
by the LLN and $\sqrt{n}\|\hat\theta-\theta_0\|^2=o_p(1)$. Therefore
\[
(B)\ =\ \Big[\mathbb{P}_n f_u\Big]\ \sqrt{n}(\hat\theta-\theta_0)\ +\ o_p(1),
\qquad
f_u(Y,Z):=Q_Y(u)\,\nabla_\theta s(\theta_0,Z)^\top.
\]
By Lemma~\ref{lem:VCproduct}, the class $\{f_u:u\in[a,b]\}$ is VC-subgraph with envelope
$F(Y,Z)\lesssim (1+\|Z\|)\sup_{u\in[a,b]}|Q_Y(u)|\in L^1$, hence
$\sup_{u\in[a,b]}|\mathbb{P}_n f_u-\partial_\theta\psi_x(u)^\top|\to_p 0$.

\textit{Joint convergence and covariance.}
Let
\[
m(W)
=
\Big(
X-\mu_X,\ 
Z-\mu_Z,\ 
\operatorname{vec}\{(Z-\mu_Z)(X-\mu_X)^\top-\Sigma_{ZX}\},\ 
\operatorname{vec}\{(Z-\mu_Z)(Z-\mu_Z)^\top-\Sigma_{ZZ}\}
\Big)
\]
and $\mathcal{G}_\theta=\{m_k:k=1,\dots,d\}$ with $d=\dim\theta$.
Since $\mathcal{G}_\theta$ is finite and square-integrable, it is $P$--Donsker.
We have shown $\mathcal{S}$ is $P$--Donsker. Let $H$ and $G(W):=\max_k |m_k(W)|$ be envelopes for $\mathcal{S}$
and $\mathcal{G}_\theta$, respectively. Under Assumption~\ref{asspt:finite_moments}, $PH^2<\infty$ and $PG^2<\infty$,
so the union $\mathcal{J}=\mathcal{S}\cup\mathcal{G}_\theta$ has square-integrable envelope $J:=H+G$
and hence $\|P\|_{\mathcal{J}}<\infty$.
By \citet[Example~2.10.7]{vaart1996weak}, $\mathcal{J}$ is $P$--Donsker and therefore
\[
\Big(\ (A),\ \sqrt{n}(\hat\theta-\theta_0)\ \Big)\ \leadsto\ \big(\mathbb{G}_1,\ \mathbb{Z}_\theta\big)
\quad\text{in } \ell^\infty([a,b])\times\R^d,
\]
with $\mathbb{G}_1$ as above and $\mathbb{Z}_\theta\sim N(0,\Sigma)$.
The mapping $(h,z)\mapsto h(\cdot)+\partial_\theta\psi_x(\cdot)^\top z$
is continuous $\ell^\infty([a,b])\times\R^d\to\ell^\infty([a,b])$ and
$\sup_u|\mathbb{P}_n f_u-\partial_\theta\psi_x(u)^\top|\to_p0$,
so by the continuous mapping theorem,
\[
\sqrt{n}(\hat\psi_x-\psi_x)\ =\ (A)+(B)\ \leadsto\
\mathbb{G}_1+\partial_\theta\psi_x(\cdot)^\top\,\mathbb{Z}_\theta
\quad\text{in } \ell^\infty([a,b]).
\]
The covariance kernel follows from the joint Gaussian limit and this linear mapping.
\end{proof}

\begin{lemma}[Average Lipschitz increments imply $L^2$ continuity]\label{lem:QY_L2_cont}
Suppose Assumptions~\ref{asspt:finite_moments} and
\ref{asspt:average_lipschitz} hold. Then, for every $u\in[a,b]$,
\[
Q_Y(u_n)\to Q_Y(u)\quad\text{a.s.}
\qquad\text{whenever }u_n\to u.
\]
Moreover,
\[
E\!\left[|Q_Y(u_n)-Q_Y(u)|^4\right]\to0.
\]
Consequently, for any random vector $H$ with $E\|H\|^4<\infty$,
\[
E\!\left[\|H\|^2\,|Q_Y(u_n)-Q_Y(u)|^2\right]\to0.
\]
\end{lemma}

\begin{proof}
Fix $u\in[a,b)$. For $h>0$ such that $u+h\in[a,b]$, monotonicity gives
\[
0\le Q_Y(u+h)-Q_Y(u).
\]
By Assumption~\ref{asspt:average_lipschitz},
\[
0\le E[Q_Y(u+h)-Q_Y(u)]\le Kh\to0.
\]
Let
\[
\Delta_Y^+(u):=Q_Y(u+)-Q_Y(u)
=
\lim_{h\downarrow0}\{Q_Y(u+h)-Q_Y(u)\}.
\]
By Fatou's lemma,
\[
0\le E[\Delta_Y^+(u)]
\le
\liminf_{h\downarrow0}E[Q_Y(u+h)-Q_Y(u)]
=0.
\]
Hence $\Delta_Y^+(u)=0$ a.s. Since $Q_Y(\cdot)$ is left-continuous by
convention, $Q_Y(\cdot)$ is a.s. continuous at the fixed point $u$. At $u=b$,
left-continuity alone gives continuity from within $[a,b]$. Therefore
$Q_Y(u_n)\to Q_Y(u)$ a.s. for every deterministic sequence $u_n\to u$.

Furthermore,
\[
|Q_Y(u_n)-Q_Y(u)|^4
\le
16\sup_{t\in[a,b]}|Q_Y(t)|^4
=
16\max\{|Q_Y(a)|,|Q_Y(b)|\}^4,
\]
and the right-hand side is integrable by
Assumption~\ref{asspt:finite_moments}. Dominated convergence gives
$E[|Q_Y(u_n)-Q_Y(u)|^4]\to0$.

Finally, by Cauchy--Schwarz,
\[
E\!\left[\|H\|^2|Q_Y(u_n)-Q_Y(u)|^2\right]
\le
\{E\|H\|^4\}^{1/2}
\{E|Q_Y(u_n)-Q_Y(u)|^4\}^{1/2}
\to0.
\]
\end{proof}

 \begin{lemma}[Continuity of the covariance kernel]\label{lem:Gamma-cont-strong}
Under Assumptions~\ref{ass:full_rank}--\ref{asspt:finite_moments}, and
\ref{asspt:average_lipschitz}, the covariance kernel
\[
\Gamma_x(u,u')
:=
E\!\left[\phi_x(W;u)\phi_x(W;u')\right]
\]
of the limiting process in Theorem~\ref{thm:convergence_unprojected} is jointly
continuous on $[a,b]^2$. Consequently, its intrinsic variance semimetric
\[
\rho_x(u,u')^2
:=
E\!\left[(\mathbb G_x(u)-\mathbb G_x(u'))^2\right]
=
\Gamma_x(u,u)+\Gamma_x(u',u')-2\Gamma_x(u,u')
\]
satisfies $\rho_x(u_n,u)\to0$ whenever $u_n\to u$.
\end{lemma}

\begin{proof}
Recall from Appendix~\ref{sec:covariance_derivation} that
\[
\phi_x(W;u)
=
\bigl(s(Z,x)Q_Y(u)-\psi_x(u)\bigr)
+
\bigl(\partial_\theta\psi_x(u)\bigr)^{\T}m(W).
\]
By Lemma~\ref{lem:QY_L2_cont}, $s(Z,x)Q_Y(u_n)\to s(Z,x)Q_Y(u)$ in $L^2(P)$,
because for fixed $x$, $s(Z,x)$ is affine in $Z$ and hence has finite fourth
moment under Assumption~\ref{asspt:finite_moments}. Therefore
\[
\psi_x(u_n)=E[s(Z,x)Q_Y(u_n)]\to E[s(Z,x)Q_Y(u)]=\psi_x(u).
\]

Next, each component of $\partial_\theta\psi_x(u)$ is an expectation of the form
\[
E[Q_Y(u)h(W)]
\]
where $h(W)$ is either constant or affine in $Z$ for the gradient blocks derived
in Appendix~\ref{sec:covariance_derivation}. Hence
$\partial_\theta\psi_x(u_n)\to\partial_\theta\psi_x(u)$ by
Lemma~\ref{lem:QY_L2_cont} and Cauchy--Schwarz. Moreover, $E\|m(W)\|^2<\infty$
under Assumption~\ref{asspt:finite_moments}. It follows that
\[
\|\phi_x(\cdot;u_n)-\phi_x(\cdot;u)\|_{L^2(P)}\to0.
\]

Let $(u_n,u_n')\to(u,u')$. Then
\begin{align*}
|\Gamma_x(u_n,u_n')-\Gamma_x(u,u')|
&\le
E\!\left[|\phi_x(W;u_n)-\phi_x(W;u)|\,|\phi_x(W;u_n')|\right] \\
&\quad+
E\!\left[|\phi_x(W;u)|\,|\phi_x(W;u_n')-\phi_x(W;u')|\right].
\end{align*}
By Cauchy--Schwarz,
\begin{align*}
|\Gamma_x(u_n,u_n')-\Gamma_x(u,u')|
&\le
\|\phi_x(\cdot;u_n)-\phi_x(\cdot;u)\|_{L^2(P)}
\|\phi_x(\cdot;u_n')\|_{L^2(P)}\\
&\quad+
\|\phi_x(\cdot;u)\|_{L^2(P)}
\|\phi_x(\cdot;u_n')-\phi_x(\cdot;u')\|_{L^2(P)}.
\end{align*}
The difference terms converge to zero. By Lemma~\ref{lem:QY_L2_cont} and the preceding argument,
$u\mapsto \phi_x(\cdot;u)$ is continuous as a map from $[a,b]$ into $L^2(P)$.
Since $[a,b]$ is compact, it follows that $u\mapsto\|\phi_x(\cdot;u)\|_{L^2(P)}$ is bounded on $[a,b]$. Hence $\Gamma_x$ is jointly continuous. The statement for $\rho_x$ follows from the displayed identity and continuity of
$\Gamma_x$.
\end{proof}

\begin{corollary}\label{cor:cont-version}
Under the assumptions of Theorem~\ref{thm:convergence_unprojected} and
Assumption~\ref{asspt:average_lipschitz}, for each fixed $x\in\mathcal X$, the
centered tight Gaussian limit process $\mathbb G_x$ in
Theorem~\ref{thm:convergence_unprojected} admits a modification with a.s.\
continuous sample paths on $[a,b]$. In what follows we replace $\mathbb G_x$ by
this continuous modification, so that $\mathbb G_x\in C([a,b])$ a.s.
\end{corollary}

\begin{proof}
Let $\rho_x$ be the canonical semimetric of $\mathbb G_x$,
\[
\rho_x(u,u')^2
:=
E\!\left[(\mathbb G_x(u)-\mathbb G_x(u'))^2\right],
\qquad u,u'\in[a,b].
\]
By Lemma~\ref{lem:Gamma-cont-strong}, $\rho_x(u_n,u)\to0$ whenever
$u_n\to u$. Since $[a,b]$ is compact, this implies that $([a,b],\rho_x)$ is
totally bounded.

Moreover, $\mathbb G_x$ is a tight Borel measurable random element of
$\ell^\infty([a,b])$ and is stochastically continuous with respect to $\rho_x$.
By Addendum~1.5.8 of \citet{vaart1996weak}, $\mathbb G_x$ admits a modification
whose sample paths are a.s.\ uniformly $\rho_x$-continuous on $[a,b]$. For this
modification, if $u_n\to u$ in the usual metric, then $\rho_x(u_n,u)\to0$ and
hence $\mathbb G_x(u_n)\to\mathbb G_x(u)$. Therefore its sample paths are a.s.\
continuous on $[a,b]$.
\end{proof}

\begin{proof}[Proof of Theorem~\ref{thm:convergence_projected}]
The structure of this proof is analogous to the proof of Theorem 4 in \citet{van2025regression}. Recall from Theorem~\ref{thm:convergence_unprojected} that
\[
\sqrt{n}\,\big(\hat\psi_x-\psi_x\big)\ \leadsto\ \mathbb{G}_x
\qquad\text{in }\ell^\infty([a,b]).
\]
By Corollary~\ref{cor:cont-version}, we may (and do) replace $\mathbb G_x$ by
a modification such that $\mathbb G_x\in C([a,b])$ a.s.

All Fr\'echet objectives and $W_2$ distances are $L^2$-integrals of quantile functions. Hence changing
representatives on null sets (e.g.\ switching from left- to right-continuous versions) does not alter the
objective nor its argmin. We can therefore compute the $L^2([a,b])$ projection and then select a
right-continuous representative, to view the result as an element of $\ell^\infty([a,b])$.

As before, let
\[
\hat q(x,\cdot):=\Pi_{\mathcal Q}(\hat\psi_x),
\qquad
q(x,\cdot):=\Pi_{\mathcal Q}(\psi_x),
\]
so that $\hat q(x,\cdot)=\hat Q_{m_{\est}(x)}(\cdot)$ and $q(x,\cdot)$ is the population
target quantile function. Then, the result follows in three steps. 

\smallskip
\noindent\textit{Step 1 (Lipschitzness and uniform consistency).}
By Lemma A-6 (iii) in \citet{van2025regression}, $\Pi_{\mathcal Q}$ is $1$--Lipschitz in the uniform norm:
\[
\|\Pi_{\mathcal Q}(f)-\Pi_{\mathcal Q}(g)\|_\infty
\le \|f-g\|_\infty
\qquad\forall f,g\in\ell^\infty([a,b]).
\]
Therefore,
\[
\|\hat q(x,\cdot)-q(x,\cdot)\|_\infty
\le \|\hat\psi_x-\psi_x\|_\infty
=o_p(1).
\]

\smallskip
\noindent\textit{Step 2 (Local differentiability).}
By Assumptions~\ref{asspt:average_strictness} and
\ref{asspt:average_lipschitz}, for all $a\le c<d\le b$,
\[
\kappa(d-c)
\le
\psi_x(d)-\psi_x(c)
=
\E\!\left[s(Z,x)\{Q_Y(d)-Q_Y(c)\}\right]
\le
K(d-c).
\]
Hence $\psi_x$ is Lipschitz on $[a,b]$, therefore absolutely continuous, and the
lower secant bound implies $\psi_x'(u)\ge\kappa$ for a.e.\ $u\in[a,b]$.
Then $\psi_x\in\mathcal Q$, so $q(x,\cdot)=\Pi_{\mathcal Q}(\psi_x)=\psi_x$.
Moreover, the hypotheses of Lemma A-7 in \citet{van2025regression} are satisfied at $m = \psi_x$,
so $\Pi_{\mathcal Q}$ is Hadamard
directionally differentiable at $\psi_x$ as a mapping from $\ell^\infty([a,b])$ to $\ell^\infty([a,b])$ tangentially to $C([a,b])$, with derivative equal to the identity:
\[
D\Pi_{\mathcal Q}[\psi_x](h)=h
\qquad\forall\,h\in C([a,b]).
\]

\smallskip
\noindent\textit{Step 3 (Functional delta method).}
Since $\sqrt n(\hat\psi_x-\psi_x)\leadsto \mathbb G_x$ in $\ell^\infty([a,b])$ by Theorem \ref{thm:convergence_unprojected} and the limit satisfies
$\mathbb G_x\in C([a,b])$ a.s. by Corollary \ref{cor:cont-version}, the functional delta method for (fully) Hadamard differentiable maps
\citep[Thm.~20.8]{van2000asymptotic}, applied tangentially to $C([a,b])$, yields
\[
\sqrt{n}\,\big(\hat q(x,\cdot)-q(x,\cdot)\big)
=
\sqrt{n}\,\Big(\Pi_{\mathcal Q}(\hat\psi_x)-\Pi_{\mathcal Q}(\psi_x)\Big)
\ \leadsto\
D\Pi_{\mathcal Q}[\psi_x](\mathbb{G}_x)
=
\mathbb{G}_x
\qquad\text{in }\ell^\infty([a,b])
\]
which gives the result.
\end{proof}

\begin{proof}[Proof of Theorem~\ref{thm:convergence_beta}]

\medskip
\noindent
\textit{Step~1: Exact representation.}
The intercept estimator satisfies $\tilde\beta_0(u)=\frac{1}{n}\sum_{j=1}^n Q_{Y_j}(u)$, hence
\begin{equation}\label{eq:intercept_exact}
\tilde\beta_0(u)-\beta_0^{\mathrm{unc}}(u)
=
\frac{1}{n}\sum_{j=1}^n \bigl(Q_{Y_j}(u)-E[Q_Y(u)]\bigr).
\end{equation}
For the slope, since $\sum_j(Z_j-\hat\mu_Z)=0$ and
$(X_j-\mu_X)=(X_j-\hat\mu_X)+(\hat\mu_X-\mu_X)$,
the standard 2SLS normal equations give
\begin{equation}\label{eq:slope_exact}
\tilde\beta_1(u)-\beta_1^{\mathrm{unc}}(u)
=
\hat S_n\,\frac{1}{n}\sum_{j=1}^n (Z_j-\hat\mu_Z)\,\xi_j(u),
\end{equation}
where $\hat S_n:=(\hat\Sigma_{ZX}^{\!\T}\hat\Sigma_{ZZ}^{-1}\hat\Sigma_{ZX})^{-1}\hat\Sigma_{ZX}^{\!\T}\hat\Sigma_{ZZ}^{-1}$ and $\xi_j(u):=Q_{Y_j}(u)-\beta_0^{\mathrm{unc}}(u)-\beta_1^{\mathrm{unc}}(u)^{\!\T}(X_j-\mu_X)$.

 Writing $\tilde Z_j:=Z_j-\mu_Z$ and replacing the centering on the sample mean by centering on the population mean,
\[
\frac{1}{\sqrt{n}}\sum_{j=1}^n \bigl[(Z_j-\hat\mu_Z)-\tilde Z_j\bigr]\xi_j(u)
=
-\sqrt{n}(\hat\mu_Z-\mu_Z)\,\bar\xi_n(u),
\]
where $\bar\xi_n(u):=\frac{1}{n}\sum_j\xi_j(u)$. By the multivariate CLT (under Assumption \ref{asspt:finite_moments}), $\sqrt{n}(\hat\mu_Z-\mu_Z)=O_P(1)$. The Donsker property established in Step~2 below gives $\sup_{u\in[a,b]}|\bar\xi_n(u)|=O_P(n^{-1/2})$, so the product is $o_P(1)$ uniformly. Therefore
\begin{equation}\label{eq:slope_pop}
\sqrt{n}\bigl(\tilde\beta_1(u)-\beta_1^{\mathrm{unc}}(u)\bigr)
=
\hat S_n\,\frac{1}{\sqrt{n}}\sum_{j=1}^n \tilde Z_j\,\xi_j(u)+o_P(1)
\qquad\text{in }\ell^\infty([a,b])^{p}.
\end{equation}

\medskip
\noindent
\textit{Step~2: Functional CLT for the score.}
Define the joint score
\[
\Phi_j(u)
:=
\begin{pmatrix}
\zeta_j(u)\\[2pt]
\tilde Z_j\,\xi_j(u)
\end{pmatrix}
\in\R^{1+l},
\qquad
\zeta_j(u):=Q_{Y_j}(u)-E[Q_Y(u)].
\]
We verify that the class $\mathcal F:=\{(X,Y,Z)\mapsto \Phi(u):u\in[a,b]\}$ is $P$-Donsker componentwise. The first component class $\{\zeta(u):u\in[a,b]\}$ is $P$-Donsker by Lemma~\ref{lem:VCproduct} with square-integrable envelope $\sup_{u\in[a,b]}|Q_Y(u)|+|E[Q_Y(u)]|$ (Assumption~\ref{asspt:finite_moments}). For each coordinate $k\ge 1$ of the slope score, decompose $\tilde Z_k\,\xi(u)=\tilde Z_k\,Q_Y(u)-\tilde Z_k\,\mathbf X^{\!\T}\beta^{\mathrm{unc}}(u)$, where $\mathbf X:=(1,(X-\mu_X)^{\!\T})^{\!\T}$. The first term is $P$-Donsker by Lemma~\ref{lem:VCproduct} with square-integrable envelope;the second lies in a finite-dimensional linear span and is therefore $P$-Donsker (the envelope is square-integrable since $\sup_u\|\beta^{\text{unc}}(u)\|<\infty$: monotonicity of $Q_Y$ on $[a,b]$ gives $|Q_Y(u)|\le\max(|Q_Y(a)|,|Q_Y(b)|)$ for all $u\in[a,b]$, so $\|E[\mathbf{Z}\,Q_Y(u)]\|\le E[\|\mathbf{Z}\|\max(|Q_Y(a)|,|Q_Y(b)|)]<\infty$ uniformly in $u$ by Cauchy--Schwarz and Assumption~\ref{asspt:finite_moments}, and $\beta^{\text{unc}}(u)=\bar{S}_{\mathrm{2SLS}}\,E[\mathbf{Z}\,Q_Y(u)]$ inherits this bound). Since there are finitely many components, $\mathcal{F}$ is $P$-Donsker \citep[Example 2.10.7]{vaart1996weak}.  

To obtain the Gaussian limit for the coefficient process, it is enough to show that the transformed score $T\Phi(u)$ is mean zero, since $\sqrt n(\tilde\beta(\cdot)-\beta^{\mathrm{unc}}(\cdot))$ is asymptotically equivalent to $T n^{-1/2}\sum_{j=1}^n \Phi_j(\cdot)$. The intercept component satisfies $E[\zeta(u)]=0$ by construction. For the
slope component, we do not require $E[\tilde Z\,\xi(u)]=0$ componentwise.
Instead, it is enough that the transformed score entering the coefficient
process has mean zero. Recall that
\[
T=\mathrm{diag}(1,S_0),
\qquad
S_0=(\Sigma_{ZX}^{\!\T}\Sigma_{ZZ}^{-1}\Sigma_{ZX})^{-1}
\Sigma_{ZX}^{\!\T}\Sigma_{ZZ}^{-1}.
\]
Then
\[
T\,E[\Phi(u)]
=
\begin{pmatrix}
0\\
S_0\,E[\tilde Z\,\xi(u)]
\end{pmatrix}.
\]
By the definition of $\beta^{\mathrm{unc}}(u)$ as the population 2SLS
coefficient function,
\[
\Sigma_{ZX}^{\!\T}\Sigma_{ZZ}^{-1}E[\tilde Z\,\xi(u)]=0,
\]
and hence $S_0\,E[\tilde Z\,\xi(u)]=0$. Therefore
\[
T\,E[\Phi(u)]=0.
\]

Since $\mathcal F$ is $P$-Donsker, it follows that
\[
\frac{1}{\sqrt n}\sum_{j=1}^n\bigl(\Phi_j(\cdot)-E[\Phi(\cdot)]\bigr)
\leadsto
\mathbb G_\Phi(\cdot)
\qquad\text{in }\ell^\infty([a,b])^{1+l},
\]
where $\mathbb G_\Phi$ is a tight mean-zero Gaussian process with covariance
kernel
\[
E\!\left[\bigl(\Phi(u)-E[\Phi(u)]\bigr)\bigl(\Phi(u')-E[\Phi(u')]\bigr)^{\!\T}\right].
\]
Multiplying by the constant matrix $T$ and using $T\,E[\Phi(u)]=0$, we obtain
\[
T\,\frac{1}{\sqrt n}\sum_{j=1}^n \Phi_j(\cdot)
\leadsto
T\,\mathbb G_\Phi(\cdot)
\qquad\text{in }\ell^\infty([a,b])^{p+1}.
\]

\medskip
\noindent
\textit{Step~3: Conclusion.} We now combine the asymptotic linear representation with the centered Gaussian limit for the transformed score process.
By the law of large numbers and Assumption~\ref{ass:full_rank},
\[
\hat S_n\toP S_0:=(\Sigma_{ZX}^{\!\T}\Sigma_{ZZ}^{-1}\Sigma_{ZX})^{-1}
\Sigma_{ZX}^{\!\T}\Sigma_{ZZ}^{-1}.
\]
Define $T:=\mathrm{diag}(1,S_0)\in\R^{(p+1)\times(1+l)}$. Combining
\eqref{eq:intercept_exact}--\eqref{eq:slope_pop} with Slutsky's theorem and the
result of Step~2,
\[
\sqrt{n}\bigl(\tilde\beta(\cdot)-\beta^{\mathrm{unc}}(\cdot)\bigr)
=
\mathrm{diag}(1,\hat S_n)\,
\frac{1}{\sqrt{n}}\sum_{j=1}^n \Phi_j(\cdot)
+o_P(1)
\;\leadsto\;
T\,\mathbb G_\Phi(\cdot)
=:\mathbb G_\beta(\cdot)
\quad\text{in }\ell^\infty([a,b])^{p+1}.
\]
Since $T\,E[\Phi(u)]=0$, the covariance kernel of $\mathbb G_\beta$ is
\[
\Omega(u,u')
=
T\,E\!\left[\Phi(u)\Phi(u')^{\!\T}\right]T^{\!\T}.
\]
In particular, the slope--slope block is
\[
S_0\,E\!\left[\tilde Z\tilde Z^{\!\T}\xi(u)\xi(u')\right]S_0^{\!\T},
\]
the standard heteroskedasticity-robust 2SLS sandwich form.

\end{proof}

\begin{lemma}[Continuity of the coefficient covariance kernel]\label{lem:Omega-cont}
Under Assumption~\ref{ass:full_rank}--\ref{asspt:finite_moments}, and
Assumption~\ref{asspt:average_lipschitz}, the covariance kernel
\[
\Omega(u,u')
=
T\,E\!\left[\Phi(u)\Phi(u')^{\T}\right]T^{\T}
\]
of the limiting process in Theorem~\ref{thm:convergence_beta} is jointly
continuous on $[a,b]^2$.
\end{lemma}

\begin{proof}
We first show that $u\mapsto\Phi(u)$ is continuous from $[a,b]$ into $L^2(P)$.
The first component follows from Lemma~\ref{lem:QY_L2_cont}:
\[
Q_Y(u_n)-E[Q_Y(u_n)]
\to
Q_Y(u)-E[Q_Y(u)]
\quad\text{in }L^2(P).
\]
For the slope component, write
\[
\xi(u_n)-\xi(u)
=
\{Q_Y(u_n)-Q_Y(u)\}
-
\mathbf X^{\T}\{\beta^{\mathrm{unc}}(u_n)-\beta^{\mathrm{unc}}(u)\}.
\]
Since
\[
\beta^{\mathrm{unc}}(u)
=
\begin{pmatrix}
E[Q_Y(u)]\\
S_0E[\tilde ZQ_Y(u)]
\end{pmatrix},
\]
Lemma~\ref{lem:QY_L2_cont} implies
$\beta^{\mathrm{unc}}(u_n)\to\beta^{\mathrm{unc}}(u)$. Hence, using
Assumption~\ref{asspt:finite_moments},
\[
E\!\left[\|\tilde Z\|^2|\xi(u_n)-\xi(u)|^2\right]\to0.
\]
Therefore
\[
\|\Phi(u_n)-\Phi(u)\|_{L^2(P)}\to0.
\]

Now let $(u_n,u_n')\to(u,u')$. Since
\[
\Phi(u_n)\Phi(u_n')^{\T}-\Phi(u)\Phi(u')^{\T}
=
\bigl(\Phi(u_n)-\Phi(u)\bigr)\Phi(u_n')^{\T}
+
\Phi(u)\bigl(\Phi(u_n')-\Phi(u')\bigr)^{\T},
\]
Cauchy--Schwarz gives
\begin{align*}
&\left\|E\!\left[
\Phi(u_n)\Phi(u_n')^{\T}-\Phi(u)\Phi(u')^{\T}
\right]\right\| \\
&\qquad\le
\|\Phi(u_n)-\Phi(u)\|_{L^2(P)}\,
\|\Phi(u_n')\|_{L^2(P)}
+
\|\Phi(u)\|_{L^2(P)}\,
\|\Phi(u_n')-\Phi(u')\|_{L^2(P)}.
\end{align*}
The difference terms converge to zero, and
$u\mapsto\|\Phi(u)\|_{L^2(P)}$ is bounded on compact $[a,b]$. Therefore
\[
E[\Phi(u_n)\Phi(u_n')^{\T}]
\to
E[\Phi(u)\Phi(u')^{\T}].
\]
Since $T$ is constant, $\Omega(u,u')$ is jointly continuous.
\end{proof}

\begin{corollary}\label{cor:cont-version-beta}
Under the assumptions of Theorem~\ref{thm:convergence_beta} and
Assumption~\ref{asspt:average_lipschitz}, the Gaussian limit
$\mathbb{G}_\beta$ admits a modification with a.s.\ continuous sample paths on
$[a,b]$.
\end{corollary}

\begin{proof}
By Lemma~\ref{lem:Omega-cont}, the covariance kernel $\Omega(u,u')$ is jointly
continuous on $[a,b]^2$. Since $[a,b]$ is compact and $\mathbb G_\beta$ is a
mean-zero Gaussian process, Addendum~1.5.8 of \citet{vaart1996weak} implies
that $\mathbb G_\beta$ admits a version with a.s.\ continuous sample paths on
$[a,b]$.
\end{proof}

\begin{proof}[Proof of Theorem~\ref{thm:pw_ols_clt}]

By Theorem~\ref{thm:convergence_beta} and
Corollary~\ref{cor:cont-version-beta},
\[
\sqrt{n}(\tilde\beta(\cdot)-\beta^{\mathrm{unc}}(\cdot))
\leadsto
\mathbb G_\beta(\cdot)
\]
in $\ell^\infty([a,b])^{p+1}$, and the limit process admits a version with
a.s.\ continuous sample paths on $[a,b]$. Under Assumption~\ref{asspt:average_strictness}, $\psi_x\in\mathcal Q$ for every $x\in\mathcal X$, so the population projection is inactive and $\beta^{\est}=\beta^{\mathrm{unc}}$. It therefore suffices to show that
\begin{equation}\label{eq:Delta_goal}
\sqrt{n}\,\|\Delta_n\|_\infty=o_P(1),
\qquad
\Delta_n:=\hat\beta^{\est}-\tilde\beta,
\end{equation}
since the result then follows from Slutsky's theorem.

The projection correction $D_x(u):=\Pi_{\mathcal Q}(\hat\psi_x)(u)-\hat\psi_x(u)$ enters through the OLS decomposition,
\[
\Delta_{0,n}(u)=\frac{1}{n}\sum_{j=1}^n D_{X_j}(u),
\qquad
\Delta_{1,n}(u)=\hat\Sigma_{XX}^{-1}\frac{1}{n}\sum_{j=1}^n(X_j-\hat\mu_X)\,D_{X_j}(u).
\]
Under the bounded support assumption \ref{asspt:bounded_support}, $\sup_{x\in\mathcal X}\|x-\mu_X\|\le B$ and thus $\|X_j-\hat\mu_X\|\le 2B$ for all $j$, so both components can be controlled once we establish
\begin{equation}\label{eq:unif_Dx}
\sup_{x\in\mathcal X}\sqrt{n}\,\|D_x\|_\infty\;\toP\;0.
\end{equation}
Indeed, \eqref{eq:unif_Dx} immediately gives $\sqrt{n}\|\Delta_{0,n}\|_\infty\to_P 0$, and for the slope \[\sqrt{n}\|\Delta_{1,n}\|_\infty\le 2B\,\|\hat\Sigma_{XX}^{-1}\|_{\mathrm{op}}\sup_x\sqrt{n}\|D_x\|_\infty\to_P 0\] since $\hat\Sigma_{XX}^{-1}=O_P(1)$.

We now prove~\eqref{eq:unif_Dx}. The estimation error $e_x(u):=\hat\psi_x(u)-\psi_x(u)$ is affine in $x$. Writing $d_n := \tilde\beta_1-\beta_1^{\mathrm{unc}} \in \R^p$ with components $d_{n,1},\ldots,d_{n,p}$, and $a_n := \tilde\beta_0-\beta_0^{\mathrm{unc}} + \tilde\beta_1^{\!\top}(\mu_X-\hat\mu_X)$, we have
\[
e_x = a_n + \sum_{r=1}^p d_{n,r}\,(x_r-\mu_{X,r}).
\]
Set $H_n^x:=\sqrt{n}\,e_x$ and denote its $p+1$ coefficient functions by $H_{0,n}:=\sqrt{n}\,a_n$ and $H_{r,n}:=\sqrt{n}\,d_{n,r}$ for $r=1,\ldots,p$, so that $H_n^x=H_{0,n}+\sum_{r=1}^p H_{r,n}\,(x_r-\mu_{X,r})$. We claim that $(H_{0,n},H_{1,n},\ldots,H_{p,n})$ is jointly asymptotically tight in $C([a,b])^{p+1}$. For each $r=1,\ldots,p$, $\sqrt{n}\,d_{n,r}$ is a coordinate of $\sqrt{n}(\tilde\beta_1-\beta_1^{\mathrm{unc}})$, which converges weakly in $\ell^\infty([a,b])^p$ to a limit with a.s.\ continuous paths (Corollary~\ref{cor:cont-version-beta}), and is therefore asymptotically tight in $C([a,b])^p$. For the intercept coefficient, split $\tilde\beta_1 = \beta_1^{\mathrm{unc}} + d_n$ to obtain
\[
\sqrt{n}\,a_n=\sqrt{n}(\tilde\beta_0-\beta_0^{\mathrm{unc}})-\beta_1^{\mathrm{unc}}{}^{\!\top}\sqrt{n}(\hat\mu_X-\mu_X)+d_n^{\!\top}\sqrt{n}(\mu_X-\hat\mu_X);
\]
the first term is asymptotically tight in $C([a,b])$ by Theorem~\ref{thm:convergence_beta} and Corollary~\ref{cor:cont-version-beta} and the second is $O_P(1)$ times the continuous function $\beta_1^{\mathrm{unc}}(\cdot)$. For the third term, $\sup_{u\in[a,b]}\|d_n(u)\|=O_P(n^{-1/2})$ by Theorem~\ref{thm:convergence_beta} and $\sqrt{n}(\hat\mu_X-\mu_X)=O_P(1)$ by the multivariate CLT under Assumption~\ref{asspt:finite_moments}, so $\sup_{u\in[a,b]}|d_n(u)^{\!\T}\sqrt{n}(\mu_X-\hat\mu_X)|\le \sup_u\|d_n(u)\|\cdot\|\sqrt{n}(\hat\mu_X-\mu_X)\|=o_P(1)$.

Fix $\varepsilon>0$. By asymptotic tightness in $C([a,b])^{p+1}$, there exist a compact $K_\varepsilon\subset C([a,b])^{p+1}$ and an integer $N_\varepsilon$ such that $P((\sqrt{n}\,a_n,\sqrt{n}\,d_{n,1},\ldots,\sqrt{n}\,d_{n,p})\in K_\varepsilon)\ge 1-\varepsilon$ for all $n\ge N_\varepsilon$. Since $K_\varepsilon$ is compact in $C([a,b])^{p+1}$, it is totally bounded and can be covered by finitely many sup-norm $\varepsilon$-balls. Approximating each ball center by a nearby $C^1$ function (using density of $C^1$ in $C([a,b])$) and setting $\Lambda_\varepsilon$ to be the maximum Lipschitz constant over these finitely many approximants, we obtain random functions $\ell_{0,n},\dots,\ell_{p,n}\in C^1([a,b])$ with $\mathrm{Lip}(\ell_{r,n})\le\Lambda_\varepsilon$ (namely, the nearest $C^1$ center to the realization) such that the event
\[
\mathcal{E}_n := \Big\{\max_r\|H_{r,n}-\ell_{r,n}\|_\infty\le\varepsilon\Big\}
\]
satisfies $P(\mathcal{E}_n)\ge 1-\varepsilon$ for all $n\ge N_\varepsilon$. Since $H_n^x=H_{0,n}+\sum_r H_{r,n}\,(x_r-\mu_{X,r})$, define the corresponding smooth proxy $L_n^x:=\ell_{0,n}+\sum_r\ell_{r,n}(x_r-\mu_{X,r})$, so that $H_n^x-L_n^x=(H_{0,n}-\ell_{0,n})+\sum_r(H_{r,n}-\ell_{r,n})(x_r-\mu_{X,r})$. On $\mathcal{E}_n$, for every $x\in\mathcal X$, the triangle inequality gives
\[
\|H_n^x-L_n^x\|_\infty
\le\|H_{0,n}-\ell_{0,n}\|_\infty+\sum_r\|H_{r,n}-\ell_{r,n}\|_\infty\,|x_r-\mu_{X,r}|
\le\varepsilon(1+\|x-\mu_X\|_1)
\le C_B\varepsilon,
\]
where $C_B:=1+\sqrt{p}\,B$ by Cauchy--Schwarz and Assumption~\ref{asspt:bounded_support}. Similarly, $\mathrm{Lip}(L_n^x)\le\Lambda_\varepsilon(1+\|x-\mu_X\|_1)\le\Lambda_\varepsilon C_B$.

For $n$ large enough that $n^{-1/2}\Lambda_\varepsilon C_B\le\kappa/2$, the function $\psi_x+n^{-1/2}L_n^x$ has secant slopes at least $\kappa-n^{-1/2}\Lambda_\varepsilon C_B\ge\kappa/2>0$ on $[a,b]$ (using Assumption~\ref{asspt:average_strictness}), hence lies in $\mathcal Q$. By the $\|\cdot\|_\infty$-contraction of $\Pi_{\mathcal Q}$ \citep[Lemma A-6 in][]{van2025regression}:
\begin{align*}
\sqrt{n}\,\|D_x\|_\infty
&=\|\Pi_{\mathcal Q}(\psi_x+n^{-1/2}H_n^x)-(\psi_x+n^{-1/2}H_n^x)\|_\infty\cdot\sqrt{n}\\
&\le\|\Pi_{\mathcal Q}(\psi_x+n^{-1/2}H_n^x)-\Pi_{\mathcal Q}(\psi_x+n^{-1/2}L_n^x)\|_\infty\cdot\sqrt{n}+\|H_n^x-L_n^x\|_\infty\\
&\le 2\|H_n^x-L_n^x\|_\infty
\le 2C_B\varepsilon.
\end{align*}
Since the bound holds simultaneously for all $x\in\mathcal X$ on $\mathcal{E}_n$, we have $P(\sup_{x\in\mathcal X}\sqrt{n}\,\|D_x\|_\infty\le 2C_B\varepsilon)\ge 1-\varepsilon$ for all $n$ large enough. Moreover, since $\varepsilon>0$ is arbitrary, \eqref{eq:unif_Dx} follows.

\end{proof}

\begin{proof}[Proof of Theorem~\ref{thm:bootstrap}]
The proof of Theorem~\ref{thm:convergence_beta} establishes that the score class $\mathcal{F}:=\{\Phi(\cdot,u):u\in[a,b]\}$ is $P$-Donsker with square-integrable envelope. By the multiplier CLT for Donsker classes \citep[Theorem~2.9.6]{vaart1996weak}, the centered process $n^{-1/2}\sum_j\omega_j(\Phi_j(\cdot)-\bar\Phi_n(\cdot))\rightsquigarrow_{\Prob}\mathbb{G}_\Phi(\cdot)$, where $\bar\Phi_n:=n^{-1}\sum_j\Phi_j$. The difference between centered and uncentered processes is $\bar\omega_n\cdot\sqrt{n}\,\bar\Phi_n(\cdot)$, where $\bar\omega_n:=n^{-1}\sum_j\omega_j=O_{P^X}(n^{-1/2})$ and $\sup_u\|\bar\Phi_n(u)\|=O_P(1)$ by the Donsker property, so the difference is $o_{P^X}(1)$ uniformly. Therefore the uncentered infeasible process satisfies $T\,n^{-1/2}\sum_j\omega_j\Phi_j(\cdot)\rightsquigarrow_{\Prob}\mathbb{G}_\beta(\cdot)$.

It remains to replace $(\Phi_j, T)$ by $(\hat\Phi_j, \hat T)$. Since $\|\hat T-T\|=o_p(1)$, it suffices to show $R_n(\cdot):=n^{-1/2}\sum_j\omega_j\Delta_j(\cdot)=o_{P^x}(1)$ in $\ell^\infty([a,b])^{1+l}$, where $\Delta_j:=\hat\Phi_j-\Phi_j$. We decompose $\Delta_j$ and bound each term.

For the intercept component, $\Delta_{j,0}(u)=E[Q_Y(u)]-\bar Q_n(u)$ does not depend on $j$, so $n^{-1/2}\sum_j\omega_j\Delta_{j,0}(u)=\sqrt{n}\,(E[Q_Y(u)]-\bar Q_n(u))\,\bar\omega_n$. Since $\sup_u|E[Q_Y(u)]-\bar Q_n(u)|=O_p(n^{-1/2})$ by the Donsker property and $\bar\omega_n=O_{P^x}(n^{-1/2})$, this is $o_{P^x}(1)$ uniformly.

For the slope components ($k=1,\ldots,l$), remember the definition $\tilde{Z}_{jk} = Z_{jk} - \mu_{Z,k}$. Then write,
\begin{align*}
\Delta_{j,k}(u)
& =(Z_{jk}-\bar Z_{n,k})\hat\xi_j(u)-\tilde Z_{jk}\,\xi_j(u) \\
&=\underbrace{(\mu_{Z,k}-\bar Z_{n,k})\,\xi_j(u)}_{(a)}
+\underbrace{\tilde Z_{jk}\big(\hat\xi_j(u)-\xi_j(u)\big)}_{(b)}
+\underbrace{(\mu_{Z,k}-\bar Z_{n,k})\big(\hat\xi_j(u)-\xi_j(u)\big)}_{(c)}.
\end{align*}

For~$(a)$: $n^{-1/2}\sum_j\omega_j(\mu_{Z,k}-\bar Z_{n,k})\xi_j(u)=(\mu_{Z,k}-\bar Z_{n,k})\cdot n^{-1/2}\sum_j\omega_j\xi_j(u)$. The first factor is $o_p(1)$; the second is $O_{P^X}(1)$ uniformly in $u$ by the multiplier CLT applied to the class $\{\xi(\cdot,u):u\in[a,b]\}$, which is $P$-Donsker as shown in the proof of Theorem~\ref{thm:convergence_beta} (Step~2). The product is $o_{P^x}(1)$ uniformly.

For~$(b)$: expand $\hat\xi_j(u)-\xi_j(u)=-\mathbf{X}_j^{\!\T}(\tilde\beta(u)-\beta^{\mathrm{unc}}(u))+(\hat\mu_X-\mu_X)^{\!\T}\tilde\beta_1(u)$. Then
\begin{align*}
& \frac{1}{\sqrt n}\sum_j\omega_j\tilde Z_{jk}\big(\hat\xi_j(u)-\xi_j(u)\big)\\
& =-\big(\tilde\beta(u)-\beta^{\mathrm{unc}}(u)\big)^{\!\T}\frac{1}{\sqrt n}\sum_j\omega_j\tilde Z_{jk}\mathbf{X}_j
+(\hat\mu_X-\mu_X)^{\!\T}\tilde\beta_1(u)\frac{1}{\sqrt n}\sum_j\omega_j\tilde Z_{jk}.
\end{align*}
In the first term, $\sup_u\|\tilde\beta(u)-\beta^{\mathrm{unc}}(u)\|=O_p(n^{-1/2})$ by Theorem~\ref{thm:convergence_beta}, while $n^{-1/2}\sum_j\omega_j\tilde Z_{jk}\mathbf{X}_j=O_{P^x}(1)$ (a finite-dimensional multiplier sum with finite second moments by Assumption~\ref{asspt:finite_moments}). The product is $o_{P^x}(1)$ uniformly. The second term is $o_p(1)\cdot O_p(1)\cdot O_{P^x}(1)=o_{P^x}(1)$ by the same reasoning.

For~$(c)$: the product of the $o_p(1)$ and $O_p(n^{-1/2})$ factors from~$(a)$ and~$(b)$ gives a contribution of smaller order.

Combining all terms yields $\sup_u\|R_n(u)\|=o_{P^x}(1)$ in probability, completing the proof.
\end{proof}

\begin{proof}[Proof of Corollary~\ref{cor:bootstrap_projected}]
Conditionally on the data, the projected bootstrap coefficients are
$\hat\beta^{\est,*}(u)=(\hat{\mathbf{X}}^{\!\T}\hat{\mathbf{X}})^{-1}\hat{\mathbf{X}}^{\!\T}\Pi_{\mathcal Q}(\hat\psi^*_{X_j})(u)$,
where $\hat\psi^*_{X_j}(u)=\tilde\beta^*_0(u)+\tilde\beta^{*\!\T}_1(u)(X_j-\hat\mu_X)$ and $\tilde\beta^*$ are the bootstrap unconstrained coefficients. The bootstrap process is $\hat{\mathbb{G}}_\beta^{\est,*}(u)=\sqrt{n}(\hat\beta^{\est,*}(u)-\hat\beta^{\est}(u))$.

The argument proceeds exactly as in the proof of Theorem~\ref{thm:pw_ols_clt}, with the bootstrap perturbation playing the role of the estimation error. Specifically, define the bootstrap analogue of the estimation error $e_x^*(u):=\hat\psi^*_x(u)-\hat\psi_x(u)$, which is affine in $x$ with coefficient functions $\tilde\beta^*(u)-\tilde\beta(u)$. The bootstrap projection correction is $D_x^*(u):=\Pi_{\mathcal Q}(\hat\psi^*_x)(u)-\hat\psi^*_x(u)$, and the difference between projected and unprojected bootstrap coefficients is
\[
\hat\beta^{\est,*}(u)-\tilde\beta^*(u)=(\hat{\mathbf{X}}^{\!\T}\hat{\mathbf{X}})^{-1}\hat{\mathbf{X}}^{\!\T} D^*_{X_j}(u).
\]
It suffices to show $\sup_{x\in\mathcal X}\sqrt{n}\,\|D^*_x\|_\infty=o_{P^x}(1)$ in probability. By Theorem~\ref{thm:bootstrap}, $\sqrt{n}(\tilde\beta^*(\cdot)-\tilde\beta(\cdot))\rightsquigarrow_{\Prob}\mathbb{G}_\beta(\cdot)$ in $\ell^\infty([a,b])^{p+1}$, which has a.s.\ continuous paths by Corollary~\ref{cor:cont-version-beta}. The bootstrap coefficient functions are therefore conditionally asymptotically tight in $C([a,b])^{p+1}$. Write $\hat\psi_x^*=\psi_x+n^{-1/2}(H_n^x+H_n^{*,x})$, where $H_n^x:=\sqrt{n}(\hat\psi_x-\psi_x)$ is the scaled estimation error and $H_n^{*,x}:=\sqrt{n}(\hat\psi_x^*-\hat\psi_x)$ is the scaled bootstrap perturbation. Denote the coefficient functions of $H_n^x$ by $(H_{0,n},\ldots,H_{p,n})$ as in the proof of Theorem~\ref{thm:pw_ols_clt}, and those of $H_n^{*,x}$ by $(H_{0,n}^*,\ldots,H_{p,n}^*)$. By Theorem~\ref{thm:convergence_beta} and Corollary~\ref{cor:cont-version-beta}, $(H_{0,n},\ldots,H_{p,n})$ is asymptotically tight in $C([a,b])^{p+1}$. By Theorem~\ref{thm:bootstrap}, $(H_{0,n}^*,\ldots,H_{p,n}^*)$ is conditionally asymptotically tight in $C([a,b])^{p+1}$ on a data event $\mathcal{A}_n$ with $P(\mathcal{A}_n)\to 1$.

Fix $\varepsilon>0$. By the same covering argument as in the proof of Theorem~\ref{thm:pw_ols_clt}, on a data event of probability $\geq 1-\varepsilon$ for $n$ large, there exist $C^1$ proxies $\ell_{r,n}$ for $H_{r,n}$ with $\max_r\|H_{r,n}-\ell_{r,n}\|_\infty\le\varepsilon$ and $\mathrm{Lip}(\ell_{r,n})\le\Lambda_\varepsilon$. Conditionally on such data, there likewise exist $C^1$ proxies $\ell_{r,n}^*$ for $H_{r,n}^*$ with $\max_r\|H_{r,n}^*-\ell_{r,n}^*\|_\infty\le\varepsilon$ and $\mathrm{Lip}(\ell_{r,n}^*)\le\Lambda_\varepsilon^*$ with conditional probability $\geq 1-\varepsilon$. Define $L_n^x:=\ell_{0,n}+\sum_r\ell_{r,n}(x_r-\mu_{X,r})$ and $L_n^{*,x}:=\ell_{0,n}^*+\sum_r\ell_{r,n}^*(x_r-\mu_{X,r})$, so that $\mathrm{Lip}(L_n^x+L_n^{*,x})\le(\Lambda_\varepsilon+\Lambda_\varepsilon^*)C_B$. For $n$ large enough that $n^{-1/2}(\Lambda_\varepsilon+\Lambda_\varepsilon^*)C_B\le\kappa/2$, Assumption~\ref{asspt:average_strictness} ensures that $\psi_x+n^{-1/2}(L_n^x+L_n^{*,x})$ has secant slopes $\geq\kappa/2>0$ on $[a,b]$ for all $x\in\mathcal X$, hence lies in $\mathcal Q$. By the $\|\cdot\|_\infty$-contraction of $\Pi_{\mathcal Q}$ \citep[Lemma A-6 (iii) in][]{van2025regression}:
\begin{align*}
\sqrt{n}\,\|D_x^*\|_\infty
&=\sqrt{n}\,\|\Pi_{\mathcal Q}(\hat\psi_x^*)-\hat\psi_x^*\|_\infty\\
&\le 2\|(H_n^x+H_n^{*,x})-(L_n^x+L_n^{*,x})\|_\infty\\
&\le 2(\|H_n^x-L_n^x\|_\infty+\|H_n^{*,x}-L_n^{*,x}\|_\infty)
\le 4C_B\varepsilon,
\end{align*}
uniformly over $x\in\mathcal X$. Since $\varepsilon$ is arbitrary, $\sup_x\sqrt{n}\|D_x^*\|_\infty=o_{P^X}(1)$ in probability.

It follows that $\hat{\mathbb{G}}_\beta^{\est,*}$ and the unprojected bootstrap process $\hat{\mathbb{G}}_\beta^*$ have the same conditional weak limit $\mathbb{G}_\beta$. Combined with Theorem~\ref{thm:bootstrap}, this gives $\hat{\mathbb{G}}_\beta^{\est,*}(\cdot)\rightsquigarrow_{\Prob}\mathbb{G}_\beta(\cdot)$ in $\ell^\infty([a,b])^{p+1}$.
\end{proof}

\begin{proof}[Proof of Proposition~\ref{prop:empirical_quantiles}]
For each $j\le n$, define the first-stage quantile estimation error
\[
\hat\Delta_j(u):=\widehat Q_{Y_j}(u)-Q_{Y_j}(u),
\qquad
u\in[a,b],
\]
and recall
\[
R_n=\max_{1\le j\le n}\|\hat\Delta_j\|_\infty.
\]
By Assumption~\ref{asspt:empirical_quantiles}, $R_n=o_p(n^{-1/2})$.

\medskip
\noindent\emph{Step 1: unconstrained coefficients.}
Because the feasible unprojected estimator is just the sample 2SLS coefficient vector with $\widehat Q_{Y_j}(u)$ in place of $Q_{Y_j}(u)$ at each $u$, we have
\[
\bar{\tilde\beta}_0(u)-\tilde\beta_0(u)
=
\frac{1}{n}\sum_{j=1}^n \hat\Delta_j(u),
\]
and
\[
\bar{\tilde\beta}_1(u)-\tilde\beta_1(u)
=
\hat S_{\mathrm{2SLS}}
\left(
\frac{1}{n}\sum_{j=1}^n (Z_j-\hat\mu_Z)\,\hat\Delta_j(u)
\right),
\]
where
\[
\hat S_{\mathrm{2SLS}}
:=
(\hat\Sigma_{ZX}^{\!\T}\hat\Sigma_{ZZ}^{-1}\hat\Sigma_{ZX})^{-1}
\hat\Sigma_{ZX}^{\!\T}\hat\Sigma_{ZZ}^{-1}.
\]
Therefore,
\[
\sup_{u\in[a,b]}\bigl|\bar{\tilde\beta}_0(u)-\tilde\beta_0(u)\bigr|
\le
\frac{1}{n}\sum_{j=1}^n \|\hat\Delta_j\|_\infty
\le R_n,
\]
and
\begin{align}
\sup_{u\in[a,b]}
\bigl\|
\bar{\tilde\beta}_1(u)-\tilde\beta_1(u)
\bigr\|
&\le
\|\hat S_{\mathrm{2SLS}}\|
\cdot
\frac{1}{n}\sum_{j=1}^n \|Z_j-\hat\mu_Z\|\,\|\hat\Delta_j\|_\infty \nonumber\\
&\le
\|\hat S_{\mathrm{2SLS}}\|
\left(
\frac{1}{n}\sum_{j=1}^n \|Z_j-\hat\mu_Z\|
\right)R_n.
\label{eq:empq_slope_bound}
\end{align}
Under Assumption~\ref{ass:full_rank}, the population matrix
\[
S_{\mathrm{2SLS}}
:=
(\Sigma_{ZX}^{\!\T}\Sigma_{ZZ}^{-1}\Sigma_{ZX})^{-1}
\Sigma_{ZX}^{\!\T}\Sigma_{ZZ}^{-1}
\]
is well-defined. Since $\hat\Sigma_{ZZ}\toP \Sigma_{ZZ}$ and $\hat\Sigma_{ZX}\toP \Sigma_{ZX}$ by the law of large numbers, it follows that $\hat S_{\mathrm{2SLS}}\toP S_{\mathrm{2SLS}}$, hence $\|\hat S_{\mathrm{2SLS}}\|=O_p(1)$. Moreover, Assumption~\ref{asspt:finite_moments} implies $E\|Z\|<\infty$, so
\[
\frac{1}{n}\sum_{j=1}^n \|Z_j-\hat\mu_Z\| = O_p(1).
\]
Combining this with \eqref{eq:empq_slope_bound} and $R_n=o_p(n^{-1/2})$ yields
\[
\sup_{u\in[a,b]}
\bigl\|
\bar{\tilde\beta}_1(u)-\tilde\beta_1(u)
\bigr\|
=
o_p(n^{-1/2}).
\]
Together with the intercept bound, this proves part~(ii):
\[
\|\bar{\tilde\beta}-\tilde\beta\|_{\ell^\infty([a,b])^{p+1}}
=
o_p(n^{-1/2}).
\]

\medskip
\noindent\emph{Step 2: IV-weighted quantile curves at a fixed $x$.}
For each fixed $x\in\R^p$,
\[
\bar\psi_x(u)-\hat\psi_x(u)
=
\bigl(\bar{\tilde\beta}_0(u)-\tilde\beta_0(u)\bigr)
+
\bigl(\bar{\tilde\beta}_1(u)-\tilde\beta_1(u)\bigr)^{\!\T}(x-\hat\mu_X).
\]
Hence
\begin{align*}
\|\bar\psi_x-\hat\psi_x\|_\infty
&\le
\sup_{u\in[a,b]}
\bigl|\bar{\tilde\beta}_0(u)-\tilde\beta_0(u)\bigr|
+
\|x-\hat\mu_X\|
\sup_{u\in[a,b]}
\bigl\|
\bar{\tilde\beta}_1(u)-\tilde\beta_1(u)
\bigr\| \\
&=
o_p(n^{-1/2}),
\end{align*}
because $\hat\mu_X\toP \mu_X$ and therefore $\|x-\hat\mu_X\|=O_p(1)$ for fixed $x$.

For the projected curves, Lemma A-6 (iii) in \citet{van2025regression} gives the contraction bound
\[
\|\Pi_{\mathcal Q}(\bar\psi_x)-\Pi_{\mathcal Q}(\hat\psi_x)\|_\infty
\le
\|\bar\psi_x-\hat\psi_x\|_\infty
=
o_p(n^{-1/2}).
\]
This proves part~(i).

\medskip
\noindent\emph{Step 3: projected coefficient functions.}
Let
\[
d_j(u)
:=
\Pi_{\mathcal Q}(\bar\psi_{X_j})(u)-\Pi_{\mathcal Q}(\hat\psi_{X_j})(u),
\qquad
\hat{\mathbf X}_j
:=
\begin{pmatrix}
1\\
X_j-\hat\mu_X
\end{pmatrix}.
\]
Then
\[
\bar\beta^{\est}(u)-\hat\beta^{\est}(u)
=
\left(\frac{\hat{\mathbf X}^{\!\T}\hat{\mathbf X}}{n}\right)^{-1}
\left(
\frac{1}{n}\sum_{j=1}^n \hat{\mathbf X}_j\, d_j(u)
\right).
\]
Taking norms and suprema,
\begin{equation}
\sup_{u\in[a,b]}
\|\bar\beta^{\est}(u)-\hat\beta^{\est}(u)\|
\le
\left\|
\left(\frac{\hat{\mathbf X}^{\!\T}\hat{\mathbf X}}{n}\right)^{-1}
\right\|
\cdot
\sup_{u\in[a,b]}
\left\|
\frac{1}{n}\sum_{j=1}^n \hat{\mathbf X}_j\, d_j(u)
\right\|.
\label{eq:empq_projcoef_1}
\end{equation}
By the contraction property of $\Pi_{\mathcal Q}$,
\[
|d_j(u)|
\le
\|\bar\psi_{X_j}-\hat\psi_{X_j}\|_\infty
\le
\sup_{v\in[a,b]}
\bigl|\bar{\tilde\beta}_0(v)-\tilde\beta_0(v)\bigr|
+
\|X_j-\hat\mu_X\|
\sup_{v\in[a,b]}
\bigl\|
\bar{\tilde\beta}_1(v)-\tilde\beta_1(v)
\bigr\|.
\]
Therefore,
\begin{align*}
\sup_{u\in[a,b]}
\left\|
\frac{1}{n}\sum_{j=1}^n \hat{\mathbf X}_j\, d_j(u)
\right\|
&\le
\left(
\frac{1}{n}\sum_{j=1}^n \|\hat{\mathbf X}_j\|
\right)
\sup_{v\in[a,b]}
\bigl|\bar{\tilde\beta}_0(v)-\tilde\beta_0(v)\bigr| \\
&\quad +
\left(
\frac{1}{n}\sum_{j=1}^n \|\hat{\mathbf X}_j\|\,\|X_j-\hat\mu_X\|
\right)
\sup_{v\in[a,b]}
\bigl\|
\bar{\tilde\beta}_1(v)-\tilde\beta_1(v)
\bigr\|.
\end{align*}
Assumption~\ref{asspt:finite_moments} implies $E\|X\|^2<\infty$, hence
\[
\frac{1}{n}\sum_{j=1}^n \|\hat{\mathbf X}_j\| = O_p(1),
\qquad
\frac{1}{n}\sum_{j=1}^n \|\hat{\mathbf X}_j\|\,\|X_j-\hat\mu_X\| = O_p(1).
\]
Further,
\[
\frac{\hat{\mathbf X}^{\!\T}\hat{\mathbf X}}{n}
=
\begin{pmatrix}
1 & 0^{\!\T}\\
0 & \hat\Sigma_{XX}
\end{pmatrix},
\]
and $\hat\Sigma_{XX}^{-1}=O_p(1)$ as shown in the proof of Theorem~\ref{thm:pw_ols_clt}.

Combining the preceding display with \eqref{eq:empq_projcoef_1} and part~(ii), we obtain
\[
\|\bar\beta^{\est}-\hat\beta^{\est}\|_{\ell^\infty([a,b])^{p+1}}
=
o_p(n^{-1/2}),
\]
which proves part~(iii).
\end{proof}

\begin{proof}[Proof of Corollary~\ref{cor:empirical_quantiles}]
Each feasible estimator decomposes as its infeasible counterpart plus a remainder that is $o_p(1)$ in the relevant $\ell^\infty$ norm after scaling by $\sqrt{n}$, by Proposition~\ref{prop:empirical_quantiles}. Specifically:
\begin{alignat*}{3}
&\text{(i)}&\quad \sqrt n\bigl(\bar\psi_x-\psi_x\bigr)
&=
\sqrt n\bigl(\hat\psi_x-\psi_x\bigr)
+
\underbrace{\sqrt n\bigl(\bar\psi_x-\hat\psi_x\bigr)}_{o_p(1)}, \\
&\text{(ii)}&\quad \sqrt n\bigl(\Pi_{\mathcal Q}(\bar\psi_x)-\psi_x\bigr)
&=
\sqrt n\bigl(\Pi_{\mathcal Q}(\hat\psi_x)-\psi_x\bigr)
+
\underbrace{\sqrt n\bigl(\Pi_{\mathcal Q}(\bar\psi_x)-\Pi_{\mathcal Q}(\hat\psi_x)\bigr)}_{o_p(1)}, \\
&\text{(iii)}&\quad \sqrt n\bigl(\bar{\tilde\beta}-\beta^{\mathrm{unc}}\bigr)
&=
\sqrt n\bigl(\tilde\beta-\beta^{\mathrm{unc}}\bigr)
+
\underbrace{\sqrt n\bigl(\bar{\tilde\beta}-\tilde\beta\bigr)}_{o_p(1)}, \\
&\text{(iv)}&\quad \sqrt n\bigl(\bar\beta^{\est}-\beta^{\est}\bigr)
&=
\sqrt n\bigl(\hat\beta^{\est}-\beta^{\est}\bigr)
+
\underbrace{\sqrt n\bigl(\bar\beta^{\est}-\hat\beta^{\est}\bigr)}_{o_p(1)}.
\end{alignat*}
In each line, the first term converges weakly to the corresponding Gaussian limit by Theorems~\ref{thm:convergence_unprojected}--\ref{thm:pw_ols_clt}. The result follows by Slutsky's lemma.
\end{proof}

\section{Derivation of the asymptotic covariance expression}
\label{sec:covariance_derivation}

This subsection derives the influence function of $\hat\psi_x(u)$ and the
corresponding asymptotic covariance kernel appearing in
Theorem~\ref{thm:convergence_unprojected}. Under correct specification, the
general expression simplifies to a closed form, which is useful for comparison
with \CLP.

\paragraph{Notation.}
To simplify these derivations, we first introduce some additional notation. Let $\widetilde Z := Z - \mu_Z$, $\widetilde X := X - \mu_X$, and fix the evaluation point $x\in\mathbb R^p$ with $\tilde x := x-\mu_X$.
Define the population moments
\[
M := \Sigma_{ZZ}^{-1},\qquad S := \Sigma_{ZX},\qquad A := S^{\T} M S.
\]
The IV weight used in our estimator can be written as,
\[
s(Z,x)\;=\; 1 + \tilde x^{\T} A^{-1} S^{\T} M\,\widetilde Z.
\]
Collect the nuisance parameters into $\theta=(\mu_X,\mu_Z, \Sigma_{ZX}, \Sigma_{ZZ})$, and write the target functional as
\[
\psi_x(u;\theta)\;=\; \E\!\big[\, s_\theta(Z,x)\,Q_Y(u)\,\big].
\]
We first derive the general asymptotic linear representation of $\hat\psi_x(u)$.
We then show that, under correct specification of the linear model, the
influence function simplifies to a closed form that coincides with the
covariance formula in \CLP.

The vector for the sample moments is
\[
m(W)\;=\;\Big(\;\widetilde X,\;\widetilde Z,\; \mathrm{vec}\big(\widetilde Z\widetilde X^{\T}-\Sigma_{ZX}\big) \;, \mathrm{vec}\big(\widetilde Z\widetilde Z^{\T}-\Sigma_{ZZ}\big)\;\Big).
\]

\paragraph{Influence Function.}

Recall that the unprojected estimator is, 
$\hat\psi_x(u)=n^{-1}\sum_{j=1}^n \hat s_j(x)\,Q_{Y_j}(u)$. A first-order expansion of $\hat s_j=g(\hat\theta,Z_j)$ about $\theta_0$ gives
\[
\sqrt{n}\big(\hat\psi_x(u)-\psi_x(u)\big)
= \frac{1}{\sqrt{n}}\sum_{j=1}^n\Big(s_j\,Q_{Y_j}(u)-\psi_x(u)\Big)
\;+\;\Big(\partial_\theta\psi_x(u)\Big)^{\T} \frac{1}{\sqrt{n}}\sum_{j=1}^n m(W_j)
\;+\;o_p(1),
\]
where
\[
\partial_\theta\psi_x(u) \;=\; \E\!\big[\,Q_Y(u)\,\nabla_\theta s(Z,x)\,\big]\in\R^{\dim(\theta)}.
\]

Thus the influence function of $\hat\psi_x(u)$ is
\[
\phi_x(W;u)
:=
\bigl(s(Z,x)Q_Y(u)-\psi_x(u)\bigr)
+
\bigl(\partial_\theta\psi_x(u)\bigr)^{\T}m(W).
\]

To simplify this expression, we now derive the gradient terms involved in the term $\partial_\theta\psi_x(u)$.

Remember the following standard formulas for matrix differentials,
\[
dM = -\,M\,(d\Sigma_{ZZ})\,M,\qquad
dA = (dS)^{\T}MS + S^{\T}M\,(dS) + S^{\T}(dM)S,\qquad
d(A^{-1}) = -\,A^{-1}(dA)\,A^{-1}.
\]
Write $s=1+\tilde x^{\T}A^{-1}S^{\T}M\,\widetilde Z$. Its total differential is
\begin{align*}
ds &= (d\tilde x)^{\T}A^{-1}S^{\T}M\widetilde Z
\;+\;\tilde x^{\T} d(A^{-1}) S^{\T}M\widetilde Z
\;+\;\tilde x^{\T} A^{-1} (dS)^{\T}M\widetilde Z
\;+\;\tilde x^{\T} A^{-1} S^{\T} (dM)\widetilde Z
\;+\;\tilde x^{\T} A^{-1} S^{\T} M (d\widetilde Z).
\end{align*}
Then, we derive expressions for the gradients for each block of $\theta$ by rewriting $ds=\mathrm{tr}\big((\nabla s)^{\T} d(\cdot)\big)$.

\begin{enumerate}
\item \textit{Gradient w.r.t.\ $\mu_X$.}
When varying $\mu_X$, we only have $d\tilde x = -\,d\mu_X$,
\[
ds = -\,(d\mu_X)^{\T} A^{-1}S^{\T}M\,\widetilde Z
\;=\; \mathrm{tr}\!\big(\,(-A^{-1}S^{\T}M\widetilde Z)\,(d\mu_X)^{\T}\big)
\;=\; \mathrm{tr}\!\big((\nabla_{\mu_X}s)^{\T} d\mu_X\big),
\]
hence
\[
\boxed{\;\nabla_{\mu_X} s \;=\; -\,A^{-1}S^{\T}M\,\widetilde Z.\;}
\]

\item \textit{Gradient w.r.t.\ $\mu_Z$.}
We only have $d\widetilde Z=-\,d\mu_Z$,
\[
ds \;=\; -\,\tilde x^{\T}A^{-1}S^{\T}M\,(d\mu_Z)
\;=\; \mathrm{tr}\!\big(\,(-M S A^{-1}\tilde x)\,(d\mu_Z)^{\T}\big)
\;=\; \mathrm{tr}\!\big((\nabla_{\mu_Z}s)^{\T} d\mu_Z\big),
\]
hence
\[
\boxed{\;\nabla_{\mu_Z} s \;=\; -\,M S A^{-1}\,\tilde x.\;}
\]

\item \textit{Gradient w.r.t.\ $S=\Sigma_{ZX}$.}
Collect the $dS$ and $(dS)^{\T}$ terms. Using $d(A^{-1})=-A^{-1}(dA)A^{-1}$ and $dA$ above,
\begin{align*}
ds
&= \tilde x^{\T} A^{-1} (dS)^{\T}M\widetilde Z
\;-\;\tilde x^{\T} A^{-1}\Big[(dS)^{\T}MS+S^{\T}M(dS)\Big]A^{-1} S^{\T}M\widetilde Z
\;+\; \text{terms in }dM,\,d\widetilde Z,\,d\tilde x.
\end{align*}
Drop the terms not involving $dS$. Rewrite each contribution as a trace against $dS$,
\begin{align*}
\tilde x^{\T} A^{-1} (dS)^{\T}M\widetilde Z
&= \mathrm{tr}\!\big(M\widetilde Z\,\tilde x^{\T}A^{-1}\,dS\big),\\
-\tilde x^{\T} A^{-1} (dS)^{\T}MS\,A^{-1} S^{\T}M\widetilde Z
&= -\,\mathrm{tr}\!\big(MS A^{-1} S^{\T}M\widetilde Z\,\tilde x^{\T}A^{-1}\,dS\big),\\
-\tilde x^{\T} A^{-1} S^{\T}M(dS)\,A^{-1} S^{\T}M\widetilde Z
&= -\,\mathrm{tr}\!\big((\tilde x^{\T}A^{-1} S^{\T}M\widetilde Z)\, A^{-1} S^{\T}M\,dS\big) \\
&= -\,\mathrm{tr}\!\big(M S A^{-1}\tilde x\,\widetilde Z^{\T} M S A^{-1}\,dS\big),
\end{align*}
where in the last equality we used cyclicity of trace and transpose identities.
Therefore
\[
ds
= \mathrm{tr}\!\Big(
\big[\,M\widetilde Z\,\tilde x^{\T}A^{-1}
\;-\;M S A^{-1} S^{\T}M\widetilde Z\,\tilde x^{\T}A^{-1}
\;-\;M S A^{-1}\tilde x\,\widetilde Z^{\T} M S A^{-1}\,\big]^{\T} dS\Big),
\]
so
\[
\boxed{\;\nabla_{S} s
= M\widetilde Z\,\tilde x^{\T}A^{-1}
- M S A^{-1}\Big(S^{\T}M\widetilde Z\,\tilde x^{\T}A^{-1}+\tilde x\,\widetilde Z^{\T} M S A^{-1}\Big).\;}
\]

 \item \textit{Gradient w.r.t.\ $\Sigma_{ZZ}$}.
Here $dM=-M(d\Sigma_{ZZ})M$ enters both the $S^{\T}(dM)\widetilde Z$ and the $dA$ term inside $d(A^{-1})$. After some algebra (collect $d\Sigma_{ZZ}$ and use $d(A^{-1})$), we get,
\begin{align*}
ds
&= \tilde x^{\T}A^{-1}S^{\T}(dM)\widetilde Z
\;-\;\tilde x^{\T}A^{-1}S^{\T}(dM)S\,A^{-1} S^{\T}M\widetilde Z\\
&= -\,\mathrm{tr}\!\Big( \big[M S A^{-1} S^{\T}M\widetilde Z - M\widetilde Z\big]\,(M S A^{-1}\tilde x)^{\T}\, d\Sigma_{ZZ}\Big),
\end{align*}
hence
\[
\boxed{\;\nabla_{\Sigma_{ZZ}} s
= \big(M S A^{-1} S^{\T}M\widetilde Z - M\widetilde Z\big)\,(M S A^{-1}\tilde x)^{\T}.\;}
\]
\end{enumerate}

\subsubsection{Simplification under correct specification}

Assume now that the linear model is correctly specified,
\[
Q_Y(u) \;=\; \beta_0(u) + \beta_1(u)^{\T}\widetilde X + \eta(u),
\qquad \E[\eta(u)]=0,\quad \E[\widetilde Z\,\eta(u)]=0,
\]
and denote $\beta_0=\beta_0(u)$, $\beta_1=\beta_1(u)$ for brevity.

Then, using $\E[\widetilde Z]=0$, $\E[\eta]=0$, and $\E[\widetilde Z\,Q_Y]=S\beta_1$, we can derive the corresponding expressions for the gradients of $\psi_x$,

\begin{enumerate}
\item \[
\partial_{\mu_X}\psi_x
= \E\big[Q_Y(-A^{-1}S^{\T}M\widetilde Z)\big]
= -A^{-1}S^{\T}M\,(S\beta_1) = -\,\beta_1.
\]

\item \[
\partial_{\mu_Z}\psi_x
= \E\big[Q_Y(-M S A^{-1}\tilde x)\big]
= -\,\beta_0\,M S A^{-1}\tilde x.
\]

\item \begin{align*}
\partial_{S}\psi_x
&= \E\Big[Q_Y\Big\{M\widetilde Z\,\tilde x^{\T}A^{-1}
- M S A^{-1}\big(S^{\T}M\widetilde Z\,\tilde x^{\T}A^{-1}+\tilde x\,\widetilde Z^{\T} M S A^{-1}\big)\Big\}\Big]\\
&= M(S\beta_1)\,\tilde x^{\T}A^{-1}
- M S A^{-1}\Big((S^{\T}MS)\beta_1\,\tilde x^{\T}A^{-1}+\tilde x\,\beta_1^{\T}S^{\T}MS A^{-1}\Big)\\
&= -\,M S A^{-1}\tilde x\,\beta_1^{\T}.
\end{align*}

\item \begin{align*}
\partial_{\Sigma_{ZZ}}\psi_x
& = \E\!\big[Q_Y\big(M S A^{-1} S^{\T}M\widetilde Z - M\widetilde Z\big)\big]\,(M S A^{-1}\tilde x)^{\T} \\
& = \big(M S A^{-1}S^{\T}M(S\beta_1)-M(S\beta_1)\big)(M S A^{-1}\tilde x)^{\T}
= 0.
\end{align*}
\end{enumerate}

\paragraph{Cancellation against estimation error term $sQ_Y-\psi_x$.}
Then
\[
sQ_Y-\psi_x
= \beta_0\,(s-1) \;+\; \beta_1^{\T}\big(s\,\widetilde X-\tilde x\big)\;+\;s\,\eta,
\]
and note $s-1=\tilde x^{\T}A^{-1}S^{\T}M\widetilde Z$, and $s\,\widetilde X=\widetilde X+(\tilde x^{\T}A^{-1}S^{\T}M\widetilde Z)\,\widetilde X$. Thus
\[
sQ_Y-\psi_x
= \underbrace{\beta_0\,\tilde x^{\T}A^{-1}S^{\T}M\widetilde Z}_{(\mathrm{A})}
\;+\;\underbrace{\beta_1^{\T}(\widetilde X-\tilde x)}_{(\mathrm{B})}
\;+\;\underbrace{\beta_1^{\T}\widetilde X\;\tilde x^{\T}A^{-1}S^{\T}M\widetilde Z}_{(\mathrm{C})}
\;+\;\underbrace{s\,\eta}_{(\mathrm{D})}.
\]

Now multiply the derivative blocks by the corresponding components of $m(W)$:
\begin{enumerate}
\item $(\partial_{\mu_Z}\psi_x)^{\T}\widetilde Z = -\,\beta_0\,\tilde x^{\T}A^{-1}S^{\T}M\widetilde Z$ cancels $(\mathrm{A})$.
\item $(\partial_{\mu_X}\psi_x)^{\T}\widetilde X = -\,\beta_1^{\T}\widetilde X$ cancels the $+\,\beta_1^{\T}\widetilde X$ part of $(\mathrm{B})$.
\item Frobenius inner product with the $S$-score:
\[
\big\langle \partial_{S}\psi_x,\; \widetilde Z\widetilde X^{\T}-S\big\rangle_F
= \mathrm{tr}\!\Big(\big[-M S A^{-1}\tilde x\,\beta_1^{\T}\big]^{\T}(\widetilde Z\widetilde X^{\T}-S)\Big)
= -\,\beta_1^{\T}\widetilde X\;\tilde x^{\T}A^{-1}S^{\T}M\widetilde Z
\;+\;\beta_1^{\T}\tilde x,
\]
which cancels $(\mathrm{C})$ and the remaining $-\,\beta_1^{\T}\tilde x$ part of $(\mathrm{B})$.
\item The $\Sigma_{ZZ}$-block is $0$ and contributes nothing.
\end{enumerate}
Therefore \emph{all deterministic} $\beta$-terms cancel, and only the noise $\eta(u)$ remains,
\[
\phi_x(W; u)
\;=\; \big(sQ_Y-\psi_x\big) + \big(\partial_\theta\psi_x\big)^{\T} m(W)
\;=\; s(Z,x)\,\eta(u).
\]
\paragraph{Covariance kernel.}
As a result of the above derivations, the asymptotic covariance kernel of $\hat\psi_x(\cdot)$ is
\[
\Gamma_x(u,u') \;=\; \E\!\big[\,s(Z,x)^2\,\eta(u)\,\eta(u')\,\big].
\]
To keep intercept handling transparent, in this paragraph let $X$ and $Z$ denote the \emph{centered} regressors and instruments, and introduce the augmented stacks
\[
\mathbf{X}:=\begin{bmatrix}1\\ X\end{bmatrix},\qquad
\mathbf{Z}:=\begin{bmatrix}1\\ Z\end{bmatrix},\qquad
\mathbf{x}:=\begin{bmatrix}1\\ \tilde{x}\end{bmatrix}.
\]
Define the (population) 2SLS operator
\[
S_{\mathrm{2SLS}}
\;:=\;
\big(\Sigma_{\mathbf{X}\mathbf{Z}}\;\Sigma_{\mathbf{Z}\mathbf{Z}}^{-1}\;\Sigma_{\mathbf{Z}\mathbf{X}}\big)^{-1}\;
\Sigma_{\mathbf{X}\mathbf{Z}}\;\Sigma_{\mathbf{Z}\mathbf{Z}}^{-1},
\qquad
\Sigma_{\mathbf{X}\mathbf{Z}}:=\E[\mathbf{X}\mathbf{Z}^\top],\ \ 
\Sigma_{\mathbf{Z}\mathbf{Z}}:=\E[\mathbf{Z}\mathbf{Z}^\top]\ (\succ0).
\]
By the Riesz representer identity,
\[
s(Z,x)\;=\;\mathbf{x}^{\T}S_{\mathrm{2SLS}}\,\mathbf{Z}.
\]
Let $a:=S_{\mathrm{2SLS}}^{\T}\mathbf{x}$. Then $s(Z,x)=a^{\T}\mathbf{Z}$ and
\[
\Gamma_x(u,u')
\;=\; \E\!\big[(a^{\T}\mathbf{Z})^2\,\eta(u)\eta(u')\big]
\;=\; a^{\T}\,\underbrace{\E\!\big[\mathbf{Z}\mathbf{Z}^{\T}\eta(u)\eta(u')\big]}_{=:J(u,u')}\,a
\;=\; \mathbf{x}^{\T} S_{\mathrm{2SLS}}\,J(u,u')\,S_{\mathrm{2SLS}}^{\T} \mathbf{x},
\]
which is exactly the $S\,J(u,u')\,S^{\T}$ covariance in \CLP, pushed through $\mathbf{x}$. In the just-identified case
($\dim\mathbf{Z}=\dim\mathbf{X}$ and $\Sigma_{\mathbf{Z}\mathbf{X}}$ invertible),
\[
S_{\mathrm{2SLS}}=(\Sigma_{\mathbf{Z}\mathbf{X}})^{-1}
\quad\Longrightarrow\quad
\Gamma_x(u,u') \;=\; \mathbf{x}^{\T}(\Sigma_{\mathbf{Z}\mathbf{X}})^{-1}\,J(u,u')\,(\Sigma_{\mathbf{Z}\mathbf{X}}^{\T})^{-1}\mathbf{x}.
\]
\paragraph{Misspecification.}
Under misspecification, let $\beta^{\mathrm{unc}}(u)$ denote the population 2SLS
coefficient functions and recall the definition of the pseudo-residual
\[
\xi(u):=Q_Y(u)-\mathbf X^{\T}\beta^{\mathrm{unc}}(u).
\]
Then $\hat\psi_x(u)$ continues to admit the linear representation
\[
\sqrt n\bigl(\hat\psi_x(u)-\psi_x(u)\bigr)
=
\frac{1}{\sqrt n}\sum_{j=1}^n \phi_x(W_j;u)+o_p(1),
\]
with
\[
\phi_x(W;u)
=
\bigl(s(Z,x)Q_Y(u)-\psi_x(u)\bigr)
+
\bigl(\partial_\theta \psi_x(u)\bigr)^{\T}m(W).
\]
In the overidentified case, the derivative correction term does not in general
collapse to $s(Z,x)\xi(u)$. Accordingly, the asymptotic covariance kernel is
\[
\Gamma_x(u,u')
=
E\!\left[\phi_x(W;u)\phi_x(W;u')\right].
\]
When the model is correctly specified, this expression reduces to
$E[s(Z,x)^2\eta(u)\eta(u')]$ as shown above.

\end{document}